\documentclass[USenglish,11pt]{mylipics}

\newif\ifslow
\slowtrue

\usepackage{nth}
\usepackage{latexsym}
\usepackage{amsmath}
\usepackage{amssymb}
\usepackage{microtype}
\usepackage{datetime}
\usepackage{wrapfig}
\usepackage{fourier}
\usepackage{caption}

\usepackage{ifpdf}

\ifpdf
\usepackage[all,arc,curve,color,frame]{xy}
\else
\usepackage[all,dvips,arc,curve,color,frame]{xy}
\fi
\usepackage{color}

\usepackage{algorithmic,algorithm}

\usepackage{hyperref}

\usepackage{pgf}
\usepackage{tikz}
\usetikzlibrary{arrows,shapes,snakes,backgrounds}
\advance \textheight by .2in

\DeclareMathOperator{\topp}{top}
\DeclareMathOperator{\best}{best}
\providecommand{\pref}{\pi}
\providecommand{\nessun}{\bot} % another possibility is \Omega
\providecommand{\interval}{J}
\providecommand{\MM}{\mathcal{M}}
\providecommand{\WW}{\mathcal{W}}

\DeclareMathOperator{\Bit}{bit}
\DeclareMathOperator{\Par}{parity}
\DeclareMathOperator{\dom}{dom}

\begin{document}

% Define theorem environment here
%%%%%%%%%%%%%%%%%%%%%%%

%%%%%%%%%%%%%%%%%%
%\newenvironment{myclaim}
%{
%	\begin{flushleft}
%	\textsf{\textbf{Claim:}}
%}
%{
%	\end{flushleft}
%
%}

%%%%%%%%%%%%%%%%%%%%%%%

%\ifslow
%\usepackage[all]{xy}
%\usepackage{amsmath}
\newcommand{\bra}[1]{\left\langle{#1}\right\vert}
\newcommand{\bit}[1]{{#1}}
    % Defines Dirac notation.
\newcommand{\qw}[1][-1]{\ar @{-} [0,#1]}
    % Defines a wire that connects horizontally.  By default it connects to the object on the left of the current object.
    % WARNING: Wire commands must appear after the gate in any given entry.
\newcommand{\qwx}[1][-1]{\ar @{-} [#1,0]}
    % Defines a wire that connects vertically.  By default it connects to the object above the current object.
    % WARNING: Wire commands must appear after the gate in any given entry.
\newcommand{\cw}[1][-1]{\ar @{=} [0,#1]}
    % Defines a classical wire that connects horizontally.  By default it connects to the object on the left of the current object.
    % WARNING: Wire commands must appear after the gate in any given entry.
\newcommand{\cwx}[1][-1]{\ar @{=} [#1,0]}
    % Defines a classical wire that connects vertically.  By default it connects to the object above the current object.
    % WARNING: Wire commands must appear after the gate in any given entry.
\newcommand{\gate}[1]{*{\xy *+<.6em>{#1};p\save+LU;+RU **\dir{-}\restore\save+RU;+RD **\dir{-}\restore\save+RD;+LD **\dir{-}\restore\POS+LD;+LU **\dir{-}\endxy} \qw}
    % Boxes the argument, making a gate.
\newcommand{\meter}{\gate{\xy *!<0em,1.1em>h\cir<1.1em>{ur_dr},!U-<0em,.4em>;p+<.5em,.9em> **h\dir{-} \POS <-.6em,.4em> *{},<.6em,-.4em> *{} \endxy}}
    % Inserts a measurement meter.
\newcommand{\measure}[1]{*+[F-:<.9em>]{#1} \qw}
    % Inserts a measurement bubble with user defined text.
\newcommand{\measuretab}[1]{*{\xy *+<.6em>{#1};p\save+LU;+RU **\dir{-}\restore\save+RU;+RD **\dir{-}\restore\save+RD;+LD **\dir{-}\restore\save+LD;+LC-<.5em,0em> **\dir{-} \restore\POS+LU;+LC-<.5em,0em> **\dir{-} \endxy} \qw}
    % Inserts a measurement tab with user defined text.
\newcommand{\measureD}[1]{*{\xy*+=+<.5em>{\vphantom{\rule{0em}{.1em}#1}}*\cir{r_l};p\save*!R{#1} \restore\save+UC;+UC-<.5em,0em>*!R{\hphantom{#1}}+L **\dir{-} \restore\save+DC;+DC-<.5em,0em>*!R{\hphantom{#1}}+L **\dir{-} \restore\POS+UC-<.5em,0em>*!R{\hphantom{#1}}+L;+DC-<.5em,0em>*!R{\hphantom{#1}}+L **\dir{-} \endxy} \qw}
    % Inserts a D-shaped measurement gate with user defined text.
\newcommand{\multimeasure}[2]{*+<1em,.9em>{\hphantom{#2}} \qw \POS[0,0].[#1,0];p !C *{#2},p \drop\frm<.9em>{-}}
    % Draws a multiple qubit measurement bubble starting at the current position and spanning #1 additional gates below.
    % #2 gives the label for the gate.
    % You must use an argument of the same width as #2 in \ghost for the wires to connect properly on the lower lines.
\newcommand{\multimeasureD}[2]{*+<1em,.9em>{\hphantom{#2}}\save[0,0].[#1,0];p\save !C *{#2},p+LU+<0em,0em>;+RU+<-.8em,0em> **\dir{-}\restore\save +LD;+LU **\dir{-}\restore\save +LD;+RD-<.8em,0em> **\dir{-} \restore\save +RD+<0em,.8em>;+RU-<0em,.8em> **\dir{-} \restore \POS !UR*!UR{\cir<.9em>{r_d}};!DR*!DR{\cir<.9em>{d_l}}\restore \qw}
    % Draws a multiple qubit D-shaped measurement gate starting at the current position and spanning #1 additional gates below.
    % #2 gives the label for the gate.
    % You must use an argument of the same width as #2 in \ghost for the wires to connect properly on the lower lines.
\newcommand{\control}{*!<0em,.025em>-=-{\bullet}}
\newcommand{\controld}{*!<0em,.025em>-=-{\blacktriangledown}}
\newcommand{\controlu}{*!<0em,.025em>-=-{\blacktriangle}}
    % Inserts an unconnected control.
\newcommand{\controlo}{*-<.21em,.21em>{\xy *=<.59em>!<0em,-.02em>[o][F]{}\POS!C\endxy}}
    % Inserts a unconnected control-on-0.
\newcommand{\ctrl}[1]{\control \qwx[#1] \qw}
\newcommand{\ctrld}[1]{\controld \qwx[#1] \qw}
\newcommand{\ctrlu}[1]{\controlu \qwx[#1] \qw}
    % Inserts a control and connects it to the object #1 wires below.
\newcommand{\ctrlo}[1]{\controlo \qwx[#1] \qw}
    % Inserts a control-on-0 and connects it to the object #1 wires below.
\newcommand{\targ}{*!<0em,.019em>=<.79em,.68em>{\xy {<0em,0em>*{} \ar @{ - } +<.4em,0em> \ar @{ - } -<.4em,0em> \ar @{ - } +<0em,.36em> \ar @{ - } -<0em,.36em>},<0em,-.019em>*+<.8em>\frm{o}\endxy} \qw}
    % Inserts a CNOT target.
\newcommand{\qswap}{*=<0em>{\times} \qw}
    % Inserts half a swap gate. 
    % Must be connected to the other swap with \qwx.
\newcommand{\multigate}[2]{*+<1em,.9em>{\hphantom{#2}} \qw \POS[0,0].[#1,0];p !C *{#2},p \save+LU;+RU **\dir{-}\restore\save+RU;+RD **\dir{-}\restore\save+RD;+LD **\dir{-}\restore\save+LD;+LU **\dir{-}\restore}
    % Draws a multiple qubit gate starting at the current position and spanning #1 additional gates below.
    % #2 gives the label for the gate.
    % You must use an argument of the same width as #2 in \ghost for the wires to connect properly on the lower lines.
\newcommand{\ghost}[1]{*+<1em,.9em>{\hphantom{#1}} \qw}
    % Leaves space for \multigate on wires other than the one on which \multigate appears.  Without this command wires will cross your gate.
    % #1 should match the second argument in the corresponding \multigate. 
\newcommand{\push}[1]{*{#1}}
    % Inserts #1, overriding the default that causes entries to have zero size.  This command takes the place of a gate.
    % Like a gate, it must precede any wire commands.
    % \push is useful for forcing columns apart.
    % NOTE: It might be useful to know that a gate is about 1.3 times the height of its contents.  I.e. \gate{M} is 1.3em tall.
    % WARNING: \push must appear before any wire commands and may not appear in an entry with a gate or label.
\newcommand{\gategroup}[6]{\POS"#1,#2"."#3,#2"."#1,#4"."#3,#4"!C*+<#5>\frm{#6}}
    % Constructs a box or bracket enclosing the square block spanning rows #1-#3 and columns=#2-#4.
    % The block is given a margin #5/2, so #5 should be a valid length.
    % #6 can take the following arguments -- or . or _\} or ^\} or \{ or \} or _) or ^) or ( or ) where the first two options yield dashed and
    % dotted boxes respectively, and the last eight options yield bottom, top, left, and right braces of the curly or normal variety.
    % \gategroup can appear at the end of any gate entry, but it's good form to pick one of the corner gates.
    % BUG: \gategroup uses the four corner gates to determine the size of the bounding box.  Other gates may stick out of that box.  See \prop. 
\newcommand{\rstick}[1]{*!L!<-.5em,0em>=<0em>{#1}}
    % Centers the left side of #1 in the cell.  Intended for lining up wire labels.  Note that non-gates have default size zero.
\newcommand{\lstick}[1]{*!R!<.5em,0em>=<0em>{#1}}
    % Centers the right side of #1 in the cell.  Intended for lining up wire labels.  Note that non-gates have default size zero.
\newcommand{\ustick}[1]{*!D!<0em,-.5em>=<0em>{#1}}
    % Centers the bottom of #1 in the cell.  Intended for lining up wire labels.  Note that non-gates have default size zero.
\newcommand{\dstick}[1]{*!U!<0em,.5em>=<0em>{#1}}
    % Centers the top of #1 in the cell.  Intended for lining up wire labels.  Note that non-gates have default size zero.
\newcommand{\Qcircuit}[1][0em]{\xymatrix @*[o] @*=<#1>}
    % Defines \Qcircuit as an \xymatrix with entries of default size 0em.  The optional argument, #1, is for use with clusters, and allows you
    % to fix the size of the nodes.  I would not advise using it with normal circuits.
\newcommand{\node}[2][]{{\begin{array}{c} \ _{#1}\  \\ {#2} \\ \ \end{array}}\drop\frm{o} }
    % When Qcircuit has been passed the optional argument for cluster states, this command produces a round node of the size specified in that
    % argument.  The optional argument #2 specifies the contents of a node, while optional argument #1 is a secondary label.  
\newcommand{\link}[2]{\ar @{-} [#1,#2]}
    % Draws a wire or connecting line to the element #1 rows down and #2 columns forward.
\newcommand{\pureghost}[1]{*+<1em,.9em>{\hphantom{#1}}}
    % Same as \ghost except it omits the wire leading to the left. 
   \renewcommand{\Qcircuit}[1][0em]{\xymatrix @*=<#1>}
%\fi   
\newcommand{\defref}[1]{Definition~\ref{def:#1}}
\newcommand{\theoref}[1]{Theorem~\ref{theo:#1}}
\newcommand{\propref}[1]{Proposition~\ref{prop:#1}}
\newcommand{\corref}[1]{Corollary~\ref{cor:#1}}
\newcommand{\lemref}[1]{Lemma~\ref{lem:#1}}
\newcommand{\exref}[1]{Example~\ref{ex:#1}}
\newcommand{\reref}[1]{Remark~\ref{re:#1}}
\newcommand{\eref}[1]{(\ref{eq:#1})}
\newcommand{\aref}[1]{Appendix~\ref{a:#1}}
\newcommand{\secref}[1]{Section~\ref{sec:#1}}
\newcommand{\figref}[1]{Fig.~\ref{fig:#1}}
\newcommand{\algoref}[1]{Algorithm~\ref{algo:#1}}

\newcommand{\proves}{\vdash}
\newcommand{\HRule}{\rule{\linewidth}{0.5mm}}
\newcommand{\calf}[1]{\mathcal{#1}}
\newcommand{\set}[1]{\{#1\}}
\newcommand{\bset}[1]{\bigl\{#1\bigr\}}
\newcommand{\It}[1]{\mathsf{#1}}
\newcommand{\MIt}[1]{\mathit{#1}}
\newcommand{\seq}[1]{\langle #1\rangle}
\newcommand{\colq}[1]{\begin{array}{c}#1\end{array}}
\newcommand*\circled[1]{\tikz[baseline=(char.base)]{
            \node[shape=circle,draw,inner sep=2pt] (char) {#1};}}

\newcommand{\fxpset}{\Phi_{\mathsf{fxp}}}
\newcommand{\smset}{\Phi_{\mathsf{sm}}}

\newcommand{\lfmm}{\delta_{\It{LFMM}}}
\newcommand{\ccv}{\delta_{\It{CCV}}}
\newcommand{\dnum}{\delta_{\It{NUM}}}
\newcommand{\conn}{\delta_{\It{CONN}}}
\newcommand{\lfmp}{\delta_{\It{LFMP}}}
\newcommand{\mfv}{\delta_{\It{MFV}}}

\makeatletter
\newcommand{\rmnum}[1]{\romannumeral #1}
\newcommand{\Rmnum}[1]{\expandafter\@slowromancap\romannumeral #1@}
\makeatother

\newcommand{\mred}{\le_{\mathsf{m}}}

\newcommand{\todo}[1]{ \textcolor{red}{TODO: #1}}
\newcommand{\co}[1]{ \textcolor{red}{#1}}

\newcommand{\M}{\calf{M}}
\newcommand{\Pair}{\mathsf{Pair}}
 \newcommand{\lra}{\leftrightarrow}
\newcommand{\ra}{\rightarrow}

%%%%%%% Complexity classes and theories
\newcommand{\VPV}{\textit{VPV}}
\newcommand{\PV}{\textit{PV}}

\newcommand{\NP}{\mathsf{NP}}
\newcommand{\FP}{\mathsf{FP}}
\newcommand{\PH}{\mathsf{PH}}
\newcommand{\C}{\mathsf{C}}
\newcommand{\CC}{\mathsf{CC}}
\newcommand{\SucCC}{\mathsf{SucCC}}
\newcommand{\PSPACE}{\mathsf{PSPACE}}
\newcommand{\EXPTIME}{\mathsf{EXPTIME}}
\newcommand{\CCnot}{\CC\lnot}
\newcommand{\mydash}{\text{-}}
\newcommand{\CCACMO}{\AC^0\mydash\CCV}
\newcommand{\CCNLMO}{\NL\mydash\CCV}
\newcommand{\CCACT}{(\AC^0)^{\CCV}}
\newcommand{\APC}{\mathsf{APC}}
\newcommand{\CCall}{\mathsf{CC}_{\mathsf{all}}}
\newcommand{\CCSubr}{\mathsf{CC}^{\mathsf{Subr}}}
\newcommand{\CCstar}{\mathsf{CC}^\ast}
\newcommand{\Class}{\mathsf{C}}
\newcommand{\FCC}{\mathsf{FCC}}
\newcommand{\FCCSubr}{\mathsf{FCC}^{\mathsf{Subr}}}
\newcommand{\FCCstar}{\mathsf{FCC}^\ast}
\newcommand{\FC}{\mathsf{FC}}
\newcommand{\AC}{\mathsf{AC}}
\newcommand{\FAC}{\mathsf{FAC}}
\newcommand{\DET}{\mathsf{DET}}
\newcommand{\RDET}{\mathsf{RDET}}
\newcommand{\NC}{\mathsf{NC}}
\newcommand{\FO}{\mathsf{FO}}
\newcommand{\FNC}{\mathsf{FNC}}
\newcommand{\FNL}{\mathsf{FNL}}
\newcommand{\FRNC}{\mathsf{FRNC}}
\newcommand{\NL}{\mathsf{NL}}
\renewcommand{\L}{\mathsf{L}}
\newcommand{\RsL}{\mathsf{R}\#\mathsf{L}}
\newcommand{\TC}{\mathsf{TC}}
\newcommand{\SL}{\mathsf{SL}}
\newcommand{\ZPLP}{\mathsf{ZPLP}}
\newcommand{\ZPL}{\mathsf{ZPL}}
\newcommand{\RP}{\mathsf{RP}}
\newcommand{\RL}{\mathsf{RL}}
\newcommand{\coRP}{\mathsf{coRP}}
\newcommand{\coRL}{\mathsf{coRL}}
\newcommand{\BPL}{\mathsf{BPL}}
\newcommand{\RNC}{\mathsf{RNC}}
\newcommand{\coRNC}{\mathsf{coRNC}}
\newcommand{\RAC}{\mathsf{RAC}}
\renewcommand{\P}{\mathsf{P}}
\newcommand{\Ppoly}{\mathsf{P}/\mathsf{poly}}
% Add the \SC command for the file to compile
\newcommand{\SC}{\mathsf{SC}}
\newcommand{\LogCFL}{\mathsf{LogCFL}}

\newcommand{\VCC}{\;\mathsf{VCC}}
\newcommand{\VAC}{\mathsf{VAC}}
\newcommand{\VNC}{\mathsf{VNC}}
\newcommand{\VNL}{\mathsf{VNL}}
\newcommand{\VL}{\mathsf{VL}}
\newcommand{\VP}{\mathsf{VP}}
\newcommand{\VTC}{\mathsf{VTC}}
\newcommand{\VSL}{\mathsf{VSL}}
\newcommand{\VsL}{\mathsf{V\#L}}
\newcommand{\LFP}{\calf{L}_{\FP}}
\newcommand{\VC}{\mathsf{VC}}
\newcommand{\V}{\mathsf{V}}
\newcommand{\CH}{\it{CH}}

\newcommand{\LMAX}{\mathit{LMAX}}
\newcommand{\LMIN}{\mathit{LMIN}}
\newcommand{\LIND}{\mathit{LIND}}
\newcommand{\PIND}{\mathit{PIND}}
\newcommand{\IND}{\textit{IND}}
\newcommand{\PHP}{\mathsf{PHP}}
\newcommand{\NUMO}{\textit{NUMONES}}
\newcommand{\MFV}{\textit{MFV}}
\newcommand{\MCV}{\textit{MCV}}
\newcommand{\CONN}{\textit{CONN}}
\newcommand{\CCVA}{\textit{AXccv}}
\newcommand{\CCVB}{\It{CCV}}

\newcommand{\SMP}{ \textmd{\textsc{Sm}}}
\newcommand{\MOSM}{ \textmd{\textsc{MoSm}}}
\newcommand{\WOSM}{ \textmd{\textsc{WoSm}}}
\newcommand{\BSVP}{ \textmd{\textsc{Bsvp}}}
\newcommand{\LFMM}{ \textmd{\textsc{Lfmm}}}
\newcommand{\TLFMM}{ \textmd{\textsc{3Lfmm}}}
\newcommand{\CCV}{ \textmd{\textsc{Ccv}}}
\newcommand{\MCVP}{ \textmd{\textsc{Mcvp}}}
\newcommand{\XNS}{ \textmd{\textsc{Xns}}}
\newcommand{\TCV}{ \textmd{\textsc{Three-valued Ccv}}}

\newcommand{\CCVN}{ \textmd{\textsc{Ccv}$\neg$}}
\newcommand{\VLFMM}{ \textmd{\textsc{vLfmm}}}
\newcommand{\TVLFMM}{ \textmd{\textsc{3vLfmm}}}
\newcommand{\UNIV}{ \mathsf{UNIV}}
\newcommand{\INPUT}{ \mathsf{INPUT}}
\newcommand{\IN}{ \mathsf{IN}}
\newcommand{\INP}{ \mathsf{INP}}
\newcommand{\CIRCUIT}{\mathsf{CIRCUIT}}
\newcommand{\CIR}{\mathsf{CIR}}

\newcommand{\NN}{\mathbb{N}}
\newcommand{\ZZ}{\mathbb{Z}}
\newcommand{\EE}{\mathbb{E}}
%\newcommand{\dom}{\mathsf{dom}}
%%%%%%%%%%%%%%%%%%%%%%%%%
%More Yuval definitions
\providecommand{\bbit}[2]{\Bit\left(#1,#2\right)}
\newcommand{\bina}[1]{\llbracket #1 \rrbracket}
\providecommand{\vc}[1]{\uppercase{#1}}
\providecommand{\vczero}{\mathbf{0}}
\providecommand{\rep}[2]{#1^{(#2 \text{ times})}}
\providecommand{\len}[1]{|#1|}
\providecommand{\ham}[1]{\|#1\|}
\providecommand{\parity}[1]{\Par(#1)}
\providecommand{\ZZ}{\mathbb{Z}}
\providecommand{\GG}{\mathbb{G}}
\providecommand{\mT}{T_{\max}}
\providecommand{\conc}{\#}

\numberwithin{equation}{section}

\title{The Complexity of the Comparator Circuit Value Problem}
\author[1]{Stephen A. Cook}
\author[1]{Yuval Filmus}
\author[1]{Dai Tri Man L\^e}
\authorrunning{S.A. Cook, Y. Filmus, and D.T.M. L\^e}
\affil[1]{Department of Computer Science, University of Toronto\\ \{sacook,yuvalf,ledt\}@cs.toronto.edu}

%\Copyright[nc-nd]{}
%\subjclass{ F.2.2 [Analysis of Algorithms and Problem Complexity]: Nonnumerical Algorithms and Problems; F.4.1 [Theory of Computation]: Mathematical logic}% mandatory: Please choose ACM 1998 classifications from http://www.acm.org/about/class/ccs98-html . E.g., cite as "F.1.1 Models of Computation". 
%\keywords{bounded arithmetic,
%complexity theory,
%stable marriage problem,
%comparator circuits}% mandatory: Please provide 1-5 keywords
%%%%%%%%%%%%%%%%%%%%%%%%%%%%%%%%%%%%%%%%%%%%%%%%%%%%%%%%%

%\setcounter{page}{0}
%\setcounter{tocdepth}{2}

\maketitle
\begin{center}
\fbox{Date: \textit{\today}}  %at \textit{\currenttime}.}
\end{center}

\begin{abstract}
In 1990 Subramanian defined the complexity class $\CC$ as the set of
problems log-space reducible to the comparator circuit value problem
($\CCV$).  He and Mayr showed that $\NL \subseteq \CC \subseteq \P$, and
proved that in addition to $\CCV$ several other problems are complete for
$\CC$, including the stable marriage problem, and finding the
lexicographically first maximal matching in a bipartite graph.
Although the class has not received much attention since then, we are
interested in $\CC$ because we conjecture that it is incomparable with the
parallel class $\NC$ which also satisfies $\NL \subseteq \NC \subseteq \P$,
and note that this conjecture implies that none of the $\CC$-complete
problems has an efficient polylog time parallel algorithm.  We
provide evidence for our conjecture by giving oracle settings
in which relativized $\CC$ and relativized $\NC$ are incomparable.

We give several alternative definitions of $\CC$, including (among others) the class of problems computed by uniform polynomial-size families of comparator circuits supplied with copies of the input and its negation, the class of problems $\AC^0$-reducible to $\CCV$, and the class of problems computed by uniform $\AC^0$ circuits with $\CCV$ gates. We also give a machine model for $\CC$, which corresponds to its characterization as log-space uniform polynomial-size families of comparator circuits. These various characterizations show that $\CC$ is a robust class. Our techniques also show that the corresponding function class $\FCC$ is closed under composition. The main technical tool we employ is universal comparator circuits. 

% We also show that $\CC$ has robust closure properties by introducing
% $\AC^0$-uniform universal comparator circuits, and use them to show
% that $\CC$ can be characterized as the set of problems $\AC^0$ many-one
% reducible (as opposed to log-space reducible) to the comparator circuit
% value problem, and equals the set of problems
% computable by uniform polynomial-size families of comparator circuits
% supplied with polynomially many copies of the input and its negation.
% We also show that $\CC$ is closed under $\AC^0$ `Turing' reductions
% (i.e.  reductions given by a uniform family of $\AC^0$ circuits with
% oracle gates making queries to other $\CC$ problems), and that the
% corresponding function class $\FCC$ is closed under composition.

Other results include a simpler proof of $\NL \subseteq \CC$, a more
careful analysis showing the lexicographically first maximal matching
problem and its variants are $\CC$-complete under $\AC^0$ many-one
reductions, and an explanation of the relation between the Gale-Shapley
algorithm and Subramanian's algorithm for stable marriage.

This paper continues the previous work of Cook, L\^e and Ye which focused
on Cook-Nguyen style uniform proof complexity, answering several open
questions raised in that paper.
\end{abstract}

\section{Introduction} \label{s:intro}

Comparator circuits are sorting networks \cite{Bat68} in
which the wires carry Boolean values. A comparator circuit is presented as a set of $m$ horizontal lines, which we call \emph{wires}. The left end of each wire is annotated by either a Boolean constant, an input variable, or a negated input variable, and one wire is designated as the output wire. In between the wires there is a sequence of comparator gates, each represented as a vertical arrow connecting some wire $w_{i}$ with some wire $w_{j}$, as shown in~\figref{f0}.  These arrows divide each wire into segments, each of which gets a Boolean value. The values of wires $w_i$ and $w_{j}$ after the arrow are the maximum and the minimum of the values of wires $w_{i}$ and $w_{j}$ right before the arrow, with the tip of the arrow pointing at the position of the maximum. Every wire is initialized (at its left end) by the annotated value, and the output of the circuit is the value at the right end of the output wire.  Thus comparator circuits are essentially Boolean
circuits in which the gates have restricted fanout.

Problems computed by uniform polynomial size families of comparator
circuits form the complexity class $\CC$, defined by
Subramanian~\cite{Sub90} in a slightly different guise. Mayr and Subramanian~\cite{MS92} showed that $\NL \subseteq \CC \subseteq \P$ (where $\NL$
is nondeterministic log space), and gave several complete problems for the class, including the stable marriage problem ($\SMP$) and lexicographically first maximal matching ($\LFMM$). Known algorithms for these problems are inherently sequential, and so they conjectured that $\CC$ is incomparable with $\NC$, the class of problems computable in polylog parallel time.  (For the other direction, it is conjectured that the $\NC^2$ problem of raising an $n\times n$ matrix to the $n$th power is not in $\CC$.) Furthermore, they proposed that $\CC$-hardness be taken as evidence that a problem is not parallelizable.

Since then, other problems have been shown to be $\CC$-complete: the stable roommate problem \cite{Sub94}, the telephone connection problem \cite{RW91}, the problem of predicting internal diffusion-limited aggregation clusters from theoretical physics \cite{MM00}, and the decision version of the hierarchical clustering problem \cite{GK08}. The maximum weighted matching problem has been shown to be $\CC$-hard~\cite{GHR95} %In several of these cases, $\CC$-completeness has been taken as evidence against the existence of an efficient parallel algorithm.
The fastest known parallel algorithms for some $\CC$-complete problems are listed in~\cite[\S B.8]{GHR95}.

\begin{figure}
\begin{small}
  \begin{center}
    $\Qcircuit @C=1.4em @R=0.5em {
\push{\bit{1}}&\push{x_{0}}	& \ctrl{2}	&\push{\,0\,}\qw& \qw 	&\qw
& \ctrl{3}&\push{\,0\,}\qw 	& \qw&   \rstick{\bit{0}=f(x_0,x_1,x_2)} \qw \\
\push{\bit{1}}&\push{x_{1}}	&\qw		&\qw		&\ctrl{2} 	&\push{\,0\,}\qw
&\qw&\qw  	&\ctrlu{1}&   \rstick{\bit{1}} \qw\\
\push{\bit{1}}&\push{x_{2}}	&\qw		&\qw		&\qw  	&\qw
&\qw&\qw  	&\qw &   \rstick{\bit{1}} \qw\\
\push{\bit{0}}&\push{\lnot x_{0}}	& \ctrld{-1} &\push{1}\qw& \qw  &\qw
&\qw& \qw 	&\ctrl{-1}&    \rstick{\bit{0}} \qw \\
\push{\bit{0}}&\push{\lnot x_{1}}	& \qw 	&\qw		& \ctrld{-1} &\push{\,1\,}\qw
&\qw& \qw & \qw&    \rstick{\bit{1}} \qw \\
\push{\bit{0}}&\push{\lnot x_{2}}	& \qw 	&\qw		& \qw 	&\qw
& \ctrld{-3} &\push{\,0\,}\qw& \qw&    \rstick{\bit{0}} \qw 
}$
  \end{center}
\end{small}
  \caption{}
  \label{fig:f0}
\end{figure}

\subsection{Our results}

We have two main results. First, we give several new characterizations of the class $\CC$, thus showing that it is a robust class. Second, we give an oracle separation between $\CC$ and $\NC$, thus providing evidence for the conjecture that the two classes are incomparable.

\subsubsection{Characterizations of $\CC$} Subramanian~\cite{Sub90} defined $\CC$ as the class of all languages log-space reducible to the comparator circuit value problem ($\CCV$), which is the following problem: given a comparator circuit with specified Boolean inputs, determine the output value of a designated wire. Cook, L\^e and Ye~\cite{LCY11,CLY11} considered two other classes, one consisting of all languages $\AC^0$-reducible to $\CCV$, and the other consisting of all languages $\AC^0$-Turing-reducible to $\CCV$, and asked whether these classes are the same as $\CC$. We answer this in the affirmative, and furthermore give characterizations of $\CC$ in terms of uniform circuits. Our work gives the following equivalent characterizations of the class $\CC$:

\begin{itemize}
 \item All languages $\AC^0$-reducible to $\CCV$.
 \item All languages $\NL$-reducible to $\CCV$.
 \item All languages $\AC^0$-Turing-reducible to $\CCV$.
 \item Languages computed by uniform families of comparator circuits. (Here \emph{uniform} can be either $\AC^0$-uniform or $\NL$-uniform.)
 \item Languages computed by uniform families of comparator circuits with inverter gates.
\end{itemize}

This shows that the class $\CC$ is robust. In addition, we show that corresponding function class $\FCC$ is closed under composition. The key novel notion for all of these results is our introduction of \emph{universal comparator circuits}.

The characterization of $\CC$ as uniform families of comparator circuits also allows us to define $\CC$ via certain Turing machines with an implicit access to the input tape.

\subsubsection{Oracle separations} 
Mayr and Subramanian~\cite{MS92} conjectured that $\CC$ and $\NC$ are incomparable. We provide evidence supporting this conjecture by separating the relativized versions of these classes, in which the circuits have access to oracle gates. Our techniques also separate the relativized versions of $\CC$ and $\SC$, where $\SC$ is the class of problems which can be solved simultaneously in polynomial time and polylog space. In fact, it appears that the three classes $\NC$, $\CC$ and $\SC$ are pairwise incomparable. In particular, although $\NL$ is a subclass of both $\NC$ and $\CC$, it is unknown whether $\NL$ is a subclass of $\SC$.

\subsubsection{Other results} The classical Gale-Shapley algorithm for the stable marriage problem~\cite{GS62} cannot be implemented as a comparator circuit. Subramanian~\cite{Sub94} devised a different fixed-point algorithm which shows that the problem is in $\CC$. We prove an interpretation of his algorithm that highlights its connection to the Gale-Shapley algorithm.

Another fixed-point algorithm is Feder's algorithm for directed reachability (described in Subramanian~\cite{Sub90}), which shows that $\NL \subseteq \CC$. We interpret this algorithm as a form of depth-first search, thus simplifying its presentation and proof.

\subsection{Background}

A fundamental problem in theoretical computer science asks whether every feasible problem can be solved efficiently in parallel. Formally, is $\NC = \P$? It is widely believed that the answer is negative, that is there are some feasible problems which cannot be solved efficiently in parallel. One class of examples consists of $\P$-complete problems, such as the \emph{circuit value problem}. $\P$-complete problems play the same role as that of $\NP$-complete problems in the study of problems solvable in polynomial time: 
\begin{center}
If $\NC \neq \P$, then $\P$-complete problems are not in $\NC$.
\end{center}

With the goal of understanding which  "tractable" problems cannot be solved efficiently in parallel, during the 1980's researchers came up with a host of $\P$-complete problems. Researchers were particularly interested in what aspects of a problem make it inherently sequential. Cook~\cite{Coo85} came up with one such aspect: while the problem of finding a maximal clique in a graph is in $\NC$~\cite{KW85}, if we require the maximal clique to be the one computed by the greedy algorithm (lexicographically first maximal clique), then the problem becomes $\P$-complete. His proof uses a straightforward reduction from the monotone circuit value problem, which is a $\P$-complete restriction of the circuit value problem~\cite{Gol77}.

Anderson and Mayr~\cite{AM87} continued this line of research by showing the $\P$-completeness of other problems asking for maximal structures computed by greedy algorithms, such as lexicographically first maximal path. However, one problem resisted their analysis, lexicographically first maximal matching ($\LFMM$), which is the same as lexicographically first independent set in line graphs. While trying to adapt Cook's proof to the case of line graphs, they encountered difficulties simulating fanout. They conjectured that the problem is not $\P$-complete.

Unbeknownst to Anderson and Mayr, Cook had encountered the same problem in 1983 (while working on his paper~\cite{Coo85}), and was able to come up with a reduction from the comparator circuit value problem ($\CCV$) to $\LFMM$. The hope was that like other variants of the circuit value problem such as the monotone circuit value problem and the planar circuit value problem, $\CCV$ would turn out to be $\P$-complete. (The planar monotone circuit value problem, however, is in $\NC$~\cite{Yan91,DK95,RY96}.)

The foregoing prompted Mayr's student Ashok Subramanian to study the relative strength of variants of the circuit value problem restricted by the set of allowed gates. Specifically, Subramanian was interested in the significance of fanout. To that end, Subramanian considered circuits without fanout, but instead allowed multi-output gates such as the COPY gate which takes one input $x$ and outputs two copies of $x$.

Subramanian classified the relative strength of the circuit value problem when the given set of gates can be used to simulate COPY: there are seven different cases, and in each of them the circuit value problem is either in $\NC$ or $\P$-complete. The case when the given set of gates cannot simulate COPY is more interesting: if all gates are monotone then the circuit value problem is either in $\NC$ or $\CC$-hard~\cite[Corollary 5.22]{Sub90}; for the non-monotone case there is no complete characterization. The class $\CC$ thus emerges as a natural ``minimal'' class above $\NC$ for the monotone case.

Cook, L\^e and Ye~\cite{CLY11,LCY11} constructed a uniform proof theory $\VCC$ (in the style of Cook and Nguyen~\cite{CN10}) which corresponds to~$\CC$, and showed that Subramanian's results are formalizable in the theory. The present paper answers several open questions raised in~\cite{CLY11,LCY11}.

\subsubsection*{Paper organization}

% In Section \ref{s:prelim} we define the basic concepts, including
% $\CC$ and its complete problems. \todo{Move the discussion on circuit uniformity to Section~\ref{s:universal}.}
% 
% In Section \ref{s:universal} we introduce the notion of a universal
% comparator circuit.  As applications we prove several results
% showing the robustness of the class $\CC$.  We show that the the
% class $\FCC$
% of functions associated with $\CC$ is closed under composition, and
% hence $\CC$ is closed under $\AC^0$ `Turing' reductions (which we call
% simply $\AC^0$ reductions \cite{AO96,CN10}); these reductions are 
% given by a uniform family of $\AC^0$ circuits with oracle gates making queries to other $\CC$ problems.
% We also characterize $\CC$ in the style of other circuit classes
% such as $\NC$ and $\AC$:  a problem is in $\CC$ if and only if it is computed by
% some uniform polynomial size family of comparator circuits (where
% the circuit inputs are allowed repeated copies of the problem
% input bits and their negations). \todo{Make sure that this corresponds to the future version of the section.}

In Section~\ref{s:prelim} we introduce definitions of basic concepts,
including comparator circuits and the comparator circuit value problem. We also introduce several equivalent definitions of $\CC$, in terms of uniform comparator circuits, and in terms of problems reducible (in various ways) to the comparator circuit value problem.

In Section~\ref{s:universal} we construct universal comparator circuits. These are comparator circuits which accept as input a comparator circuit $C$ and an input vector $Y$, and compute the output wires of $C$ when run on input $Y$. As an application, we prove that the various definitions of $\CC$ given at the end of the preceding section are indeed equivalent, and we provide yet another definition in terms of restricted Turing machines.

In Section \ref{s:oracleSep} we define a notion of relativized $\CC$
and prove that relativized $\CC$ is incomparable with relativized
$\NC$. This of course implies that relativized $\CC$ is strictly contained 
in relativized $\P$.  We also argue that $\CC$ and $\SC$ might be incomparable.

In Section \ref{sec:lfcom} we prove that the lexicographically
first maximal matching problem and its variants are complete
for $\CC$ under $\AC^0$ many-one reductions. (Two of the present authors
made a similar claim in \cite{LCY11}, but that proof works for
log space reductions rather than $\AC^0$ reductions.)

In Section \ref{sec:stable} we show that the stable marriage problem is complete for $\CC$, using Subramanian's algorithm~\cite{Sub90,Sub94}. We show that Subramanian's fixed-point algorithm, which uses three-valued logic, is related to the Gale-Shapley algorithm via an intermediate \emph{interval algorithm}. The latter algorithm also explains the provenance of three-valued logic: an interval partitions a person's preference list into \emph{three} parts, so we use three values $\set{0,\ast,1}$ to encode these three parts of a preference list.

Section \ref{conclusion} includes some open problems and our concluding remarks.

In Appendix \ref{NLinCC} we include a simple proof that $\CC$ contains $\NL$. 

%Some of the results in this paper appeared in \cite{CLY11}.  That paper also introduced a formal theory $\VCC$ which `captures' the complexity class $\CC$, and proved that many of the results in the present paper can be formalized in $\VCC$.

\section{Preliminaries}\label{s:prelim}
 \subsection{Notation}
 We use lower case letters, e.g. $x,y,z,\ldots$, to denote unary arguments and upper case letters, e.g. $X,Y,Z,\ldots$, to denote binary string arguments. For a binary string $X$,  we write $|X|$ to denote the length
 of $X$.

 \subsection{Function classes and search problems} \label{s:problems}
 A complexity class consists of relations $R(X)$, where $X$ is a binary string argument.
 Given a class of relations $\Class$, we associate a class $\FC$ of
 functions $F(X)$ with $\Class$ as follows. We require
 these functions to be $p$-bounded, i.e., $|F(X)|$ is bounded by
 a polynomial in $|X|$. Then we define $\FC$
 to consist of all
 %$p$-bounded number functions whose graphs are in $\Class$ and all 
 $p$-bounded string
 functions whose bit graphs are in $\Class$.  (Here the \emph{bit graph}
 of $F(X)$ is the relation $B_F(i,X)$ which holds iff the $i$th bit of
$F(X)$ is 1.)   For all classes $\Class$ of interest here the function
class $\FC$ is closed under composition.  In particular $\FNL$ is closed
under composition because $\NL$ (nondeterministic log space)
is closed under complementation.

Most of the computational problems we consider here can be nicely
expressed as decision problems (i.e. relations), but the stable
 marriage problem is an exception, because in general a given
 instance has more than one solution (i.e. there is more than one stable
 marriage).
 Thus the problem is properly described as a search problem.  A
 {\em search problem} $Q_R$ is a multivalued function with graph
 $R(X,Z)$, so  $Q_R(X) = \bset{Z\mid R(X,Z)}$.

 The search problem is {\em total} if the set $Q_R(X)$
 is non-empty for all $X$.  The search problem is a
 {\em function problem} if $|Q_R(X)|=1$  for all $X$.
 A function $F(X)$ {\em solves} $Q_R$ if  $F(X)\in Q_R(X)$ for all $X$.
 We will be concerned only with total search problems in this paper.

\subsection{Reductions} \label{sec:reductions}
 %The following notion of reduction preserves totality.

 Let $\Class$ be a complexity class. A relation $R_1(X)$ is $\Class$ many-one reducible
 to a relation $R_2(Y)$ (written $R_1\mred^\Class R_2) $ if there is a function $F$ in $\FC$ such that
 $R_1(X) \lra R_2(F(X)).$

 A search problem $Q_{R_1}(X)$ is $\Class$ many-one reducible to
 a search problem $Q_{R_2}(Y)$ if there are functions
 $G,F$ in $\FC$ such that
 $G(X,Z) \in Q_{R_1}(X)$ for all $Z\in Q_{R_2}(F(X))$.

Recall that problems in $\AC^0$ are computed by uniform families of
polynomial size constant depth circuits.  (See~\cite{BIS90} for
many equivalent definitions, including `First Order definable'.)
We are interested in $\AC^0$ many-one reducibility, but also in the
more general notion of $\AC^0$ Turing reducibility.  Such a reduction
is given by a uniform family of $\AC^0$ circuits with oracle gates.
We note that standard small complexity classes including
 $\AC^{0}$, $\TC^0$, $\NC^{1}$, $\NL$ and $\P$
are closed under $\AC^0$ Turing reductions.

%\subsubsection*{Notation} We use upper case letters, e.g. $X,Y,Z,\ldots$, to denote binary vectors (also known as strings). For a binary vector $X$, we write $|X|$ to denote the length of $X$.

\subsection{Comparator circuits} \label{s:compCirc}

A \emph{comparator gate} is a function $C\colon\set{0,1}^{2}\rightarrow \set{0,1}^{2}$ that takes an input  pair $(p,q)$ and outputs a pair $(p\wedge q,p \vee q)$. Intuitively, the first output in the pair is the smaller bit among the two input bits $p,q$, and the second output is the larger bit. \smallskip\\
\begin{minipage}{.65\textwidth}
\hspace{.5cm} We will use the graphical notation on the right to denote a comparator gate, where $x$ and $y$ denote the names of the wires, and the direction of the arrow denotes the direction to which we move the larger bit as shown in the picture. \smallskip
\end{minipage}
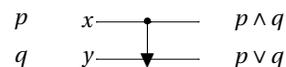
\begin{minipage}{.3\textwidth}
\vspace{-5mm}
\begin{center}\small 
$\Qcircuit @C=2.5em @R=1em {
\push{\bit{p}}&\push{x\,} & \ctrl{1}	& \rstick{~~~~~\bit{p\wedge q}}\qw\\
\push{\bit{q}}&\push{y\,} & \ctrld{-1} & \rstick{~~~~~\bit{p\vee q}}\qw  
}$
\end{center}
\end{minipage}

A \emph{comparator circuit} consists of $m$ wires and a sequence $(i_1,j_1),\ldots,(i_n,j_n)$ of $n$ comparator gates. We allow ``dummy'' gates of the form $(i,i)$, which do nothing. A comparator circuit computes a function
$f\colon \{0,1\}^m \to \{0,1\}^m$ in the obvious way (See Section
\ref{s:intro}). A \emph{comparator circuit with negation gates} additionally has negation gates $N\colon \set{0,1} \to \set{0,1}$ which invert their input.

An \emph{annotated comparator circuit} consists of a comparator circuit with a distinguished output wire together with an annotation of the input of each wire by an input bit $x_i$, the negation of an input bit $\lnot x_i$, or a constant $0$ or~$1$. A \emph{positively annotated comparator circuit} has no annotations of the type $\lnot x_i$. An annotated comparator circuit computes a function $f\colon \{0,1\}^k \to \{0,1\}$, where $k$ is the largest index appearing in an input annotation $x_k$.

An \emph{$\AC^0$-uniform} family of annotated comparator circuits is one
whose circuits are are computable by a function in $\FAC^0$.  An \emph{$\NL$-uniform} 
family of annotated comparator circuits is one whose circuits are
computable in $\FNL$.

\smallskip

Uniform annotated comparator circuits will serve as the basis of some
definitions of the complexity class $\CC$. Other definitions will be based on a complete problem, the \emph{comparator circuit value problem}.

\begin{definition}\label{d:CC} 
The comparator circuit value problem ($\CCV$) is the decision
problem:  given a comparator circuit and an assignment of bits to the
input of each wire, decide whether a designated wire outputs  one. By default,
we often let the designated wire be the $0$th wire of a circuit.
\end{definition}

Comparator circuits can have some gates pointing up, and others pointing
down.  The next result shows that in most of our proofs there is no
harm in assuming that all gates point in the same direction.

\begin{proposition}\label{prop:p1}
$\CCV$ is $\AC^0$ many-one reducible to the special case in which
all comparator gates point down (or all point up).
\end{proposition}

See Appendix \ref{a:pointDown} for the proof.
%For a complexity class $\C$, the corresponding function class $\FC$ consists of functions $F(X)\colon\{0,1\}^n \to \{0,1\}^{m(n)}$ such that $m$ is polynomial in $n$ and the function $(X,i) \mapsto F(X)_i$ (the $i$th bit of $F(X)$) belongs to $\C$. A predicate $P_1$ is $\C$-many-one reducible to a predicate $P_2$ if there is a function $F \in \FC$ such that $P_1(X) = P_2(F(X))$.
%
%A predicate $P_1$ is $\AC^0$-Turing reducible to a predicate $P_2$ if $P_1$ can be computed using an $\AC^0$-uniform family of bounded depth circuits with polynomially many wires having AND-gates, OR-gates and $P_2$-gates of arbitrary fan-in, as well as NOT-gates. For each $n$, a $P_2$-gate of width $n$ has $n$ inputs and $1$ output, and the function it computes is $(x_1,\ldots,x_n) \mapsto P_2(x_1,\ldots,x_n)$. 
%
%An \emph{$\AC^0$-many-one reduction} is a function $\{0,1\}^n \to \{0,1\}^m$ computed by an $\AC^0$-uniform family of $\AC^0$ circuits with two inputs: $X \in \{0,1\}^n$ and $Y \in \{0,1\}^{\lceil \log_2 m \rceil}$. An $\NL$-many-one reduction is a function $\{0,1\}^n \to \{0,1\}^m$ computable in $\NL$. An \emph{$\AC^0$-Turing reduction} is an $\AC^0$-uniform family of $\AC^0$ circuits with oracle gates. 
%
%\smallskip

The class $\CC$ has many equivalent definitions. We adopt one of them as the ``official'' definition, and all the rest are shown to be equivalent in Section~\ref{s:universal}.
\begin{definition} \label{d:classes} $\;$
\begin{itemize}
 \item $\CC$ is the class of relations computed by an $\AC^0$-uniform family of polynomial size annotated comparator circuits.
 \item $\CCnot$ is the class of relations computed by an $\AC^0$-uniform family of polynomial size positively annotated comparator circuits with negation gates.
 \item $\CCACMO$ is the class of relations $\AC^0$-many-one reducible to $\CCV$.
 \item $\CCNLMO$ is the class of relations $\NL$-many-one reducible to $\CCV$.
 \item $\CCACT$ is the class of relations $\AC^0$-Turing reducible to $\CCV$.
\end{itemize}
\end{definition}

\section{Universal comparator circuits} \label{s:universal}

An annotated comparator circuit $\UNIV_{m,n}$ is \emph{universal} if it can simulate any comparator circuit with at most $m$ wires and $n$ gates, as precisely stated in the following theorem. Here we present the first known (polynomial size) construction of such a
universal comparator circuit family. For simplicity, our universal circuit only simulates unannotated comparator circuits, but the idea extends to cover annotated comparator circuits as well.

\begin{theorem} \label{t:universal}
 There is an $\AC^0$-uniform family of annotated comparator circuits $\UNIV_{m,n}$ which satisfy the following property: for any comparator circuit $C$ having at most $m$ wires and $n$ gates, the designated output wire of $\UNIV_{m,n}$ fed with inputs $C,Y$ equals the $0$th output wire of $C$ fed with input $Y$.
\end{theorem}
\begin{proof}
The key idea is a {\em gadget} consisting of
a comparator circuit with four wires and four gates which allows a
conditional application of a gate to two of its
inputs $x,y$, depending on whether a control bit $b$ is 0 or 1.  The
other two inputs are $\overline{b}$ and $b$ (see~\figref{universal}).
The gate is applied only when $b = 1$
(see~\figref{universal:proof}).

\medskip

\begin{minipage}{0.37\textwidth}
\centering
{\small
  $\Qcircuit @C=1.5em @R=0.5em {
    \push{\bit{b\vphantom{\overline{b}y'}}\;\;} & \ctrl{2}  & \qw & \ctrl{1}  & \qw & \ctrl{2}  & \rstick{\bit{0}} \qw \\
    \push{\bit{x\vphantom{\overline{b}y'}}\;\;} & \qw       & \qw & \ctrld{0} & \qw & \qw       & \rstick{\bit{x'}} \qw \\
    \push{\bit{y\vphantom{\overline{b}y'}}\;\;} & \ctrld{0} & \qw & \ctrl{1}  & \qw & \ctrld{0} & \rstick{\bit{y'}} \qw \\
    \push{\bit{\overline{b}\vphantom{ y'}}\;\;} & \qw       & \qw & \ctrld{0} & \qw & \qw       & \rstick{\bit{1}} \qw
  }$}

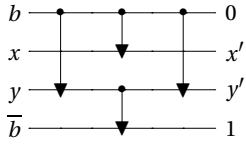
\captionof{figure}{Conditional comparator gadget}
 \label{fig:universal}
\end{minipage}
\hspace{1cm}
\begin{minipage}{0.5\textwidth}
\centering
{\small
  $\Qcircuit @C=1.2em @R=0.5em {
    \push{\bit{0\vphantom{\overline{b}y'}}\;\;} & \ctrl{2}  & \push{\bit{0}} \qw & \ctrl{1}  & \push{\bit{0}} \qw & \ctrl{2}  & \rstick{\bit{0}} \qw & \push{\qquad\quad\bit{1}\;\;} & \ctrl{2}  & \push{\bit{y}} \qw & \ctrl{1}  & \push{\bit{x\land y}} \qw & \ctrl{2}  & \rstick{\bit{0}} \qw \\
    \push{\bit{x\vphantom{\overline{b}y'}}\;\;} & \qw       & \push{\bit{x}} \qw & \ctrld{0} & \push{\bit{x}} \qw & \qw       & \rstick{\bit{x}} \qw & \push{\qquad\quad\bit{x}\;\;} & \qw       & \push{\bit{x}} \qw & \ctrld{0} & \push{\bit{x \lor y}} \qw & \qw & \rstick{\bit{x \lor y}} \qw \\
    \push{\bit{y\vphantom{\overline{b}y'}}\;\;} & \ctrld{0} & \push{\bit{y}} \qw & \ctrl{1}  & \push{\bit{y}} \qw & \ctrld{0} & \rstick{\bit{y}} \qw & \push{\qquad\quad\bit{y}\;\;} & \ctrld{0} & \push{\bit{1}} \qw & \ctrl{1}  & \push{\bit{0}} \qw & \ctrld{0} & \rstick{\bit{x\land y}} \qw \\ 
    \push{\bit{1\vphantom{\overline{b}y'}}\;\;} & \qw       & \push{\bit{1}} \qw & \ctrld{0} & \push{\bit{1}} \qw & \qw       & \rstick{\bit{1}} \qw & \push{\qquad\quad\bit{0}\;\;} & \qw       & \push{\bit{0}} \qw & \ctrld{0} & \push{\bit{1}} \qw & \qw & \rstick{\bit{1}} \qw
  }$}

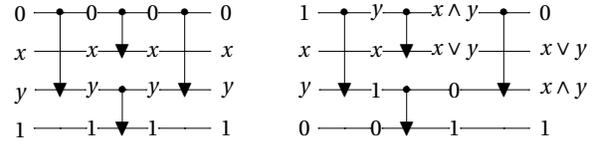
\captionof{figure}{Operation of conditional comparator gadget}
 \label{fig:universal:proof}
\end{minipage}

\smallskip

In order to simulate a single arbitrary gate in a circuit with $m$
wires we put in $m(m-1)$ gadgets in a row, for the $m(m-1)$ possible
gates. Simulating $n$ gates requires $m(m-1)n$ gadgets.
The bits of $C$ are the control bits for the gadgets.
The resulting circuit can be constructed in an $\AC^0$-uniform fashion.
\end{proof}

As a consequence, we can identify the classes $\CC$ and $\CCACMO$.

\begin{lemma} \label{l:ccacmo}
 The two complexity classes $\CC$ and $\CCACMO$ are identical.
\end{lemma}
\begin{proof}
 We start with the easy direction $\CC \subseteq \CCACMO$. Suppose $R$ is a relation computed by $\AC^0$-uniform polynomial size annotated comparator circuits $C_m$. Given an input $X$ of length $m$, we can construct in uniform $\AC^0$ the circuit $C_m$ and replace each input annotation with the corresponding constant, bit of $X$ or its negation. The value of $R$ is computed by applying $\CCV$ to this data.

 We proceed to show that $\CCACMO \subseteq \CC$. Our first observation is that every $\AC^0$ circuit can be converted to a polynomial size formula, and so to a comparator circuit. Every relation $R$ in $\CCACMO$ is given by a uniform $\AC^0$ reduction $F$, such that $R(X)=1$ if and only if $\CCV(F(X)) = 1$. Suppose that the circuit computed by $F(X)$ on an input of size $k$ has $m_k$ wires and $n_k$ comparator gates. We construct an annotated comparator circuit computing $R(X)$ by first computing $F(X)$ and then feeding the result into $\UNIV_{m_k,n_k}$ (see Fig.~\ref{fig:ccc}). The resulting circuit is $\AC^0$-uniform.

\begin{figure}
\begin{small}
\begin{center}
        \tikzstyle{vertex}=
        [%
          	minimum size=5mm,%
          	rectangle,%
	        	thick,%
		text centered
        ]
      \begin{tikzpicture}[>=stealth',scale=1]	
	\draw (3,1) rectangle (12,7.85);

	\foreach \pos/\label/\name in 
		{{(13,7.4)/R(X)/},
		{(5.25,4.425)/\mathrm{Compute}\;\; F(X)/},
		{(9.75,4.425)/\UNIV_{m_k,n_k}/}}
        	\node[vertex] (\name) at \pos {$\label$};
		
	\node at(1.5,1.45){$\neg X(k-1)$};
	\node at(1.5,2.3){$\neg X(k-1)$};		
	\node at(1.5,3.15){$\neg X(0)$};
	\node at(1.5,4){$\neg X(0)$};		
	\node at(1.65,4.85){$X(k-1)$};
	\node at(1.65,5.7){$X(k-1)$};		
	\node at(1.65,6.55){$X(0)$};
	\node at(1.65,7.4){$X(0)$};
	
	\node at(1.65,3.7){$\vdots$};	
	\node at(1.65,2.85){$\vdots$};
	\node at(1.65,2){$\vdots$};
	\node at(1.65,5.4){$\vdots$};
	\node at(1.65,6.25){$\vdots$};
	\node at(1.65,7.1){$\vdots$};
		
	\node at(2.75,3.7){$\vdots$};	
	\node at(2.75,2.85){$\vdots$};
	\node at(2.75,2){$\vdots$};
	\node at(2.75,5.4){$\vdots$};
	\node at(2.75,6.25){$\vdots$};
	\node at(2.75,7.1){$\vdots$};
	\draw (12,7.4) -- (12.5,7.4);	
	\draw (7.5,1)--(7.5,7.85);
	
	\draw (2.5,1.45) -- (3,1.45);	
	\draw (2.5,2.3) -- (3,2.3);
	\draw (2.5,3.15) -- (3,3.15);	
	\draw (2.5,4) -- (3,4);	
	\draw (2.5,4.85) -- (3,4.85);	
	\draw (2.5,5.7) -- (3,5.7);		
	\draw (2.5,6.55) -- (3,6.55);	
	\draw (2.5,7.4) -- (3,7.4);
		
      \end{tikzpicture}
\end{center}
\end{small}
\vspace{-3mm}
\caption{The construction of a comparator circuit for any relation $R \in \CCACMO$ for a fixed input length $k$.}
\label{fig:ccc}
\end{figure}
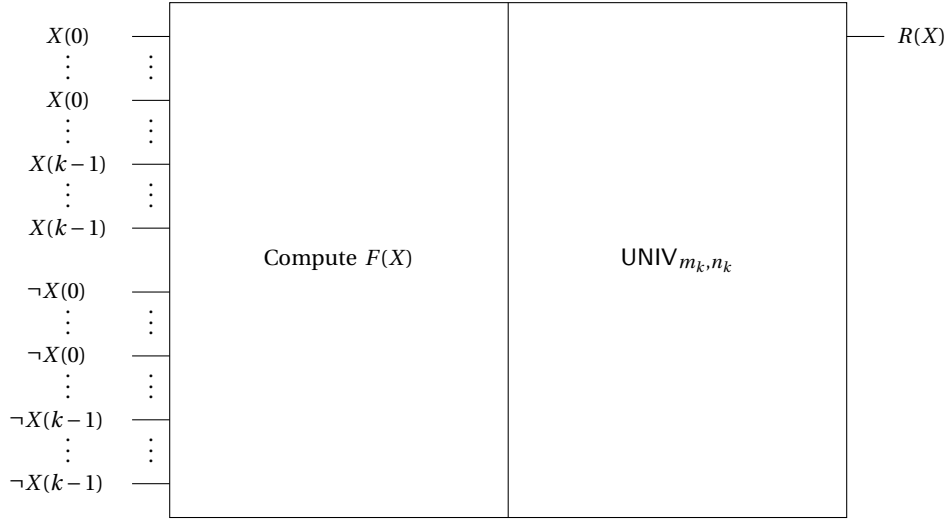
\end{proof}

The preceding construction replaces an arbitrary $\AC^0$ function $F$ with a corresponding comparator circuit. By allowing stronger functions $F$, we get additional characterizations of $\CC$. In particular, in Appendix~\ref{NLinCC} we show that $\NL \subseteq \CC$, and this allows us to replace $\AC^0$-many-one reductions with $\NL$-many-one reductions.

\begin{lemma} \label{l:ccnlmo}
 The two complexity classes $\CCACMO$ and $\CCNLMO$ are identical.
\end{lemma}
\begin{proof}
 It is clear that $\CCACMO \subseteq \CCNLMO$. For the other direction, let $R \in \CCNLMO$ be computed as $R(X) = \CCV(F(X))$, where $F \in \NL$. Theorem~\ref{theo:NL2CCV} shows that $F \in \CC$. We can therefore use the construction of Lemma~\ref{l:ccacmo}.
\end{proof}

Another characterization is via $\AC^0$-Turing reductions. We require a preliminary lemma corresponding to de Morgan's laws.

\begin{lemma} \label{l:demorgan}
 Given an annotated comparator circuit $C$ computing a function $F(X)$, we can construct an annotated comparator circuit $C'$ computing the function $\lnot F(X)$ with the same number of wires and comparator gates. Furthermore, $C'$ can be constructed from $C$ in $\AC^0$.
\end{lemma}
\begin{proof}
 The idea is to push negations to the bottom using de Morgan's laws. Given $C$, construct $C'$ by negating all inputs (switch $0$ and $1$, $x_i$ and $\lnot x_i$) and flipping all comparator gates: replace each gate $(i,j)$ with a gate $(j,i)$ which has the same inputs but the reversed outputs. An easy induction shows that $C'(X) = \lnot C(X)$.
\end{proof}

\begin{lemma} \label{l:ccact}
 The two complexity classes $\CCACMO$ and $\CCACT$ are identical.
\end{lemma}
\begin{proof}
 It is clear that $\CCACMO \subseteq \CCACT$. For the other direction, let $R \in \CCACT$ be computed by a uniform $\AC^0$ family of circuits $C_k$ with $\CCV$ oracle gates.
 %Without loss of generality, we can assume that the circuits $C_k$ do not have fan-out (but can still have unbounded fan-in).
 The basic idea is to replace each oracle gate with the corresponding universal comparator circuit (Theorem~\ref{t:universal}). However, a universal comparator circuit expects each of its inputs to appear polynomially many times, some of them negated. We will handle this by duplicating the corresponding portion of the constructed circuit enough times, using Lemma~\ref{l:demorgan} to handle negations. The resulting circuit will be polynomial size since $C_k$ has constant depth.

 We proceed with the details. For each gate $g$ in the circuit $C_k$, we define a transformation $T(g)$ mapping it to an annotated comparator circuit computing the same function, as follows. If $g = x_i$ then $T(g)$ is a circuit with one wire annotated $x_i$. If $g = \lnot g_1$ then $T(g)$ is obtained from $T(g_1)$ via Lemma~\ref{l:demorgan}. If $g = g_1 \lor \cdots \lor g_\ell$ or $g = g_1 \land \cdots \land g_\ell$ then $T(g)$ is obtained by constructing fresh copies of $T(g_1),\ldots,T(g_\ell)$ and joining them with $\ell-1$ comparator gates. If $g$ is an oracle gate with inputs $g_1,\ldots,g_\ell$ then we take the universal comparator circuit from Theorem~\ref{t:universal} and replace each input bit by the corresponding copy of $T(g_i)$ or of its negation obtained via Lemma~\ref{l:demorgan}.

 If $r$ is the root gate of $C_k$ then it is clear that $T(r)$ is an annotated comparator circuit computing the same function as $C_k$. It remains to estimate the size of $T(r)$, and for that it is enough to count the number of comparator gates $|T(r)|$. Suppose that $C_k$ has $n$ gates and depth $d$, where $n$ is polynomial in $k$ and $d$ is constant. For $0 \leq \Delta \leq d$, we compute a bound $B_\Delta$ on $|T(g)|$ for a gate $g$ at depth $d - \Delta$. If $d = 0$ then $|T(g)| = 0$, and so we can take $B_0 = 0$. For $d = \Delta + 1 \neq 0$, the costliest case is when $g$ is an oracle gate. Suppose the inputs to $g$ are $g_1,\ldots,g_\ell$ of depth at least $d-\Delta$; note that $\ell \leq n$. We construct $T(g)$ by taking $\ell^{O(1)}$ copies of $T(g_i)$ and appending to them a comparator circuit of size $\ell^{O(1)}$. Therefore $|T(g)| \leq \ell^{O(1)}(B_\Delta + 1)$ and so $B_{\Delta+1} \leq n^{O(1)} (B_\Delta + 1)$. Solving the recurrence, we deduce $B_d \leq n^{O(d)}$ and so $|T(r)| \leq B_d$ is at 
most polynomial in $k$ (since $n = k^{O(1)}$ and $d = O(1)$).

 To complete the proof, we observe that $T(r)$ can be computed in $\AC^0$, and so $\CCACT \subseteq \CC = \CCACMO$ (using Lemma~\ref{l:ccacmo}).

%Without loss of generality, we can assume that the circuits $C_k$ are actually formulas. We would like to replace each oracle gate with the corresponding universal comparator circuit (Theorem~\ref{t:universal}). However, a universal comparator circuit expects each of its inputs to appear polynomially many times, some of them negated. For each occurrence of each input bit of the universal comparator circuit, we duplicate the corresponding portion of the circuit $C_k$; the result still has polynomial size because $C_k$ has constant depth. \todo{Check?} Finally, to handle negation, we \emph{push negations to the bottom}: given an annotated comparator circuit (in this case, a portion of $C_k$), we can construct an annotated comparator circuit computing the negated function by flipping all comparator gates (switching ANDs and ORs) and negating all inputs; de Morgan's laws imply the correctness of this construction.
\end{proof}

Finally, we show that the same class is obtained if we allow negation gates, using a reduction outlined in Section~\ref{sec:cvn2c}.

\begin{lemma} \label{l:ccnot}
 The two complexity classes $\CC$ and $\CCnot$ are identical.
\end{lemma}
\begin{proof}
 Recall that $\CC$ consists of relations computed by uniform annotated comparator circuits, while $\CCnot$ consists of relations computed by uniform positively annotated comparator circuits with negation gates. It is easy to see that $\CC \subseteq \CCnot$: every negated input $\lnot x_i$ can be replaced by the corresponding positive input followed by a negation gate. For the other direction, Section~\ref{sec:cvn2c} shows how to simulate a comparator circuit with negation gates using a standard comparator circuit.
\end{proof}

In summary we have shown that all complexity classes listed at the end of Section~\ref{s:compCirc} are identical.

\begin{theorem} \label{t:cc}
 All complexity classes $\CC,\CCnot,\CCACMO,\CCNLMO,\CCACT$ are identical.
\end{theorem}

Subramanian~\cite{Sub90} originally defined $\CC$ as the class of relations log-space reducible to $\CCV$. Theorem~\ref{t:cc} shows that this class is identical to the one we consider in the present paper.

Our methods can also be used to show that the function class $\FCC$ corresponding to $\CC$ is closed under composition.

\begin{theorem} \label{t:fcc}
 The class $\FCC$ is closed under composition.
\end{theorem}
\begin{proof}
 Suppose that $F,G \in \CC$ are given by uniform positively annotated comparator circuits with negation gates. By taking enough copies of $G$, we can construct a uniform positively annotated comparator circuit with negation gates computing $F(G(X))$.
\end{proof}

Theorem~\ref{t:cc} implies that $\CC$ is the class of all relations computed by $\L$-uniform positively annotated comparator circuits with negation gates. This characterization corresponds to the following Turing machine model.

\begin{theorem} \label{t:cctm}
 Every relation in $\CC$ is computable by a Turing machine of the following type. The machine has three tapes: a read/write work tape with one head $W$, a read-only input tape with one head $I$, and a `comparator' tape with two heads $M_1,M_2$. All heads are movable as usual (one step at a time). The 
work tape is initalized by the size of the input $n$ encoded in binary,
and is limited to $O(\log n)$ cells. The input tape is initialized with
the input, and comparator tape is initially blank.  The work tape
controls the actions of heads $I, M_1, M_2$ but cannot read their scanned
symbols.  The input and comparator tapes are only accessible
via the following special operations, under control of the work tape:
\begin{itemize}
 \item Move heads $I$, $M_1$, and $M_2$ to specified locations.
 \item Write a blank cell $M_1$ (or $M_2$) with $0$ (or $1$).
 \item Copy the value at cell $I$ to a blank cell $M_1$ (or $M_2$).
 \item Negate the value at cell $M_1$ (or $M_2$).
 \item ``Sort'' the values at cells $M_1$ and $M_2$.
 \item Output the value at cell $M_1$ (or $M_2$), and halt the machine.
\end{itemize}
 Conversely, every relation computed by such a machine is in $\CC$.
\end{theorem}
\begin{proof}
 For every relation $R \in \CC$ there is a log-space machine $T$ that on input $n$ outputs a positively annotated comparator circuit with negation gates computing $R$ on inputs of length $n$. We can convert $T$ to a machine $T'$ of the type described in the theorem by thinking of each cell of the
comparator tape as one wire in the circuit.  We replace each annotation by
the second or third special operation and each gate by one of the following two special operations. The final special operation is used to single out the distinguished output wire.

 For the other direction, given a machine $T'$ of the type described we
can construct a log-space machine $T$ which on input $n$ outputs 
a suitable annotated comparator circuit for inputs of length $n$ as
follows:   $T$ on input $n$ runs $T'$ with its work tape initiated
to $n$ in binary, and uses the control signals to the input and
comparator tapes to describe the circuit.
\end{proof}

This machine model is resilient under the following changes: allowing more heads on each of the tapes, allowing the second or third special operations to write over a non-blank cell, and adding other operations such as swapping the values of two cells in the comparator tape. We leave the proof to the reader.

\section{Oracle separations} \label{s:oracleSep}

Here we support our conjecture that the complexity classes $\NC$ and $\CC$ are incomparable by defining and separating relativized versions of the
classes.  (See Section \ref{s:intro} for a discussion of the conjecture.)
Problems in relativized $\CC$ are computed by comparator circuits which
are allowed to have oracle gates, as well as comparator gates and $\neg$
gates.  (We allow $\neg$ gates to make our results more general.
Note that $\neg$ gates can be eliminated from nonrelativized comparator
circuits, as explained in Section~\ref{sec:cvn2c}.)
Each oracle gate computes some function $G\colon \{0,1\}^n \ra \{0,1\}^n$
for some $n$.    We can insert such an oracle gate anywhere
in an oracle comparator circuit with $m$ wires, as long as $m \ge n$,
by selecting a level in the circuit, selecting any $n$ wires and
using them as inputs to the gate (so each gate input gets one of the
$n$ distinct wires), and then the $n$ outputs feed
to some set of $n$ distinct output wires.  Note that this definition
preserves the limited fan-out property of comparator circuits:
each output of a gate is connected to at most one input of one other gate.
%We require each oracle gate
%in the circuit to compute same function.

We are interested in oracles $\alpha\colon  \{0,1\}^* \ra  \{0,1\}^*$
which are length-preserving, so $|\alpha(Y)| = |Y|$.
We use the notation $\alpha_n$ to refer to the restriction of
$\alpha$ to $\{0,1\}^n$.  We define the relativized complexity
class $\CC(\alpha)$ based on the circuit-family characterization
of $\CCnot$ given in Definition~\ref{d:classes}.   Thus a
relation $R(X,\alpha)$ is in $\CC(\alpha)$ iff it is computed by
a polynomial size family of annotated comparator circuits which are allowed
comparator gates, $\neg$-gates, and $\alpha_n$ oracle gates, where
$n = |X|$.  We consider both a uniform version (in which each circuit
family satisfies a uniformity condition) and a nonuniform version of
$\CC(\alpha)$.

Analogous to the above,
we define the relativized class $\NC^k(\alpha)$ 
to be the class of relations $R(X,\alpha)$ computed by some family
of depth $O(\log^k n)$ polynomial size Boolean circuits with
$\wedge$, $\vee$, $\neg$, and $\alpha_n$-gates (where $n=|X|$)
in which $\wedge$-gates and $\vee$-gates have fan-in at most two, and oracle gates are nested at most $O(\log^{k-1} n)$ levels deep\footnote{Cook~\cite{Coo85} defines the depth of an oracle gate of fan-in $m$ to be $\log m$. Our definition follows the one by Aehlig at al.~\cite{ACN}, and has the advantage that it preserves classical results like $\NC^1(\alpha) \subseteq \L(\alpha)$. The class $\NC(\alpha)$ is the same under both definitions.}. 
As above, we consider both uniform and non-uniform versions of these
classes. Also $\NC(\alpha) = \bigcup_k \NC^k(\alpha)$.

%Analogous to the above,
%we define the relativized classes $\NC^k(\alpha)$ (resp.
% $ (\AC^k(\alpha))$
%to be the class of relations $R(X,\alpha)$ computed by some family
%of depth $O(\log^k n)$ polynomial size Boolean circuits with
%$\wedge$, $\vee$, $\neg$, and $\alpha_n$-gates in which all gates
%have fan-in at most two (resp. unlimited fan-in).  (In the case
%of oracle gates the fan-in restriction applies to each input of
%the gate. \footnote{This restriction is not present in Definition 2
%of \cite{ACN}, but it strengthens our Theorem \cite{t:NCproper}.}
%As above, we consider both uniform and non-uniform versions of these
%classes.

As observed earlier, flipping one input of a comparator gate flips
exactly one output.  We can generalize this notion to oracles $\alpha$
as follows.

\begin{definition}\label{d:lipschitz}
A partial function $\alpha: \{0,1\}^* \rightharpoonup \{0,1\}^*$ which is
length-preserving on its domain is \emph{(weakly) 1-Lipschitz} if for all
strings $X,X'$ in the domain of $\alpha$, if $|X| = |X'|$ and
$X$ and $X'$ have Hamming distance 1, then $\alpha(X)$ and
$\alpha(X')$ have Hamming distance \emph{at most} 1.

A partial function $\alpha: \{0,1\}^* \rightharpoonup \{0,1\}^*$ which is
length-preserving on its domain is \emph{strictly 1-Lipschitz} if for all
strings $X,X'$ in the domain of $\alpha$, if $|X| = |X'|$ and
$X$ and $X'$ have Hamming distance 1, then $\alpha(X)$ and
$\alpha(X')$ have Hamming distance \emph{exactly} 1.
\end{definition}

Subramanian~\cite{Sub90} uses a slightly different terminology: (strictly) adjacency-preserving for (strictly) 1-Lipschitz.

Since comparator gates compute strictly 1-Lipschitz functions, 
it may seem reasonable to restrict comparator oracle circuits to
strictly 1-Lipschitz oracles $\alpha$. 
Our separation results below hold whether or not we make this restriction.

Roughly speaking, we wish to prove
$\CC(\alpha) \not\subseteq \NC(\alpha)$
and $\NC^3(\alpha) \not\subseteq \CC(\alpha)$.
More precisely, we have the following result.

\begin{theorem}\label{t:separation} $\;$
\begin{itemize}
\item[(i)] There is a relation $R_1(\alpha)$ which is computed by some
uniform polynomial size family of comparator oracle circuits, but which
cannot be computed by any $\NC(\alpha)$ circuit family (uniform or not),
even when the oracle $\alpha$ is restricted to be strictly 1-Lipschitz.
\item[(ii)]
There is a relation $R_2(\alpha)$ which is computed by some
uniform $\NC^3(\alpha)$ circuit family which cannot be computed
by any polynomial size family of comparator oracle circuits
(uniform or not), even when the oracle $\alpha$ is restricted
to be strictly 1-Lipschitz.
\end{itemize}
\end{theorem}

The proof of Theorem~\ref{t:separation}(ii) only uses the fact that comparator gates are 1-Lipschitz, and therefore the lower bound implies to the more general class $\APC$ considered by Subramanian~\cite{Sub90} in which comparator gates are replaced by arbitrary 1-Lipschitz gates. Conversely, the comparator circuit in the proof of Theorem~\ref{t:separation}(i) uses only oracle gates, and so the relation $R_1(\alpha)$ belongs to any circuit class whose depth is unrestricted.

The restriction to strictly 1-Lipschitz oracles might seem severe. However the following lemma shows how to extend every 1-Lipschitz function to a strictly 1-Lipschitz function. As a result, it is enough to prove a relaxed version of Theorem~\ref{t:separation} where the oracles are restricted to be only (weakly) 1-Lipschitz.

\begin{lemma} \label{lem:parity}
Suppose $f$ is a 1-Lipschitz function. Define $g(X)$ by adding a
`parity bit' in front of $f(X)$ as follows:
\[ g(X) = \bigl[\parity{X} \oplus \parity{f(X)}\bigr] f(X). \]
Then $g$ is strictly 1-Lipschitz.
\end{lemma}
\begin{proof}
Suppose $d(X,Y) = 1$. Clearly $d(g(X),g(Y)) \leq 2$. On the other hand, $\parity{g(X)} = \parity{X}$, and hence $d(g(X),g(Y))$ is odd. We conclude that $d(g(X),g(Y)) = 1$.
\end{proof}

Given a relation $S(\beta)$ which separates two relativized complexity classes even when the oracle $\beta$ is restricted to 1-Lipschitz functions, define a new relation $R(\alpha) = S(\mathrm{chop}(\alpha))$, where $\mathrm{chop}(\alpha)_n(X)$ results from $\alpha_{n+1}(0 X)$ by chopping off the leading bit. Given an oracle $\beta$ which is 1-Lipschitz, define a new oracle $\alpha$ by
\[ \alpha_{n+1}(b X) = \bigl[\parity{b X} \oplus \parity{\beta_n(X)}\bigr] \beta_n(X). \]
Lemma~\ref{lem:parity} shows that $\alpha$ is strictly 1-Lipschitz. Notice that $R(\alpha) = S(\beta)$. Hence $R(\alpha)$ separates the two relativized complexity classes even when the oracle $\alpha$ is restricted to strictly 1-Lipschitz functions. 

Henceforth we will prove the relaxed version of Theorem~\ref{t:separation} in which the oracles are only restricted to be (weakly) 1-Lipschitz.

\subsection{Proof of item (i) of Theorem \ref{t:separation}}
It turns out that item (i) is easy to prove if we require the $\NC(\alpha)$
circuit family to work on all length-preserving oracles $\alpha$, and
not just 1-Lipschitz oracles.
This is a consequence of the next proposition, which
follows from the proof of  \cite[Theorem~14]{ACN}, and states that
the $\ell$th iteration of an oracle requires a circuit with oracle
nesting depth $\ell$ to compute.

\begin{definition}\label{d:nesting}
The {\em nesting depth} of an oracle gate $G$ in an oracle circuit 
is the maximum number
of oracle gates (counting $G$) on any path in the circuit from an
input (to the circuit) to $G$.
\end{definition}

\begin{proposition} \label{t:aehlig}
Let $d,n >0$ and
let  $C(\alpha)$ be a circuit with any number of Boolean gates
but with fewer than $2^n$ $\alpha_n$-gates such that the nesting
depth of any $\alpha_n$-gate is at most $d$.
If the circuit correctly computes the first bit of $\alpha_n^\ell$ (the $\ell$th iteration of $\alpha_n$),
and this is true for all oracles $\alpha_n$, then $\ell \le d$.
\end{proposition}

We apply Proposition \ref{t:aehlig}  with $d = n$, and conclude that the
first bit
of $\alpha_n^n$ cannot be computed in $\NC(\alpha)$.  But $\alpha^n_n$
obviously can be computed in $\CC(\alpha)$ by placing $n$
oracle gates $\alpha_n$ in series.  This proves item (i) without the
1-Lipschitz restriction.

\medskip

For the proof of Proposition \ref{t:aehlig} we use the following definition
and lemma from \cite{ACN}.

\begin{definition}\label{def:ell-sequential-partial-function}
A partial function
$f\colon\{0,1\}^n\rightharpoonup\{0,1\}^n$ is called
\emph{$\ell$-sequential} if  (abbreviating $0^n$ by $\mathbf{0}$)
\[\mathbf{0},f(\mathbf{0}),f^2(\mathbf{0}),\ldots,f^\ell(\mathbf{0})\]
are all defined, 
but $f^\ell(\mathbf{0})\not\in\dom(f)$.
\end{definition}

Note that
in Definition~\ref{def:ell-sequential-partial-function} it is
necessarily the case that $\mathbf{0},f(\mathbf{0}),f^2(\mathbf{0}),\ldots,f^\ell(\mathbf{0})$ are distinct.

\begin{lemma}\label{lem:ell-sequential-extension}
Let $n\in\NN$ and $f\colon\{0,1\}^n\rightharpoonup\{0,1\}^n$ be an
$\ell$-sequential partial function.  Let $M\subset\{0,1\}^n$
be such that $|\dom (f) \cup M|<2^n$. Then there is an $(\ell+1)$-sequential
extension $f'\supseteq f$ with $\dom(f')=\dom(f)\cup M$.
\end{lemma}

\begin{proof}
Let $Y\in\{0,1\}^n\setminus \bigl(M\cup\dom(f)\bigr)$. Such a $Y$ exists by our assumption
on the cardinality of $M\cup\dom(f)$. Let $f'$ be $f$ extended by
setting $f'(x)=Y$ for all $X\in M\setminus\dom(f)$. This $f'$ is as
desired.

Indeed, assume that $\mathbf{0},f'(\mathbf{0}),\ldots,f'^{\ell+1}(\mathbf{0}),f'^{\ell+2}(\mathbf{0})$ are
all defined. Then, since $Y\not\in\dom(f')$, it follows that all the
$\mathbf{0},f'(\mathbf{0}),\ldots,f'^{\ell+1}(\mathbf{0})$ have to be different from $Y$. Hence these
values have already been defined in $f$. But this contradicts the
assumption that $f$ was $\ell$-sequential.
\end{proof}

\begin{proof}[Proof of Proposition \ref{t:aehlig}]
We use $f$ to stand for the oracle function $\alpha_n$.
Assume that such a circuit computes $f^\ell(\mathbf{0})$ correctly for all
oracles. We have to find a setting for the oracle that witnesses
$\ell\leq d$.

By induction on $k\ge 0$ we define partial functions
$f_k\colon\{0,1\}^n\rightharpoonup\{0,1\}^n$ with the following
properties.
\begin{itemize}
\item $f_0\subseteq f_1\subseteq f_2\subseteq\ldots$
\item The size $|\dom(f_k)|$ of the domain of $f_k$ is at most the
number of oracle gates of nesting depth $k$ or less.
\item $f_k$ determines the values of all oracle gates of nesting depth
$k$ or less.
\item $f_k$ is $k$-sequential.
\end{itemize}
We can take $f_0$ to be the totally undefined function, since
$f^0(\mathbf{0}) = 0$ by definition, so $f_0$ is $0$-sequential. 
For the induction step let $M$ be the set of all strings $Y$ of length $n$
such that $Y$ is queried by an oracle
gate at level $k$.  Let $f_{k+1}$ be a $k{+}1$-sequential extension of
$f_k$ to domain $\dom(f_k)\cup M$ according to
Lemma~\ref{lem:ell-sequential-extension}.

For $k=d$ we get the desired bound. As $f_d$ already
determines the values of all gates, the output of the circuit is
already determined, but $f^{d+1}(\mathbf{0})$ is still undefined and we can
define it in such a way that it differs from the first bit of the
output of the circuit.
\end{proof}

Now we are ready to prove item (i) for the case when oracles are
restricted to be 1-Lipschitz. \\[2mm]
{\bf Notation}. We use $T$ to stand for both a bit string and the number
it represents in binary.  The $i$th bit of $T$ is $\bbit{T}{i}$;
the least significant bit (lsb) is bit~$1$.
% The binary encoding of a number $x$ in the range $[0,2^n)$ (where $n$ is understood) is $\bina{x}$. 
For a bit $b$, $\rep{b}{m}$ is the bit $b$ repeated $m$ times. The Hamming
weight of a string $\vc{x}$ is $\ham{\vc{x}}$. The Hamming distance between $\vc{x}$ and $\vc{y}$ is $d(\vc{x},\vc{y}) = \ham{\vc{x} \oplus \vc{y}}$. The length of $\vc{x}$ is $\len{\vc{x}}$. %The parity of $\vc{x}$ is $\parity{\vc{x}}$. 
%Concatenation is denoted by $\conc$. 
%A zero vector of appropriate length (understood from context) is $\vczero$.

\begin{definition} \label{def:goal}
 Let $n = 2^{\ell}$, and define $m = 2n+1$. Given $f \colon \{0,1\}^{m\ell + n} \rightarrow \{0,1\}^n$, define the slice functions
 \[
  f_{\vc{t}}(\vc{x}) = f\left(\rep{\bbit{\vc{t}}{1}}{m}\ldots\rep{\bbit{\vc{t}}{\ell}}{m}\vc{x}\right), \text{ where } \len{\vc{t}} = \ell \text{ and } \len{\vc{x}} = n.
 \]
 Define the iterations
 \[
  \vc{x}_0 = \rep{0}{n}, \quad \vc{x}_{\vc{t}+1} = f_{\vc{t}}(\vc{x}_{\vc{t}}).
 \]
 Finally, define
 \[
  F = \bbit{\vc{x}_{\lfloor \sqrt{n} \rfloor}}{1}.
 \]
\end{definition}

\begin{lemma} \label{lem:cc}
The function $F = F(f)$ can be computed using a uniform family of
comparator circuits of size polynomial in $n$ which use $f$ as an
oracle.\footnote{We can pad the output of
$f$ with $m\ell$ zeros so that $f$ has the same number of outputs as inputs.}
\end{lemma}
\begin{proof}
 Obvious.
\end{proof}

If $f$ ignores its first $m\ell$ input bits then $F$ is the first
bit of of the $\sqrt{n}$-th iteration of $f$, and hence by 
Proposition~\ref{t:aehlig} any subexponential size circuit computing $F$ requires
depth $\sqrt{n}$.
In the rest of this section we will show %by using all inputs of $f$
that even if $f$ is assumed to be 1-Lipschitz,
$F$ cannot be computed by any circuit with only polynomially many oracle
gates which are nested only polylogarithmically deep.

The first step is to reduce the problem of constructing a 1-Lipschitz $f$ to the problem of constructing 1-Lipschitz $f_0,\ldots,f_{n-1}$.

\begin{definition} \label{def:blob}
 Let $f$ be a function as in Definition~\ref{def:goal}. Let $\vc{r}_1\ldots\vc{r}_{\ell}\vc{x}$ be an input to $f$, where $\len{\vc{r}_i} = m$, $\len{\vc{x}} = n$. Suppose $\vc{r}_i$ contains $z_i$ zeroes and $o_i$ ones. Define $t_i = 0$ if $z_i > o_i$ and $t_i = 1$ if $z_i < o_i$ (one of these must happen since $m$ is odd). Let $x_i = \min(z_i,o_i)$ and $x = \max_i x_i$. The values $t_1,\ldots,t_{\ell}$ define a string $\vc{t}$. We say that
$\vc{r}_1\ldots\vc{r}_{\ell}X$ belongs to the \emph{blob}
$B(\vc{t},X)$, and is at distance $x$ from the center string
$\rep{t_1}{m}\ldots\rep{t_{\ell}}{m}X$.  Thus the blobs form a partition of
the domain $\{0,1\}^{m\ell+n}$ of $f$.

 We say that $f$ is \emph{blob-like} if for all $R_1,\ldots,R_{\ell},X$,
with $T$ as defined above,
\begin{equation}\label{e:bloblike}
  f(\vc{r}_1\ldots\vc{r}_{\ell}\vc{x}) = f_{\vc{t}}(\vc{x}) \land \bigl(\rep{0}{x} \rep{1}{n-x}\bigr). 
\end{equation}
 (Here we use bitwise $\land$.) In words, the value of $f$ at a point $\vc{r}$ which is at distance $x$ from the center of some blob $B$ is equal to the value of $f$ at the center of the blob, with the first $x$ bits set to zero.

 We say that $f$ is a \emph{blob-like partial function} if it is a partial function whose domain is a union of blobs, and inside each blob it satisfies~\eqref{e:bloblike}.  
\end{definition}

Note that the values at centers of blobs are unconstrained by
(\ref{e:bloblike})
because then $x=0$. %\todo{What does this mean?}

\begin{lemma} \label{lem:blob}
 If $f$ is blob-like and $f_{\vc{t}}$ is 1-Lipschitz for all $0 \leq \vc{t} < n$ then $f$ is 1-Lipschitz.
\end{lemma}
\begin{proof}
 Let $\vc{r}_1\ldots\vc{r}_{\ell},\vc{x}$ be an input to $f$. We argue that if we change a bit in the input, then at most one bit changes in the output. If we change a bit of $\vc{x}$, then this follows from the 1-Lipschitz property of the corresponding $f_{\vc{t}}$. If we change a bit of $\vc{r}_i$ without changing $\vc{t}$, then we change $x$ by at most~$1$, and so at most one bit of the output is affected. Finally, if we change a bit of $\vc{r}_i$ and this does change $\vc{t}$, then we must have had (without loss of generality) $z_i = n, o_i = n+1$, and we changed a $1$ to $0$ to make $z_i = n+1, o_i = n$. In both inputs, $x = n$, and so the output is $\vczero$ in both cases.
\end{proof}

The second step is to find a way to construct 1-Lipschitz functions from
$\{0,1\}^n$ to itself, given a small number of constraints.

\begin{lemma} \label{lem:point}
 Suppose $g \colon \{0,1\}^n \rightarrow \{0,1\}^n$ is a partial function, and $g(\vc{p}) = \vczero$ for all $\vc{p} \in \dom(g)$. Let $\vc{x}$ be a point of Hamming distance at least $d$ from any point in $\dom(g)$. Then for every $\vc{y}$ of Hamming weight at most $d$, we can extend $g$ to a 1-Lipschitz total function satisfying $g(\vc{x}) = \vc{y}$.
\end{lemma}
\begin{proof}
 Given $\vc{x},\vc{y}$, define $h(\vc{z})$ to be $\vc{y}$ with the first $\min(d(\vc{z},\vc{x}),\ham{\vc{y}})$ ones changed to zeros. We have $h(\vc{x}) = \vc{y}$ since $d(\vc{x},\vc{x}) = 0$. For $\vc{p} \in \dom(g)$, $d(\vc{p},\vc{x}) \geq d$ implies $h(\vc{p}) = \vczero$, using $\ham{\vc{y}} \leq d$. Therefore $h$ extends $g$. On the other hand, $h$ is 1-Lipschitz since changing a bit of the input $\vc{z}$ can change $d(\vc{z},\vc{x})$ by at most $1$, and so at most one bit of the output is affected.
\end{proof}

Finally, we need a technical lemma about the volume of Hamming balls.

\begin{definition} \label{def:volume}
 Let $n,d$ be given. Then $V(n,d)$ is the number of points in $\{0,1\}^n$
of Hamming weight at most $d$, that is
 \[ V(n,d) = \sum_{k \leq d} \binom{n}{k}. \]
\end{definition}

\begin{lemma} \label{lem:volume}
 For $d \geq 0$, $V(n,d+1)/V(n,d) \leq n+1$. 
\end{lemma}
\begin{proof}
 Each point in $V(n,d+1)$ is either already a point of $V(n,d)$, or it can be obtained by taking a point of $V(n,d)$ and changing one bit from $0$ to $1$. Conversely, for each point of $V(n,d)$, a bit can be changed from $0$ to $1$ in at most $n$ different ways.
\end{proof}

\begin{corollary} \label{cor:volume}
 If $V(n,d) \geq r\ge 1$ then there exists $d' \geq 0$ such that
 \[ r \leq \frac{V(n,d)}{V(n,d')} < (n+1)r. \]
\end{corollary}
\begin{proof}
 Let $d'$ be the maximum number satisfying $r \leq V(n,d)/V(n,d')$. Since $r \leq V(n,d) = V(n,d)/V(n,0)$, such a number exists. On the other hand,
 \[ \frac{V(n,d)}{V(n,d')} \leq (n+1) \frac{V(n,d)}{V(n,d'+1)} < (n+1)r. \qedhere \]
\end{proof}

We are now ready to prove the main lower bound, which implies
item (i) of Theorem \ref{t:separation}.

\begin{theorem} \label{thm:main}
 Let $a>0$ be given. For large enough $n$, every circuit $C(f)$ with at most $n^a$ oracle gates, nested less than $\sqrt{n}$ deep, fails to compute $F$ for some 1-Lipschitz function $f$.
\end{theorem}
\begin{proof}
 Put $\mT = \lfloor \sqrt{n} \rfloor - 1$. Let $d_0,\ldots,d_{\mT}$ be a sequence of positive integers satisfying
 \begin{equation} \label{eq:condition}
  \frac{V(n,d_T)}{V(n,d_{T+1}-1)} > n^a, \quad 0 \leq T < \mT.
 \end{equation}
 Such numbers exist whenever $2^n \geq (n+1)^{\mT(a+1)}$, which holds when $n$ is large enough. Indeed, we will construct such a sequence inductively using Corollary~\ref{cor:volume}, keeping the invariant
 \[ V(n,d_T) \geq \frac{2^n}{(n+1)^{T(a+1)}}. \]
 For the base case, $d_0 = n$ certainly satisfies the invariant. Given $d_T$, use Corollary~\ref{cor:volume} with $d = d_T$ and $r = (n+1)^a$. Since $V(n,d_T) \geq 2^n/(n+1)^{T(a+1)} \geq (n+1)^{a+1}$ for large $n$, the corollary supplies us with $d'$ satisfying
 \[ (n+1)^a \leq \frac{V(n,d_T)}{V(n,d')} < (n+1)^{a+1}. \]
 Let $d_{T+1} = d' + 1$. This certainly satisfies condition~\eqref{eq:condition}, and the invariant is satisfied since
 \[
  V(n,d_{T+1}) > V(n,d') > \frac{V(n,d_T)}{(n+1)^{a+1}} \geq \frac{2^n}{(n+1)^{(T+1)(a+1)}}.
 \]

 \smallskip

 We will define the function $f$ in $\mT$ stages, similar to the proof of
Proposition \ref{t:aehlig}, except we use the notation $f^{(k)}$ instead 
of $f_k$.  At every stage the function $f^{(k)}$ will be a blob-like 
partial function which defines the output of every oracle gate in $C(f)$ of
nesting depth $k$ or less. The starting point is $f^{(0)}$, which is the empty function.
At stage $k$ we will define the partial function $f^{(k+1)}$, which
extends $f^{(k)}$, keeping the following invariants:
 \begin{itemize}
  \item $f^{(0)} \subseteq f^{(1)} \subseteq f^{(2)}  \subseteq \ldots$
  \item $f^{(k)}$ is a blob-like partial function.
  \item For $T < k$, $f^{(k)}_T$ is a total 1-Lipschitz function.
  \item For $T \geq k$, at any point $P$ at which $f^{(k)}_T$ is defined it
is equal to $\vczero$ and some gate in $C(f)$ of nesting depth $k$ or
less has its input in the blob $B(T,P)$.  Hence $|\dom(f^{(k)}_T)| \le n^a$.
  \item $\vc{x}_k$ is defined by $f^{(k)}$.
  \item $f^{(k)}_k(\vc{x}_k)$ is undefined.
  \item[(*)] $d(\vc{p},\vc{x}_k) \geq d_k$ for any $\vc{p} \in \dom(f^{(k)}_k)$.
  \item $f^{(k)}$ determines that the output of every oracle gate of nesting
depth $k$ or less.

% The distance of an oracle gate $G$ from the inputs is the least number $u$ such that each path from $G$ to the inputs contains at most $u$ gates (including $G$).
 \end{itemize}
\bigskip
\noindent
 It is easy to verify that the empty function $f^{(0)}$ satisfies the invariants. The function $f^{(\mT)}$ determines the output of the circuit
$C(f)$. However, $\vc{x}_{\lfloor \sqrt{n} \rfloor} = f^{(\mT)}_T(\vc{x}_{\mT})$ is undefined. We can
extend $f^{(\mT)}$ to a 1-Lipschitz function in two different ways: 
Put $f_T = \vczero$ for $T > \mT$, and let $f_{\mT}$ be either (1) the constant zero function, or (2) the function which is zero everywhere except for $f_{\mT}(\vc{x}_{\mT}) = \rep{0}{n-1} 1$. Since $F$ is different in these two extensions, the circuit fails to compute $F$ correctly in one of them.

 It remains to show how to define $f^{(k+1)}$ given $f^{(k)}$. Let
$\GG$ be the set of oracle gates of nesting depth exactly $k+1$. 
For any $G \in \GG$ whose input belongs to a blob $B(T,\vc{x})$ 
for $T > k$, if $f^{(k)}_{T}(\vc{x})$ is undefined, then define 
$f^{(k+1)}$ so that it extends $f^{(k)}$ and is $\vczero$
on the entire blob $B(T,\vc{x})$
(this is a blob-like assignment). Let $A = \dom(f^{(k+1)}_{k+1})$; 
note that $|A| \leq n^a$. Condition~\eqref{eq:condition} implies
 \[ V(n,d_{k+1}-1) |A| < V(n,d_k), \]
 and so there is a point $Y$ of Hamming weight at most $d_k$ which is of 
distance at least $d_{k+1}$ from each point in $A$. Define
$f^{(k+1)}_k(\vc{x}_k) = Y$ (so $\vc{x}_{k+1} = Y$), and extend 
$f^{(k+1)}_k$ to a total 1-Lipschitz function using Lemma~\ref{lem:point} with $d = d_k$ (use invariant (*)).
Then extend $f^{(k+1)}$ to a blob-like partial function using 
\eqref{e:bloblike}. It is routine to verify that the invariants hold for $f^{(k+1)}$.
\end{proof}

\begin{remark}
\begin{itemize}
\item A more natural target function is $F' = \bbit{\vc{x}_n}{1}$. We can easily modify the proof of Theorem~\ref{thm:main} to handle this function. We set
the first $n - \lfloor \sqrt{n} \rfloor$ functions 
$f^{(0)},\ldots,f^{(n-\lfloor \sqrt{n} \rfloor-1)}$ to be constant, and then run the proof from that point on.
\item An even more natural target function has an unstructured $f$ as input, and $F'' = \bbit{f^{(N)}(\vczero)}{1}$. We leave open the question of whether the method can be adapted to work in this case.  %\todo{Why a question here?}
\item We have previously shown how to construct $F'''$ that separates $\CC$ and $\NC$ even under the restriction that the oracle be strictly 1-Lipschitz. The function $F'''$ basically ignores one of the outputs of the oracle while iterating it. It is possible to slightly modify the proof of Theorem~\ref{thm:main} so that it directly applies to $F$ even under the restriction that the oracle is strictly 1-Lipschitz. We leave this modification as an exercise to the reader. %\todo{Shouldn't mention exercise!} (OK to mention an excercise.)
 \end{itemize}
\end{remark}

%\begin{corollary} \label{cor:main}
% Relativized CC is not contained in relativized NC, even if the oracle is promised to be 1-Lipschitz (for any fixed number of inputs).
%\end{corollary}

%We can improve on Theorem \ref{thm:main} by making the oracle $f$
%``exactly'' 1-Lipschitz, that is, with the property that if one input bit is changed, one output bit is changed. 
%
%\begin{definition} \label{def:lipschitz}
%A function $f\colon \ZZ_2^n \rightarrow \ZZ_2^m$ is \emph{exactly 1-Lipschitz} if whenever we change one bit of the input, exactly one bit of the output changes.
%\end{definition}

%There are two ways to use this construction. One is to ignore the extra bit in the definition of $F$. The other one is to slightly change the proof. We demand that our partial functions be not only blob-like, but also parity-like.
%
%\begin{definition} \label{def:parity-like}
%A partial function $f \colon \ZZ_2^n \rightarrow \ZZ_2^m$ is \emph{parity-like} if for all $\vc{x}$ in the domain:
%\begin{itemize}
% \item $\parity{f(\vc{x})} = \parity{\vc{x}}$.
% \item $\vc{x'} = \vc{x} \oplus 1 \rep{0}{n-1}$ is in the domain and $f(\vc{x'}) = f(\vc{x}) \oplus 1 \rep{0}{m-1}$.
%\end{itemize}
%\end{definition}
%
%Maintaining this invariant, the proof carries through with the appropriate changes (for most of the proof, we just ignore the parity bit).

\subsection{Proof of item (ii) of Theorem \ref{t:separation}}

Here we exploit the 1-Lipschitz property of comparator gates and
$\neg$-gates by using oracles which are weakly 1-Lipschitz, so that all
gates in the relativized circuits have this property.
The idea is to use an
oracle $\alpha$ with $dn^2$ inputs ($d\ge 3$) but only $n$ useful
outputs.  We can feed the $n$ useful outputs back into another instance of
$\alpha$ by using $dn$ copies of each output bit. Because of the
1-Lipschitz property it seems as though a comparator circuit
computing the $m$th iteration of $\alpha$ in this way needs either
at least $2^{\Omega(m)}$
copies of $\alpha$, or alternatively $2^n$ copies of $\alpha$ and a
complicated circuit analyzing the output.
When $m = \Omega(\log^2 n)$, this construction requires a super-polynomial
size comparator circuit computing the $m$th iteration of $\alpha$.
On the other hand, for $m = O(\log^2 n)$, the $m$th iteration can be
easily computed in relativized $\NC^3$ (following Aehlig et
al.~\cite{ACN}, we require oracle gates to be nested at most
$O(\log^{k-1} n)$ deep in relativized $\NC^k$).

To make this argument work we initially assume that instead of computing
the $m$-th
iteration of $\alpha$ we compose $m$ different oracles $A_1,\ldots,A_m$
in the way just described.  The crucial property of comparator circuits
we use is the \emph{flip-path} property:
\begin{quote}
If one input wire is changed, then (given that each gate has the
 1-Lipschitz property) there is a unique path through the circuit tracing
the effect of the original flip.
\end{quote}
We use a Gray code to order the possible $n$-bit outputs of the oracle
and study the effects of the $2^n$ flip-paths generated as the
definition of the oracle is successively changed.

In detail,
Let $n,m,d \in \mathbb{N}$, with $d \ge 3$.  For each $k \in [m]$ and
$i \in [n]$, let $a^k_i \colon \{0,1\}^{dn} \rightarrow \{0,1\}$ be
a Boolean oracle with $dn$ input bits.
Let $A^k = (a^k_1,\ldots,a^k_n)$.
We define a function $y = f[A^1,\ldots,A^m]$ as follows:
\begin{align} 
 x^k_i &= a^k_i(\overbrace{x^{k+1}_1,\ldots, x^{k+1}_1}^{d \text{ times}},\ldots,\overbrace{x^{k+1}_n,\ldots, x^{k+1}_n}^{d \text{ times}}), && k \in [m], \, i \in [n], 
 \label{e:Xdefine}\\
 x^{m+1}_i &= 0, && i \in [n], \label{e:Zdefine}\\
 y &= x^1_1 \oplus \cdots \oplus x^1_n. \label{e:yDefine}
\end{align}

As stated the oracle $a^k_i$ has $dn$ inputs and just one output,
but we can make it fit our convention that an oracle gate has the
same number of outputs as inputs simply by assuming that the gate
has an additional $dn-1$ outputs which are identically zero. 

Note that the function computed by such an oracle is necessarily
(weakly) 1-Lipschitz.

Let $X^k = (x^k_1,\ldots,x^k_n)$ and $A^k = (a^k_1,\dots,a^k_n)$.
Note that an oracle circuit of depth
$m + O(\log n)$ with $mn$ gates can compute $y$ simply by
successively computing $X_m,X_{m-1},\ldots,X_1$ and computing the
parity of $X_1$, provided that
the circuit is allowed to have gates with fan-out $d$.  
However, the fan-out restriction for comparator circuits allows us
to prove the following.

\begin{theorem}\label{t:expBound}
If $n \ge 3$, then every oracle comparator circuit computing
$y = f[A^1,\ldots,A^m]$ has at least \[\min\left(2^n, (d-2)^{m-1}\right)\]
gates.
\end{theorem}

By setting $m = \log^2 n$ and $d = 4$ this almost proves item (ii)
of Theorem \ref{t:separation},
except we need to argue that the array of oracles $a^k_i$ can be
replaced by a single oracle.  Later we will show how a simple
adaptation of the proof of Theorem \ref{t:expBound} accomplishes this.

\begin{proof}[Proof of Theorem \ref{t:expBound}]

Fix an oracle comparator circuit $C$ which computes
$y = f[A^1,\ldots,A^m]$.

\begin{definition}  We say that an input $(z_1,\ldots,z_{dn})$ to some
oracle $a^k_i$ in $C$ is \emph{regular} if it has the form of the inputs in~\eqref{e:Xdefine}; that is if $z_{(a-1)d+b} = z_{(a-1)d+c}$ for all $a\in [n]$ and $b,c \in [d]$.  We say that an oracle $a^k_i$ is \emph{regular} if $a^k_i(Z) = 0$ for all irregular inputs $Z$.
\end{definition}

Note that any irregular oracle $a^k_i$ can be replaced by an equivalent
regular oracle  which does not affect (\ref{e:Xdefine}).

\begin{definition}
Let $g$ be the
total number of any of the gates $a^k_i$ in the circuit $C$.
For a given assignment to the oracles, a particular gate $a^k_i$ is
\emph{active} in $C$ if its input is as specified by
(\ref{e:Xdefine},\ref{e:Zdefine}).   

Let $g_k$ be the expected total number of active gates
$a^k_1,\ldots,a^k_n$ in $C$ under a uniformly random \emph{regular} setting of
\emph{all} oracles. 
\end{definition}

It is easy to see that
\begin{equation}\label{e:gn}
g_1 \geq n,
\end{equation}
since we need at least one active gate $a^1_i$ for each $i \in [n]$.

Let $k \in [m]$ be greater than $1$. We will show that
\begin{equation} \label{eq:main}
g_{k-1} \leq \frac{g}{2^n} + \frac{g_k}{d-2}.
\end{equation}

We use the following consequence of the (weakly) 1-Lipschitz property of all gates in the circuit:  If we change the definition of some copy of some gate $a^k_i$ at its input in the circuit $C$, this generates a unique flip-path which may end at some copy of some other gate, in which case we say that the latter gate {\em consumes} the flip-path.   
(The flip-path is a path in the circuit such that the Boolean value of each edge in the path is negated.)

Let $G_1,\ldots,G_{2^n}$ be a Gray code listing all strings in
$\{0,1\}^n$, where $G_1 = 0^n$.
Thus the Hamming distance between any two successive strings $G_i$ and
$G_{i+1}$, and between $G_{2^n}$ and $G_1$, is one.
Take a random regular setting of all the oracles, and let $Z_1$ be the
value of $X_k$ under this setting.  Shift the above Gray code to form
a new one $Z_1,\ldots,Z_{2^n}$ by setting $Z_t = G_t \oplus Z_1$.
Then for each $t \in [2^n]$, $Z_t$ is uniformly distributed and
independent of $X_\ell$ for $\ell \neq k$.   Thus if we change the output
of $A^k$ at its active input to $Z_t$, the result is again a
uniformly random regular oracle setting.
Let $\gamma_t$ be the number of active $A^k$ gates (i.e. any active gate
of the form $a^k_i$ for some $i$) after this change, and let $\delta_t$
be the number of active $A^{k-1}$ gates after the change. 
Taking expectations we have for each $t \in [2^n]$
\begin{equation}\label{e:expGammaDelta}
   \mathbb{E}(\gamma_t) = g_k,  \qquad  \mathbb{E}(\delta_t) = g_{k-1}.
\end{equation}

We will change the output of $A^k$ (at its active input) successively from $Z_1$ to $Z_{2^n}$, and consider the relationship between $\gamma_t$ and $\delta_t$.  The total number of flip paths generated during the process is 
\[
     \sum_{t=1}^{2^n-1} \gamma_t.
\]
 Each time an $A^{k-1}$ gate is rendered active for the first time, we will call the gate \emph{fresh}. Otherwise, it is \emph{reused}.   At time $t$, let $\delta'_t$ ($\delta''_t$) be the number of fresh (reused) $A^{k-1}$ gates.  Thus $ \delta''_1= 0$, and
\begin{equation}\label{e:delta}
    \delta_t = \delta'_t + \delta''_t
\end{equation}
Since a given gate can be fresh at most once, we have
\begin{equation}\label{e:freshOnce}
    \sum_{t=1}^{2^n} \delta'_t \le g
\end{equation}
Each time an $a^{k-1}_i$ gate is reused it has consumed at least $d-2$ flip-paths
since the last time it was active.  This is because at least $d$ consecutive inputs must be changed from all 0's to all 1's (or vice versa), and since the gate is regular, its output will be constantly 0 during at least $d-2$ consecutive changes.

Since there must be at least as many flip-paths generated as consumed, we have
\begin{equation}\label{e:sumdelta}
   (d-2) \sum_{t=1}^{2^n}\delta''_t \le \sum_{t=1}^{2^n} \gamma_t.
 \end{equation}

From (\ref{e:delta}), (\ref {e:freshOnce}), (\ref{e:sumdelta}) we have
\begin{equation}\label{e:deltagamma}
 \sum_{t=1}^{2^n} \delta_t \leq g + \frac{1}{d-2} \sum_{t=1}^{2^n} \gamma_t. 
\end{equation}
Now \eqref{eq:main} follows from \eqref{e:deltagamma} and
\eqref{e:expGammaDelta} by linearity of expectations.
%Taking expectations, and using the facts 
%\begin{equation}\label{e:expectID}
%  \EE(\sum_{t=1}^{2^n} \gamma_t) = 2^n g^{act}_{k}   \qquad \qquad
% \EE(\sum_{t=1}^{2^n} \delta_t) = 2^n g^{act}_{k-1} 
% \end{equation}

Hence either $g > 2^n$ or
\[
   g_k \ge (d-2) [g_{k-1} - 1].
\]
From this and (\ref{e:gn}) we have a recurrence whose solution shows
\[
   g_t \ge (d-2)^{t-1} n - \frac{(d-2)^t - (d-2)}{d-3}.
\]
If $n \ge 3$ then $g_m \ge (d-2)^{m-1}$, and Theorem~\ref{t:expBound} follows.
\end{proof}

Now we change the setting in Theorem \ref{t:expBound} so that
it applies to a single oracle.  The new oracle $a(k,i,x)$ is used
in the same way as $a^k_i(x)$.  The first two arguments can be encoded
in binary or unary, and we don't care what happens when
they are not "legal" (we don't require the output to be 0 unless the $x$
argument is illegal).
Define active $a^k_i$ gates as gates whose inputs are $(k,i,x)$, where
$x$ is the relevant active input.  We argue as before,
and again conclude that if $n \ge 3$ then $g_m \geq (d-2)^{m-1}$.
Hence Theorem \ref{t:expBound} follows in the single oracle setting,
and Theorem  \ref{t:separation} (ii) follows as explained right after the
statement of Theorem  \ref{t:expBound}.

\subsection{$\SC$ vs $\CC$}\label{s:SCvsCC}

Uniform $\SC^k$ is the class of problems decidable by Turing machines
running in polynomial time using $O(\log^k n)$ space.  Non-uniform
$\SC^k$ is the class of problems solvable by circuits of polynomial size
and $O(\log^k n)$ width.  Just as $\NC$ and $\CC$ appear to be
incomparable, it seems
plausible that $\SC$ and $\CC$ are incomparable.  For one direction,
$\NL$ is a subclass of both $\CC$ and $\NC$, but is conjectured not
to be a subclass of $\SC$ (Savitch's algorithm takes $2^{O(\log^2)}$
time).  For the other direction we can give a convincing oracle
separation as follows.

We apply Theorem \ref{t:expBound} to a problem with a padded input
of length $N$, and set the `real' input $n = \log^2 N$, and
also $m = \log^2 N$ and $d = 4$.  The theorem implies that every
comparator circuit solving $F$ has size at least
\[2^{\min(m-1,n)} = 2^{\log^2 N},\] 
which is superpolynomial.
Thus this padded problem is not in relativized $\CC$.

However, a Turing machine $M$ equipped with an oracle tape which can
query $4\log^2 N$ bits of each oracle $a^k_i$ can compute this
padded version of $F$ in linear time and $O(n) =O(\log^2 N)$ space,
so this problem is in relativized $\SC^2$. The machine $M$ proceeds
by successively computing $X_m,X_{m-1},\ldots,X_1$, writing each
of these $X_i$ on its work tape and then erasing the previous one.
The machine computes $X^k$ from $X^{k+1}$ bit by bit, making a
query of size $4\log^2 N$ to its query tape for each bit.
(We assume that $M$ can access its oracle in such a way that it
can determine $N$, and hence $m$ and $n$.)

\section{Lexicographically first maximal matching problems are $\CC$-complete \label{sec:lfcom}}
Let $G=(V,W,E)$ be a bipartite graph, where $V=\set{v_i}_{i=0}^{m-1}$, $W=\set{w_i}_{i=0}^{n-1}$ and $E\subseteq V\times W$. The \emph{lexicographically first maximal matching} (lfm-matching) is the matching produced by successively matching each vertex $v_{0},\ldots,v_{m-1}$ to the least vertex available  in $W$ (see Fig. \ref{fig:lfmmatching} for an example).  We refer to $V$ as the set of bottom nodes and $W$ as the set of top nodes. 

In this section we will show that two decision problems concerning the
lfm-matching of  a bipartite graph are $\CC$-complete under $\AC^{0}$ many-one reductions.
The lfm-matching  problem ($\LFMM$) is to decide if a designated
edge belongs to the lfm-matching of  a bipartite graph $G$.
The vertex version of lfm-matching problem ($\VLFMM$) is to decide if a designated top node is matched in the lfm-matching of a bipartite graph $G$. $\LFMM$ is the usual way to define a decision problem for  lfm-matching as seen in \cite{MS92,Sub94}; however, as shown in Sections \ref{sec:c2vl} and \ref{sec:vl2cc}, the $\VLFMM$ problem is even more closely related to the $\CCV$ problem.

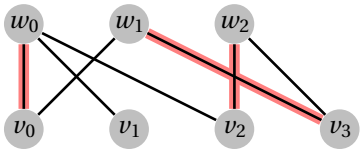
\begin{figure}
\tikzstyle{vertex}=[circle,fill=black!25,minimum size=15pt,inner sep=0pt]
\tikzstyle{edge} = [draw,line width=1pt,-]
\tikzstyle{cedge} = [draw,line width=4pt,-,red!50]
\centering
\begin{tikzpicture}[scale=0.7, auto,swap]
    % First we draw the vertices
    \foreach \pos/\name/\label in {{(0,2)/a/w_{0}}, {(2,2)/b/w_{1}}, {(4,2)/c/w_{2}},
                            {(0,0)/x/v_{0}}, {(2,0)/y/v_{1}}, {(4,0)/z/v_{2}}, {(6,0)/w/v_{3}}}
        \node[vertex] (\name) at \pos {$\label$};

    \foreach \source/ \dest  in {x/a,z/c,w/b}
        \path[cedge] (\source) -- (\dest);

    % Connect vertices with edges 
    \foreach \source/ \dest  in {x/a,x/b,y/a,z/a,z/c,w/b,w/c}
        \path[edge] (\source) -- (\dest);

\end{tikzpicture}
\caption{The thick edges form the lfm-matching of the above bipartite graph.}
\label{fig:lfmmatching}
\end{figure}

%In this section we will show that two decision problems concerning the 

%lfm-matching of  a bipartite graph are $\CC$-complete under $\AC^{0}$ many-one reductions.   
%The lfm-matching  problem ($\LFMM$) is to decide if a designated 
%edge belongs to the lfm-matching of  a bipartite graph $G$.
%The vertex version of lfm-matching problem ($\VLFMM$) is to decide if a designated top node is matched in the lfm-matching of a bipartite graph $G$. $\LFMM$ is the usual way to define a decision problem for  lfm-matching as seen in \cite{MS92,Sub94}; however, as shown in Sections \ref{sec:c2vl} and \ref{sec:vl2cc}, the $\VLFMM$ problem is even more closely related to the $\CCV$ problem.

We will show that the following two more restricted lfm-matching problems  are also $\CC$-complete.  We define $\TLFMM$ to be the restriction of $\LFMM$ to bipartite graphs of degree at most three. We define $\TVLFMM$ to be the restriction of $\VLFMM$ to bipartite graphs of degree at most three. 

To show that the problems defined above are equivalent under  $\AC^{0}$ many-one reductions, it turns out that we also need the following intermediate problem.  A negation gate flips the value on a wire.  For comparator circuits with negation gates, we allow negation gates to appear on any wire
(see the left diagram of \figref{cn2c} below for an example). 
The comparator circuit value problem with negation gates ($\CCVN$) is: given a comparator circuit with negation gates and input assignment, and a designated wire, decide if that wire outputs 1.

All reductions in this section are summarized in \figref{redu}.
\ifslow
\begin{figure}
\begin{center}
        \tikzstyle{vertex}=
        [%
          	minimum size=5mm,%
          	rectangle,%
	        	thick,%
		text centered
        ]

       \tikzstyle{edge} = [draw,line width=1pt,->]
      \begin{tikzpicture}[>=stealth',scale=1.5]

    	\foreach \pos/\label/\name in 
			{{(0,1)/\VLFMM/vl}, {(2,0)/\CCV/c}, {(4,1)/\TVLFMM/tv},
			{(0,-1)/\CCVN/cvn}, {(2,-2)/\LFMM/el}, {(4,-1)/\TLFMM/te}}
        		\node[vertex] (\name) at \pos {$\label$};
    
    	\foreach \source/ \dest/\label  in {{vl/c/vl2cc},{c/tv/c2vl},{c/te/c2te},{el/cvn/el2cvn},{cvn/c/cvn2c}}
       		\path [edge] (\source) -- node [circle,below,minimum size=5mm,inner sep=1pt] {\S\ref{sec:\label}} (\dest) ;

    	\foreach \source/ \dest  in {{tv/vl},{te/el}}
       		\path[edge] (\source) -- (\dest);

      \end{tikzpicture}
\end{center}

\caption{The label of an arrow denotes the section in which the reduction is described. Arrows without labels denote trivial reductions. All six problems are $\CC$-complete.}
\label{fig:redu}
\end{figure}
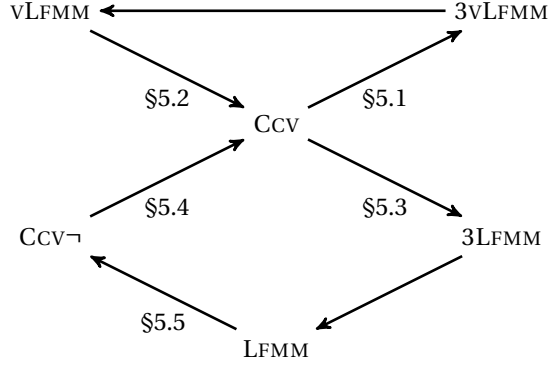
\fi

\subsection{\texorpdfstring{$\CCV  \mred^{\AC^{0}}  \TVLFMM$}{Ccv is AC0 many-one reducible to 3vLfmm} \label{sec:c2vl}}

By \propref{p1} it suffices to consider only instances of $\CCV$
in which all comparator gates point upward. We will show that these instances of $\CCV$ are $\AC^{0}$ many-one reducible to instances of $\TVLFMM$, which consist of bipartite graphs with \emph{degree at most three}.

The key observation is that a comparator gate on the left below closely relates to an instance of $\TVLFMM$ on the right. We use the top nodes $p_{0}$ and $q_{0}$ to represent the values $p_{0}$ and $q_{0}$ carried by the wires $x$ and $y$ respectively before the comparator gate, and the nodes $p_{1}$ and $q_{1}$ to represent the values of $x$ and $y$ after the comparator gate, where a top node is matched iff its respective value is one.
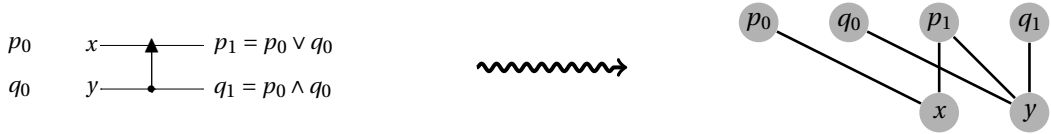
\ifslow
\begin{center}%\vspace{-5mm}
\begin{minipage}{0.4\textwidth}
\centering
\small 
\hspace{-2cm}$\Qcircuit @C=2.5em @R=1.3em {
\push{\bit{p_{0}}}&\push{x\,} & \ctrlu{1}	& \rstick{p_{1} = p_{0}\vee q_{0}}\qw\\
\push{\bit{q_{0}}}&\push{y\,} & \ctrl{-1} & \rstick{q_{1} = p_{0}\wedge q_{0}}\qw  
}$
\end{minipage}
\begin{minipage}{0.2\textwidth}
\centering
\begin{tikzpicture}
\draw[-to,line width=1.5pt,snake=snake,segment amplitude=.5mm,
         segment length=2mm,line after snake=1mm]  (0,0) -- (2,0);
\end{tikzpicture}
\end{minipage}
\begin{minipage}{0.35\textwidth}
\centering
\tikzstyle{vertex}=[circle,fill=black!30,minimum size=15pt,inner sep=0pt]
\tikzstyle{edge} = [draw,line width=1pt,-]
\begin{tikzpicture}[scale=0.6, auto,swap]
    % First we draw the vertices
    \foreach \pos/\label/\name in {{(0,2)/p_{0}/p0}, {(2,2)/q_{0}/q0}, {(4,2)/p_{1}/p1},{(6,2)/q_{1}/q1},
                            {(4,0)/x/x}, {(6,0)/y/y}}
        \node[vertex] (\name) at \pos {$\label$};
        
    % Connect vertices with edges 
    \foreach \source/ \dest  in {x/p0,x/p1,y/q0,y/p1,y/q1}
        \path[edge] (\source) -- (\dest);

\end{tikzpicture}

\end{minipage}
\end{center}
\fi
If nodes $p_{0}$ and $q_{0}$ have not been previously matched, i.e. $p_{0}=q_{0}=0$ in the comparator circuit, then the edges $\seq{x,p_{0}}$ and $\seq{y,q_{0}}$ are added to the lfm-matching. So the nodes $p_{1}$ and $q_{1}$ are not matched.  If $p_{0}$ has been previously matched, but $q_{0}$ has not, then edges $\seq{x,p_{1}}$ and $\seq{y,q_{0}}$ are added to the lfm-matching. So the node $p_{1}$ will be matched but $q_{1}$ will remain  unmatched. The other two cases are similar.

Thus, we can reduce a comparator circuit to the bipartite graph of an $\TVLFMM$ instance  by  converting each  comparator gate into the ``gadget'' described above. We will describe our method through an example, where we are given the comparator circuit in \figref{cex}. 

\ifslow
\begin{wrapfigure}{r}{0.3\textwidth}
\vspace{-.2cm}
\centering
{\small
$\Qcircuit @C=2.2em @R=0.7em {
\push{\bit{0}}&\push{a\,} & \ctrlu{1}	& \qw&    \rstick{\bit{1}} \qw \\
\push{\bit{1}}&\push{b\,} & \ctrl{-1} & \ctrlu{1}&   \rstick{\bit{1}} \qw  \\
\push{\bit{1}}&\push{c\,} & \qw &	\ctrl{-1}&   \rstick{\bit{0}} \qw  \\
		&	    &\lstick{0} & \lstick{1} &\lstick{2}\\
}$}
\vspace{-.15cm}
\caption{}
\label{fig:cex}
\end{wrapfigure}
\fi

We divide the comparator circuit into vertical layers 0, 1, 2 as shown in \figref{cex}. Since the  circuit has three wires $a$, $b$, $c$, for each layer $i$, we use six nodes, including three top nodes $a_{i}$, $b_{i}$ and $c_{i}$ representing the values of the wires $a$, $b$, $c$ respectively, and three bottom nodes $a'_{i}$, $b'_{i}$, $c'_{i}$, which are auxiliary nodes used to simulate the effect of the comparator gate at layer $i$.\\
\textbf{Layer 0:} This is the input layer, so we add an edge $\set{x_{i},x'_{i}}$ iff the wire $x$ takes input value 1. In this example, since $b$ and $c$ are wires taking input 1, we need to add the edges $\set{b_{0},b'_{0}}$ and $\set{c_{0},c'_{0}}$.
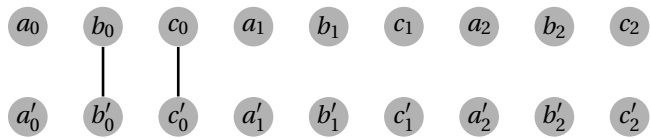
\ifslow
\begin{center}
\tikzstyle{vertex}=[circle,fill=black!30,minimum size=15pt,inner sep=0pt]
\tikzstyle{edge} = [draw,line width=1pt,-]
\begin{tikzpicture}[scale=1, auto,swap]
    % First we draw the vertices
    \foreach \i in {0,1,2}{
    	\node[vertex] (a\i) at (0+\i*3,1.2) {$a_{\i}$} ; 
	\node[vertex] (b\i) at (1+\i*3,1.2) {$b_{\i}$} ; 
	\node[vertex] (c\i) at (2+\i*3,1.2) {$c_{\i}$} ; 
	
	\node[vertex] (a'\i) at (0+\i*3,0) {$a'_{\i}$} ; 
	\node[vertex] (b'\i) at (1+\i*3,0) {$b'_{\i}$} ; 
	\node[vertex] (c'\i) at (2+\i*3,0) {$c'_{\i}$} ; 
    }
        
    % Connect vertices with edges 
    \foreach \source/ \dest  in {b'0/b0,c'0/c0}
        \path[edge] (\source) -- (\dest);

\end{tikzpicture}
\end{center}
\fi
\textbf{Layer 1:} We then add the gadget simulating the comparator gate from wire $b$ to wire $a$ as follows.
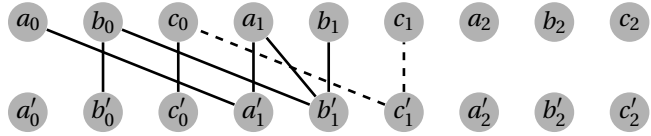
\ifslow
\begin{center}
\tikzstyle{vertex}=[circle,fill=black!30,minimum size=15pt,inner sep=0pt]
\tikzstyle{edge} = [draw,line width=1pt,-]
\tikzstyle{dedge} = [draw,dashed,line width=1pt,-]
\begin{tikzpicture}[scale=1, auto,swap]
    % First we draw the vertices
    \foreach \i in {0,1,2}{
    	\node[vertex] (a\i) at (0+\i*3,1.2) {$a_{\i}$} ; 
	\node[vertex] (b\i) at (1+\i*3,1.2) {$b_{\i}$} ; 
	\node[vertex] (c\i) at (2+\i*3,1.2) {$c_{\i}$} ; 
	
	\node[vertex] (a'\i) at (0+\i*3,0) {$a'_{\i}$} ; 
	\node[vertex] (b'\i) at (1+\i*3,0) {$b'_{\i}$} ; 
	\node[vertex] (c'\i) at (2+\i*3,0) {$c'_{\i}$} ; 
    }
        
    % Connect vertices with edges 
    \foreach \source/ \dest  in {b'0/b0,c'0/c0,a'1/a0,a'1/a1,b'1/b0,b'1/a1,b'1/b1}
        \path[edge] (\source) -- (\dest);

    \foreach \source/ \dest  in {c'1/c0,c'1/c1}
        \path[dedge] (\source) -- (\dest);

\end{tikzpicture}
\end{center}
\fi
Since the value of wire $c$ does not change when going from layer 0 to layer 1, we can simply propagate the value of $c_{0}$ to $c_{1}$ using the pair of dashed edges in the picture.\\
\textbf{Layer 2:} We proceed very similarly to layer 1 to get the following bipartite graph.
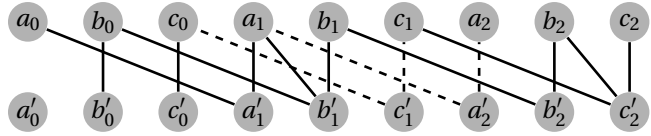
\ifslow
\begin{center}
\tikzstyle{vertex}=[circle,fill=black!30,minimum size=15pt,inner sep=0pt]
\tikzstyle{edge} = [draw,line width=1pt,-]
\tikzstyle{dedge} = [draw,dashed,line width=1pt,-]
\begin{tikzpicture}[scale=1, auto,swap]
    % First we draw the vertices
    \foreach \i in {0,1,2}{
    	\node[vertex] (a\i) at (0+\i*3,1.2) {$a_{\i}$} ; 
	\node[vertex] (b\i) at (1+\i*3,1.2) {$b_{\i}$} ; 
	\node[vertex] (c\i) at (2+\i*3,1.2) {$c_{\i}$} ; 
	
	\node[vertex] (a'\i) at (0+\i*3,0) {$a'_{\i}$} ; 
	\node[vertex] (b'\i) at (1+\i*3,0) {$b'_{\i}$} ; 
	\node[vertex] (c'\i) at (2+\i*3,0) {$c'_{\i}$} ; 
    }
        
    % Connect vertices with edges 
    \foreach \source/ \dest  in 	{b'0/b0,c'0/c0,a'1/a0,
    						a'1/a1,b'1/b0,b'1/a1,b'1/b1,
						b'2/b1,b'2/b2,c'2/c1,c'2/c2,c'2/b2}
        \path[edge] (\source) -- (\dest);

    \foreach \source/ \dest  in {c'1/c0,c'1/c1,a'2/a1,a'2/a2}
        \path[dedge] (\source) -- (\dest);

\end{tikzpicture}
\end{center}
\fi
Finally, we can get the output values of the comparator circuit by looking at the ``output'' nodes $a_{2},b_{2},c_{2}$ of this bipartite graph. We can easily check that $a_{2}$ is the only node that remains unmatched, which corresponds exactly to the only zero produced by wire $a$ of the comparator circuit in \figref{cex}.

It remains to argue that the construction above is an $\AC^0$ many-one reduction. We observe that each gate in the comparator circuit can be independently reduced to exactly one gadget in the bipartite graph that simulates the effect of the comparator gate; furthermore, the position of each gadget can be easily calculated from the position of each gate in the comparator circuit using very simple arithmetic.

\subsection{\texorpdfstring{$\VLFMM \mred^{\AC^{0}} \CCV$}{vLfmm is AC0 many-one reducible to Ccv} \label{sec:vl2cc}}
Consider the instance of $\VLFMM$  consisting of the bipartite graph in \figref{l2c}. Recall that we find the lfm-matching by matching the bottom nodes $x, y, z$ successively  to the first available node on the top. Hence we can simulate the matching of the bottom nodes to the top nodes using the comparator circuit on the right of \figref{l2c}, where we can think of the moving of a 1, say from wire $x$ to wire $a$, as the matching of node $x$ to node $a$ in the original bipartite graph. In this construction, a top node is matched iff its corresponding wire in the circuit outputs 1.
\ifslow
\begin{figure}[!h]
\centering
\begin{minipage}{0.3\textwidth}
\tikzstyle{vertex}=[circle,fill=black!30,minimum size=15pt,inner sep=0pt]
\tikzstyle{edge} = [draw,line width=1pt,-]
\tikzstyle{cedge} = [draw,line width=4pt,-,red!50]
\begin{tikzpicture}[scale=0.6, auto,swap]
    % First we draw the vertices
    \foreach \pos/\name in {{(0,2)/a}, {(2,2)/b}, {(4,2)/c},{(6,2)/d},
                            {(0,0)/x}, {(2,0)/y}, {(4,0)/z}}
        \node[vertex] (\name) at \pos {$\name$};
        
    % Connect vertices with edges 
    \foreach \source/ \dest  in {x/a,x/b,x/c,y/a,y/c,z/b,z/d}
        \path[edge] (\source) -- (\dest);

    \begin{pgfonlayer}{background}
    \foreach \source/ \dest  in {x/a,y/c,z/b}
        \path[cedge] (\source) -- (\dest);
    \end{pgfonlayer}

\end{tikzpicture}
\end{minipage}
\begin{minipage}{0.18\textwidth}
\centering
\begin{tikzpicture}
\draw[-to,line width=1.5pt,snake=snake,segment amplitude=.5mm,
         segment length=2mm,line after snake=1mm]  (0,0) -- (1.5,0);
\end{tikzpicture}
\end{minipage}
\begin{minipage}{0.45\textwidth}
\centering
{\small$\Qcircuit @C=1.4em @R=0.5em {
\push{\bit{0}}&\push{a\,} & \ctrlu{4}	& \qw 	&\qw		&\qw		&\ctrlu{5}	& \qw 	& \qw	& \qw	& \qw	& \qw	& \qw	& \qw &\rstick{\bit{1}} \qw \\
\push{\bit{0}}&\push{b\,} &	\qw		& \ctrlu{3} &\qw  	&\qw		&\qw		&\qw		&\qw		&\qw		& \qw	& \ctrlu{5}	& \qw	& \qw &\rstick{\bit{1}} \qw\\
\push{\bit{0}}&\push{c\,} &	\qw		& \qw  	& \ctrlu{1}  	&\qw		&\qw 	&\qw		&\ctrlu{3} 	&\qw		&\qw		&\qw		&\qw		&\qw&\rstick{\bit{1}} \qw\\
\push{\bit{0}}&\push{d\,} & \qw	 	& \qw 	&\qw 	&\qw		&\qw		&\qw		&\qw 	&\qw		& \qw	& \qw 	& \qw	& \ctrlu{3}&\rstick{\bit{0}} \qw \\
\push{\bit{1}}&\push{x\,} & \ctrl{-1} 	& \ctrl{-1}  & \ctrl{-1} 	&\ctrl{0}	& \qw	&\qw		&\qw 	&\qw		&\qw		&\qw		& \qw	& \qw&    \rstick{\bit{0}} \qw \\
\push{\bit{1}}&\push{y\,} & \qw 		& \qw 	&\qw		&\qw		& \ctrl{-1}  &\ctrl{0}	&\ctrl{-1}	&\ctrl{0}	& \qw	& \qw 	& \qw	& \qw&\rstick{\bit{0}} \qw \\
\push{\bit{1}}&\push{z\,} &	\qw		& \qw  	&\qw  	&\qw		& \qw	&\qw		&\qw 	&\qw		&\ctrl{0}	&\ctrl{-1}	& \ctrl{0}	&\ctrl{-1}&\rstick{\bit{0}} \qw
\gategroup{1}{3}{7}{6}{1em}{-}
\gategroup{1}{7}{7}{10}{1em}{-}
\gategroup{1}{11}{7}{14}{1em}{-}
}$}
\end{minipage}
\caption{}
\label{fig:l2c}
%\vspace{-3mm}
\end{figure}
\fi

Note that we draw bullets without any arrows going out from them in the circuit to  denote dummy gates, which do nothing. These dummy gates are introduced for the following technical reason. Since the bottom nodes might not have the same degree, the position of a comparator gate really depends on the structure of the graph, which makes it harder to give a direct $\AC^{0}$ reduction. By using dummy gates, we can treat the graph as if it is a complete bipartite graph, the missing edges represented by dummy gates. This can easily be shown to be an $\AC^0$ reduction from $\VLFMM$ to $\CCV$. Together with the reduction from \secref{c2vl}, we get the following theorem.

\begin{theorem}\label{theo:lfcom} 
The problems $\CCV$, $\TVLFMM$  and  $\VLFMM$ are equivalent under $\AC^0$ many-one reductions.
\end{theorem}

\subsection{\texorpdfstring{$\CCV  \mred^{\AC^{0}}  \TLFMM$}{Ccv is AC0 many-one reducible to 3Lfmm} \label{sec:c2te}}
We start by applying the reduction  $\CCV  \mred^{\AC^{0}}  \TVLFMM$
of \secref{c2vl} to get an instance of $\TVLFMM$, and notice that 
the degrees of the top ``output'' nodes of the resulting bipartite
graph, e.g. the nodes $a_{2},b_{2},c_{2}$ in the example of \secref{c2vl},
have degree at most two.
Now we show how to reduce such instances of $\TVLFMM$ (i.e. those
whose designated top vertices have degree at most two) to $\TLFMM$.
Consider the graph $G$ with degree at most three and a designated top
vertex $b$ of degree two as shown on the left of \figref{vl2el}. We extend it to a bipartite graph $G'$ by adding an additional top node $w_t$ and an additional bottom node $w_b$, alongside two edges $\set{b, w_b}$ and $\set{w_t, w_b}$, as shown in \figref{vl2el}. Observe that the degree of the new graph $G'$ is at most three.
\ifslow
\begin{figure}[!h]
\begin{center}%\vspace{-5mm}
\begin{minipage}{0.3\textwidth}
\centering 
\tikzstyle{vertex}=[circle,fill=black!30,minimum size=15pt,inner sep=0pt]
\tikzstyle{edge} = [draw,line width=1pt,-]
\tikzstyle{cedge} = [draw,line width=4pt,-,black!35]
\tikzstyle{dedge} = [draw,dashed,line width=1pt,-]
\begin{tikzpicture}[scale=0.6, auto,swap]
    % First we draw the vertices
    \foreach \pos/\name in {{(0,2)/a}, {(2,2)/b}, {(4,2)/c},
                            {(0,0)/x}, {(2,0)/y}, {(4,0)/z}}
        \node[vertex] (\name) at \pos {$\name$};
        
    % Connect vertices with edges 
    \foreach \source/ \dest  in {x/a,x/b,x/c,y/a,y/c,z/b}
        \path[edge] (\source) -- (\dest);

\end{tikzpicture}
\end{minipage}
\begin{minipage}{0.2\textwidth}
\centering
\begin{tikzpicture}
\draw[-to,line width=1.5pt,snake=snake,segment amplitude=.5mm,
         segment length=2mm,line after snake=1mm]  (0,0) -- (2,0);
\end{tikzpicture}
\end{minipage}
\begin{minipage}{0.4\textwidth}
\centering
\tikzstyle{vertex}=[circle,fill=black!30,minimum size=15pt,inner sep=0pt]
\tikzstyle{edge} = [draw,line width=1pt,-]
\tikzstyle{dedge} = [draw,dashed,line width=1pt,-]

\begin{tikzpicture}[scale=0.6, auto,swap]
    % First we draw the vertices
    \foreach \pos/\name in {{(0,2)/a}, {(2,2)/b}, {(4,2)/c}, {(6,2)/w_t}, 
    					{(0,0)/x}, {(2,0)/y}, {(4,0)/z},{(6,0)/w_b}}
        \node[vertex] (\name) at \pos {$\name$};
        
    % Connect vertices with edges 
    \foreach \source/ \dest  in {x/a,x/b,x/c,y/a,y/c,z/b}
        \path[edge] (\source) -- (\dest);
        
    \foreach \source/ \dest  in {{w_b/b},{w_b/w_t}}
        \path[dedge] (\source) -- (\dest);

\end{tikzpicture}
\end{minipage}
\end{center}
\caption{}
\label{fig:vl2el}
\end{figure}
\fi

We treat the resulting bipartite graph $G'$ and the edge $\set{w_{t},w_{b}}$ as an instance of $\TLFMM$. It is not hard to see that the vertex $b$ is matched in the lfm-matching of the original bipartite graph $G$ iff the edge $\set{w_t, w_b}$ is in the lfm-matching of the new bipartite graph $G'$.

\subsection{\texorpdfstring{$\CCVN \mred^{\AC^{0}}  \CCV$}{Ccv with negation gates is AC0 many-one reducible to Ccv} \label{sec:cvn2c}}
Recall that a comparator circuit value problem with negation gates ($\CCVN$) is the task of deciding,  given a comparator circuit with negation gates  and an input assignment, whether a designated wire outputs one. 
It should be clear that $\CCV$ is a special case of $\CCVN$ and hence $\AC^{0}$ many-one reducible to $\CCVN$. Here, we show the nontrivial direction that $\CCVN \mred^{\AC^{0}}  \CCV$. Our proof is based on Subramanian's idea from \cite{Sub94}.

\ifslow
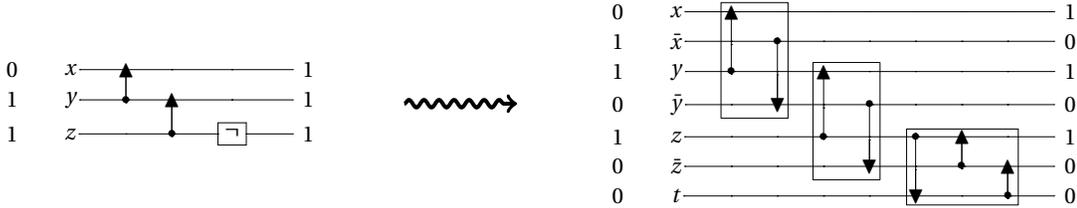
\begin{figure}[!h]
\begin{center}%\vspace{-5mm}
\begin{minipage}{0.3\textwidth}
\centering
\small 
$\Qcircuit @C=2.2em @R=0.7em {
\push{\bit{0}}&\push{x\,} & \ctrlu{1}	& \qw&  \qw  &\rstick{\bit{1}} \qw \\
\push{\bit{1}}&\push{y\,} & \ctrl{-1} & \ctrlu{1}&  \qw &\rstick{\bit{1}} \qw  \\
\push{\bit{1}}&\push{z\,} & \qw &	\ctrl{-1}&   \gate{\neg}&\rstick{\bit{1}} \qw  
}$
\end{minipage}
\begin{minipage}{0.2\textwidth}
\centering
\begin{tikzpicture}
\draw[-to,line width=1.5pt,snake=snake,segment amplitude=.5mm,
         segment length=2mm,line after snake=1mm]  (0,0) -- (1.5,0);
\end{tikzpicture}
\end{minipage}
\begin{minipage}{0.4\textwidth}
\centering
\small 
$\Qcircuit @C=2.2em @R=0.7em {
\push{\bit{0}}&\push{x\,}  	& \ctrlu{1}	& \qw	& \qw	& \qw	&  \qw 	& \qw	 & \qw	&\rstick{\bit{1}} \qw \\
\push{\bit{1}}&\push{\bar{x}\,} 	& \qw	& \ctrl{2}	& \qw	& \qw	&  \qw  	& \qw	& \qw	&\rstick{\bit{0}} \qw \\
\push{\bit{1}}&\push{y\,} 	& \ctrl{-1} 	& \qw	& \ctrlu{1}	& \qw	&  \qw 	& \qw	& \qw	&\rstick{\bit{1}} \qw  \\
\push{\bit{0}}&\push{\bar{y}\,} 	& \qw	& \ctrld{-1}	& \qw	& \ctrl{2}	&  \qw  	& \qw	& \qw	&\rstick{\bit{0}} \qw \\
\push{\bit{1}}&\push{z\,} 	& \qw 	& \qw	&\ctrl{-1}	& \qw	&  \ctrl{2}	& \ctrlu{1}	&\qw	&\rstick{\bit{1}} \qw \\
\push{\bit{0}}&\push{\bar{z}\,} 	& \qw	& \qw	& \qw	&\ctrld{-1}	&  \qw  	& \ctrl{-1}	&\ctrlu{1}&\rstick{\bit{0}} \qw \\ 
\push{\bit{0}}&\push{t\,} 	& \qw	& \qw	& \qw	&\qw		& \ctrld{-1}& \qw & \ctrl{-1}	&\rstick{\bit{0}} \qw  
\gategroup{1}{3}{4}{4}{1em}{-}
\gategroup{3}{5}{6}{6}{1em}{-}
\gategroup{5}{7}{7}{9}{1em}{-}
}$
\end{minipage}
\end{center}
\caption{Successive gates on the left circuit correspond to successive boxes of gates on the right circuit.}
\label{fig:cn2c}
\end{figure}
\fi

%%% version with "switch"
%\begin{minipage}{0.4\textwidth}
%\centering
%\small 
%$\Qcircuit @C=2.2em @R=0.7em {
%\push{\bit{0}}&\push{x\,}  	& \ctrlu{1}	& \qw	& \qw	& \qw	&  \qw 	&\rstick{\bit{1}} \qw \\
%\push{\bit{1}}&\push{\bar{x}\,} 	& \qw	& \ctrl{2}	& \qw	& \qw	&  \qw 	&\rstick{\bit{0}} \qw \\
%\push{\bit{1}}&\push{y\,} 	& \ctrl{-1} 	& \qw	& \ctrlu{1}	& \qw	&  \qw 	&\rstick{\bit{1}} \qw  \\
%\push{\bit{0}}&\push{\bar{y}\,} 	& \qw	& \ctrld{-1}	& \qw	& \ctrl{2}	&  \qw	&\rstick{\bit{0}} \qw \\
%\push{\bit{1}}&\push{z\,} 	& \qw 	& \qw	&\ctrl{-1}	& \qw	&\qswap	&\rstick{\bit{1}} \qw \\
%\push{\bit{0}}&\push{\bar{z}\,} 	& \qw	& \qw	& \qw	&\ctrld{-1}	&  \qswap \qwx  & \rstick{\bit{0}} \qw
%\gategroup{1}{3}{4}{4}{1em}{-}
%\gategroup{3}{5}{6}{6}{1em}{-}
%}$
%\end{minipage}
%%% version with "switch"

The reduction is based on ``double-rail'' logic, which can be traced to Goldschlager's proof of the $\P$-completeness of the monotone circuit value problem~\cite{Gol77}.
Given an instance of $\CCVN$ consisting of a comparator circuit with negation gates $C$ with its input $I$ and a designated wire $s$, we construct  an instance of $\CCV$ consisting of a comparator circuit $C'$ with its input $I'$ and a designated wire $s'$ as follows.
For every wire $w$ in $I$ we put in two corresponding wires, $w$ and
$\overline{w}$, in $C'$.  We define the input $I'$ of $C'$ such that the input value
of $\overline{w}$ is the negation of the input value of $w$. We want to
fix things so that the value carried by the wire $\overline{w}$ at each layer
is always the negation of the value carried by $w$. For any comparator
gate $\seq{y,x}$ in $C$ we put in $C'$ the gate $\seq{y,x}$ followed by the
gate $\seq{\overline{x},\overline{y}}$.
It is easy to check using De Morgan's laws that the wires $x$ and $y$ in
$C'$ carry the corresponding values of $x$ and $y$ in $C$, and the
wires $\overline{x}$ and $\overline{y}$ in $C'$ carry the negations
of the wires $x$ and $y$ in $C$.

%In order to simulate a negation gate on a wire $x$ in $C$, we simply switch the wires $x$ and $\overline{x}$ in $C'$.

The circuit $C'$ has one extra wire $t$ with input value $0$ to help in
translating negation gates.  For each negation gate on a wire, says $z$
in the example from \figref{cn2c}, we add three comparator gates
$\seq{z,t}$, $\seq{\overline{z},z}$,  $\seq{t,\overline{z}}$ as shown in the right
circuit of \figref{cn2c}. Thus $t$ as a
temporary ``container'' that we use to swap the values carried by the
wires $z$ and $\overline{z}$. Note that the swapping of values of $z$ and $\overline{z}$ in $C'$ simulates the effect of a negation in $C$.  Also note that
after the swap takes place, the value of $t$ is restored to $0$. (The more straightforward solution of simply switching the wires $z$ and $\overline{z}$ does not result in an $\AC^0$ many-one reduction.)

Finally note that the output value of the designated wire $s$ in
$C$ is 1 iff the output value of the
corresponding wire $s$ in $C'$ with input $I'$ is 1.  Thus we set
the designated wire $s'$ in $I'$ to be $s$.

\subsection{\texorpdfstring{$\LFMM \mred^{\AC^{0}} \CCVN$ }{Lfmm is AC0 many-one reducible to Ccv with negation gates}\label{sec:el2cvn}}
Consider an instance of $\LFMM$  consisting of the bipartite graph on the left of \figref{el2c}, and a designated edge $\set{y,c}$. Without loss of generality, 
we can safely ignore all top vertices occurring after $c$, all bottom vertices occurring after $y$, and all the edges associated with them, since they are not going to affect the outcome of the instance. 
Using the construction from \secref{vl2cc},  we can simulate the matching of the bottom nodes to the top nodes using the  comparator circuit in the upper box on the right of \figref{el2c}.
\ifslow
\begin{figure}[!h]
\centering
\begin{minipage}{0.25\textwidth}
\centering
\tikzstyle{vertex}=[circle,fill=black!30,minimum size=15pt,inner sep=0pt]
\tikzstyle{edge} = [draw,line width=1pt,-]
\tikzstyle{cedge} = [draw,line width=4pt,-,red!50]
\begin{tikzpicture}[scale=0.6, auto,swap]
    % First we draw the vertices
    \foreach \pos/\label/\name in {{(0,2)/a/a}, {(2,2)/b/b}, {(4,2)/c/c},
                            {(0,0)/x/x}, {(2,0)/y/y}}
        \node[vertex] (\name) at \pos {$\label$};
        
    % Connect vertices with edges 
    \foreach \source/ \dest  in {x/a,x/b,y/a,y/c}
        \path[edge] (\source) -- (\dest);
        
    \begin{pgfonlayer}{background}
    \foreach \source/ \dest  in {x/a,y/c}
        \path[cedge] (\source) -- (\dest);
    \end{pgfonlayer}
    
\end{tikzpicture}
\end{minipage}
\begin{minipage}{0.18\textwidth}
\centering
\begin{tikzpicture}
\draw[-to,line width=1.5pt,snake=snake,segment amplitude=.5mm,
         segment length=2mm,line after snake=1mm]  (0,0) -- (1.5,0);
\end{tikzpicture}
\end{minipage}
\begin{minipage}{0.45\textwidth}
\centering
{\small$\Qcircuit @C=2em @R=0.5em {
\push{\bit{0}}&\push{a\,} & \ctrlu{3}	& \qw 	&\qw		&\ctrlu{4}	&\qw		&\qw			&\qw		&\rstick{\bit{1}} \qw \\
\push{\bit{0}}&\push{b\,} &	\qw		&\ctrlu{2} 	&\qw		&\qw		&\qw 	&\qw 		&\qw		&\rstick{\bit{0}} \qw\\
\push{\bit{0}}&\push{c\,} & \qw 		& \qw 	&\qw		&\qw		&\qw		&\ctrlu{2}		&\ctrlu{5}	&\rstick{\bit{1}} \qw \\
\push{\bit{1}}&\push{x\,} & \ctrl{-1} 	& \ctrl{-1}  &\ctrl{0}	&\qw		& \qw         &\qw			& \qw	&\rstick{\bit{0}} \qw \\
\push{\bit{1}}&\push{y\,} & \qw 		& \qw	&\qw 	&\ctrl{-1}	&\ctrl{0}	&\ctrl{-1}		& \qw	&\rstick{\bit{0}} \qw \\
\push{\bit{0}}&\push{a'\,} &\ctrlu{3}	&\qw 	&\qw		&\ctrlu{4}	&\qw  	&\qw			& \qw	&\rstick{\bit{1}} \qw\\
\push{\bit{0}}&\push{b'\,} &\qw		&\ctrlu{2}  	&\qw		&\qw		&\qw 	&\qw			&\qw 	&\rstick{\bit{0}} \qw\\
\push{\bit{0}}&\push{c'\,} & \qw 		& \qw 	&\qw		&\qw		&\qw 	&\gate{\neg}	&\ctrl{-1}	&\rstick{\bit{1}} \qw \\
\push{\bit{1}}&\push{x'\,} & \ctrl{-1} 	& \ctrl{-1}	&\ctrl{0}	&\qw		& \qw	&\qw			& \qw	&\rstick{\bit{0}} \qw \\
\push{\bit{1}}&\push{y'\,} & \qw 		& \qw	&\qw 	&\ctrl{-1}	&\ctrl{0}	&\qw			& \qw	&\rstick{\bit{1}} \qw
\gategroup{1}{3}{5}{8}{1em}{-}
\gategroup{6}{3}{10}{7}{1em}{-}
}$}
\end{minipage}
\caption{}
\label{fig:el2c}
\end{figure}
\fi

We keep another running copy of this simulation on the bottom
(see the wires labelled $a',b',c',x',y'$ in \figref{el2c}). The only
difference  is that the comparator gate $\seq{y',c'}$ corresponding to
the designated edge $\set{y,c}$ is not added. Finally, we add a negation
gate on $c'$ and  a comparator gate $\seq{c',c}$.   We let the desired
output of the $\CCV$ instance be the output of $c$, since $c$ outputs 1
iff the edge $\set{y,c}$ is added to the lfm-matching. It  is not hard to generalize this construction to an arbitrary bipartite graph and designated edge. %see that such construction can be generalized, and the output correctly computes if the designated edge is in the lfm-matching. 

Combined with the constructions from Sections \ref{sec:c2vl} and \ref{sec:vl2cc}, we have the following corollary.

\begin{corollary}\label{cor:lfm}
  The problems $\CCV$, $\TVLFMM$, $\VLFMM$, $\CCVN$, $\TLFMM$  and  $\LFMM$ are equivalent under $\AC^0$ many-one reductions.
\end{corollary}

%Since $\CCVN$ is complete for $\CC$, we can use comparator circuits to decide the complement of the $\CCV$ problem: given a comparator circuit and and input assignment, does a designated wire output zero? Thus, we have the following corollary.
%\begin{corollary}[Subramanian \cite{MS92}] $\CC$ is closed under complementation.
%\end{corollary}

\section{The $\SMP$ problem is $\CC$-complete}
\label{sec:stable}

An instance of the stable marriage problem ($\SMP$), proposed by Gale and Shapley~\cite{GS62} in the context of college admissions, is given by a
number $n$ (specifying the number of men and the number of women),
together with a preference list for each man and each woman
specifying a total ordering on all people of the opposite sex.
The goal of $\SMP$ is to produce a perfect matching between men and women,
i.e., a bijection from the set of men to the set of women, such that the
following \emph{stability} condition is satisfied: there are no two
people of opposite sex who like each other more than their current
partners.  Such a stable solution always exists, but it may not be unique.
Thus $\SMP$ is a search problem rather than a decision problem.

However there is always a unique {\em man-optimal} and a unique
{\em woman-optimal} solution.   In the man-optimal solution each
man is matched with a woman whom he likes at least as well as any
woman that he is matched with in any stable solution.  Dually
for the woman-optimal solution.  Thus we define the
{\em man-optimal stable marriage decision problem} ($\MOSM$) as follows:
given an instance of the stable marriage problem together with
a designated man-woman pair, determine whether that pair is 
married in the man-optimal stable marriage.
We define the {\em woman-optimal stable marriage decision problem}
($\WOSM$) analogously.

%Thus both the man-optimal and
%the woman-optimal versions are function problems (and hence equivalent
%to decision problems.)

We show here that the search version and the decision versions are
computationally equivalent, and each is complete for $\CC$.
 %with respect to the appropriate reducibility in Definition \ref{d:manyOne}.
\secref{yuli} shows how to reduce the
lexicographically first maximal matching problem (which is complete for
$\CC$) to the $\SMP$ search problem, and Section~\ref{s:SMP-CCV}
shows how to reduce both the $\MOSM$ and $\WOSM$ problems to $\CCV$. 

\subsection{$\TLFMM$ is $\AC^{0}$ many-one reducible to $\SMP$, $\MOSM$ and $\WOSM$} \label{sec:yuli}

We start by showing that $\TLFMM$ is $\AC^0$ many-one reducible to
$\SMP$ when we regard both $\TLFMM$ and $\SMP$ as search problems.  (Of course the
lfm-matching is the unique solution to $\TLFMM$ formulated as a
search problem, but it is still a total search problem.)

Let $G=(V,W,E)$ be a bipartite graph  from an instance of  $\TLFMM$, where $V$ is the set of bottom nodes, $W$ is the set of top nodes, and $E$ is the edge relation such that the degree of each node is at most three (see the example in the figure on the left  below).  Without loss of generality, we can assume that $|V|=|W|=n$. To reduce it to an instance of $\SMP$, we  double the number of nodes in each partition, where the new nodes are enumerated after the original nodes and the original nodes are enumerated using the ordering of the original bipartite graph, as shown in the diagram on the right below. We also let the bottom nodes and top nodes represent the men and women respectively.
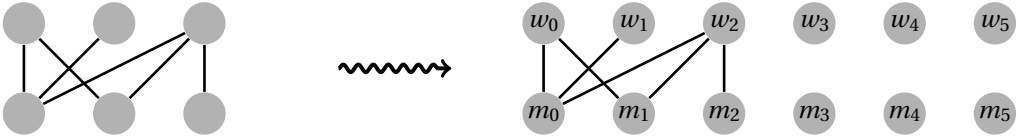
\ifslow
\begin{center}
\tikzstyle{vertex}=[circle,fill=black!30,minimum size=16pt,inner sep=0pt]
\tikzstyle{edge} = [draw,line width=1pt,-]
\begin{minipage}{0.25\textwidth}
\centering
\begin{tikzpicture}[scale=0.6, auto,swap]
    % First we draw the vertices
    \foreach \pos/\name in {{(0,2)/a}, {(2,2)/b}, {(4,2)/c},
    					{(0,0)/x}, {(2,0)/y}, {(4,0)/z}}
        \node[vertex] (\name) at \pos{};
        
    % Connect vertices with edges 
    \foreach \source/ \dest  in {x/a,x/b,x/c,y/a,y/c,z/c}
        \path[edge] (\source) -- (\dest);

\end{tikzpicture}
\end{minipage}
\begin{minipage}{0.2\textwidth}
\centering
\begin{tikzpicture}
\draw[-to,line width=1.5pt,snake=snake,segment amplitude=.5mm,
         segment length=2mm,line after snake=1mm]  (0,0) -- (1.5,0);
\end{tikzpicture}
\end{minipage}
\begin{minipage}{0.5\textwidth}
\begin{tikzpicture}[scale=0.6, auto,swap]
    % First we draw the vertices
    \foreach \i  in {0,...,5}{
    		\node[vertex](w\i) at (2*\i,2){$w_{\i}$};
		\node[vertex](m\i) at (2*\i,0){$m_{\i}$};
  	}
    % Connect vertices with edges 
    \foreach \source/ \dest  in {m0/w0,m0/w1,m0/w2,m1/w0,m1/w2,m2/w2}
        \path[edge] (\source) -- (\dest);

\end{tikzpicture}
\end{minipage}
\end{center}
\fi

It remains to define a preference list for each person in this $\SMP$ instance. The preference list of each man $m_{i}$, who represents a bottom node in the original graph, starts with all the women $w_{j}$ (at most three of them) adjacent to $m_{i}$ in the order that these women are enumerated, followed by all the women $w_{n},\ldots,w_{2n-1}$; the list ends with all women $w_{j}$ not adjacent to  $m_{i}$ also in the order that they are enumerated. For example, the preference list of $m_{2}$ in our example is  $w_{2},w_{3},w_{4},w_{5},w_{0},w_{1}$. The preference list of each newly introduced man $m_{n+i}$ simply consists of $w_{0},\ldots,w_{n-1},w_{n},\ldots,w_{2n-1}$, i.e., in the order that the top nodes are listed. Preference lists for the women are defined dually. 

%Note that we seem to require a bit of counting when constructing the preference lists for men and women; however, since the degree of each node is at most three, this counting can be easily done in $\AC^{0}$. 

Intuitively, the preference lists are constructed so that any stable marriage (not necessarily man-optimal) of the new $\SMP$ instance must contain  the lfm-matching of $G$. Furthermore, if a bottom node $u$ from the original graph is not matched to any top node in the lfm-matching of $G$, then the man $m_{i}$ representing $u$ will marry some top node $w_{n+j}$, which is a dummy node that does not correspond to any node of $G$. 

The above construction gives us a $\AC^{0}$ many-one reduction from $\TLFMM$ to $\SMP$ as search problems, any solution of a stable marriage instance constructed by the above reduction providing us all the information to decide whether an edge is in the lfm-matching of the original $\TLFMM$ instance. The key explanation is that every instance of stable marriage produced by the above reduction has a unique solution; thus the man-optimal solution coincides with the woman-optimal solution. Moreover,
the above construction also shows that the decision version of $\TLFMM$ is 
$\AC^{0}$ many-one reducible to either of the decision problems
$\MOSM$ and $\WOSM$. Hence we have proven the following theorem, whose detailed proof can be found in
\cite{CLY11}.

\begin{theorem}\label{theo:me}
  $\TLFMM$  is $\AC^{0}$ many-one reducible to $\SMP$, $\MOSM$ and $\WOSM$.
\end{theorem}

\subsection{$\TCV$ is $\CC$-complete \label{sec:ccvpi}}
In the remainder of the section, we will be occupied with developing an algorithm due to Subramanian~\cite{Sub90,Sub94} that finds a stable marriage using comparator circuits, thus furnishing an $\AC^0$ reduction from $\SMP$ to $\CCV$. To this end, it turns out to be conceptually simpler to go through a new variant of $\CCV$, where the wires are three-valued instead of Boolean.

We define the $\TCV$ problem similarly to $\CCV$, i.e., we want to decide, on a given input assignment, if a designated wire of a comparator circuit outputs one.  The only difference is that each wire can now take either value 0, 1 or $\ast$, where a wire takes value $\ast$ when its value is not known to be 0 or 1.  The output values of the comparator gate on two input values $p$ and $q$ will be defined as follows.
\begin{center}
$p\wedge q = \begin{cases}
0	& \text{if $p=0$ or $q=0$}\\
1	& \text{if $p=q=1$}\\
\ast	& \text{otherwise.} 
\end{cases}$
\hspace{2cm} 
$p\vee q = \begin{cases}
0	& \text{if $p=q=0$}\\
1	& \text{if $p=1$ or $q=1$}\\
\ast	& \text{otherwise.} 
\end{cases}$
\end{center}
Clearly every instance of $\CCV$ is also an instance of $\TCV$. We will show that every instance of $\TCV$ is $\AC^{0}$ many-one reducible to an instance of $\CCV$ by using a pair of Boolean wires to represent each three-valued wire and adding comparator gates appropriately to simulate three-valued comparator gates. 

\begin{theorem}\label{theo:tcv}
  $\TCV$ and $\CCV$ are equivalent under $\AC^0$ many-one reductions.
%The $\TCV$ problem is $\CC$-complete.
\end{theorem}

\begin{proof}Since each instance of $\CCV$ is a special case of $\TCV$, it only remains to show that every instance of $\TCV$ is $\AC^{0}$ many-one reducible to an instance of $\CCV$. 

First, we will describe a gadget built from standard comparator gates that simulates a three-valued comparator gate  as follows. Each wire of an instance of $\TCV$  will be represented by a pair of wires in an instance of $\CCV$. Each three-valued comparator gate on the left below, where $p,q,p\wedge q,p\vee q \in \set{0,1,\ast}$, can be simulated by  a gadget consisting of two standard comparator gates on the right below.
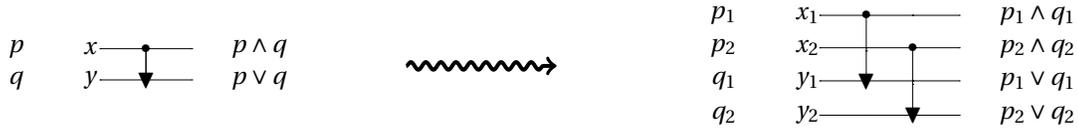
\ifslow
\begin{center}
\begin{minipage}{0.3\textwidth}
\centering
\hspace{-1.5cm}$\Qcircuit @C=2em @R=0.7em {
\lstick{\bit{p}}&\push{x\,} & \ctrl{1}	& \rstick{~~~~~\bit{p\wedge q}}\qw\\
\lstick{\bit{q}}&\push{y\,} & \ctrld{-1} & \rstick{~~~~~\bit{p\vee q}}\qw  
}$
\end{minipage}
\begin{minipage}{0.2\textwidth}
\centering
\begin{tikzpicture}
\draw[-to,line width=1.5pt,snake=snake,segment amplitude=.5mm,
         segment length=2mm,line after snake=1mm]  (0,0) -- (2,0);
\end{tikzpicture}
\end{minipage}
\begin{minipage}{0.4\textwidth}
\centering
$\Qcircuit @C=2em @R=.7em {
\lstick{\bit{p_1}}&\push{x_1} & \ctrl{1}	&\qw	& \rstick{~~~~~\bit{p_1\wedge q_1}}\qw\\
\lstick{\bit{p_2}}&\push{x_2} &\qw	& \ctrl{1}	& \rstick{~~~~~\bit{p_2\wedge q_2}}\qw\\
\lstick{\bit{q_1}}&\push{y_1} & \ctrld{-1} &\qw	& \rstick{~~~~~\bit{p_1\vee q_1}}\qw\\  
\lstick{\bit{q_2}}&\push{y_2} &\qw	& \ctrld{-1} & \rstick{~~~~~\bit{p_2\vee q_2}}\qw
}$
\end{minipage}
\end{center}
\fi

The wires $x$ and $y$ are represented using the two pairs of wires $\seq{x_{1},x_{2}}$ and $\seq{y_{1},y_{2}}$, and three possible values 0, 1 and $\ast$ will be encoded by $\seq{0,0}$, $\seq{1,1}$, and $\seq{0,1}$ respectively.  The fact that our gadget  correctly simulates the three-valued comparator gate is shown in the following table.
\begin{center}
\begin{tabular}{|c|c||c|c||c|c||c|c|}
\hline
$~p~$ & $~q~$ & $\seq{p_1,p_2}$ &  $\seq{q_1,q_2}$ & $p\wedge q$ &$p\vee q$ 
& $\seq{p_1\wedge q_1,p_2 \wedge q_2}$& $\seq{p_1\vee q_1,p_2 \vee q_2}$\\
\hline
\hline
0 & 0 & $\seq{0,0}$ &  $\seq{0,0}$ & 0 & 0 & $\seq{0,0}$& $\seq{0,0}$\\
\hline
0 & 1 & $\seq{0,0}$ &  $\seq{1,1}$ & 0 & 1 & $\seq{0,0}$& $\seq{1,1}$\\
\hline
0 & $\ast$ & $\seq{0,0}$ &  $\seq{0,1}$ & 0 & $\ast$ & $\seq{0,0}$& $\seq{0,1}$\\
\hline
1 & 0 & $\seq{1,1}$ &  $\seq{0,0}$ & 0 & 1 & $\seq{0,0}$& $\seq{1,1}$\\
\hline
1 & 1 & $\seq{1,1}$ &  $\seq{1,1}$ & 1 & 1 & $\seq{1,1}$& $\seq{1,1}$\\
\hline
1 & $\ast$ & $\seq{1,1}$ &  $\seq{0,1}$ & $\ast$ & 1 & $\seq{0,1}$& $\seq{1,1}$\\
\hline
$\ast$ & 0 & $\seq{0,1}$ &  $\seq{0,0}$ & 0 & $\ast$ & $\seq{0,0}$& $\seq{0,1}$\\
\hline
$\ast$ & 1 & $\seq{0,1}$ &  $\seq{1,1}$ & $\ast$ & 1 & $\seq{0,1}$& $\seq{1,1}$\\
\hline
$\ast$ & $\ast$ & $\seq{0,1}$ &  $\seq{0,1}$ & $\ast$ & $\ast$ & $\seq{0,1}$& $\seq{0,1}$\\
\hline
\end{tabular} 
\end{center}

Using this gadget, we can reduce an instance of $\TCV$ to an instance of $\CCV$ by doubling the number of wires, and replacing every three-valued comparator gate of the $\TCV$ instance with a gadget with two standard comparator gates simulating it.

The above construction shows how to reduce the question
of whether a designated wire outputs 1 for a given instance of $\TCV$
to the question of whether a \emph{pair} of wires of an instance of $\CCV$
output $\seq{1,1}$. However for an instance of $\CCV$ we are only
allowed to decide whether a \emph{single} designated wire outputs 1.
This technical difficulty can be easily overcome since  we can use an $\wedge$-gate (one of the two outputs of a comparator gate) to test whether
a pair of wires outputs $\seq{1,1}$, and output the result on a single designated wire.
\end{proof}

Subramanian~\cite[\S 6.2.6]{Sub90} generalizes the construction of Theorem~\ref{theo:tcv} to 1-Lipschitz gates which are \emph{uniparous}: if the input to a gate contains at most one non-star, then the output contains at most one non-star.

\subsection{Algorithms for solving stable marriage problems} \label{s:SMP-CCV}
In this section, we develop a reduction from $\SMP$ to $\CCV$ due to Subramanian~\cite{Sub90,Sub94}, and later extended to a more general class of problems by Feder~\cite{Fed92,Fed95}. Subramanian did not reduce  $\SMP$ to $\CCV$ directly, but to the \emph{network stability problem} built from the less standard X gate, which takes two inputs $p$ and $q$ and produces two outputs $p'=p\wedge \neg q$ and $q'=\neg p\wedge q$.  It is important to note that the ``\emph{network}'' notion in Subramanian's work denotes a generalization of circuits  by allowing a connection from the output of a gate to the input of any gate including itself, and thus a network in his definition might contain cycles. An X-network is a network consisting only of X gates under the important restriction that each X gate has fan-out exactly one for each output it computes. The network stability problem for X gates ($\XNS$) is then to decide if an X-network has a stable configuration, i.e., a way to assign Boolean values to the wires of the 
network so that  the values are 
compatible with all the X gates of the network. 
Subramanian showed in his dissertation \cite{Sub90}  that $\SMP$, $\XNS$ and  $\CCV$ are all equivalent under $\log$ space reductions. 

We do not work with $\XNS$ in this paper since networks are less intuitive and do not have a nice graphical representation as do comparator circuits. By utilizing Subramanian's idea, we give a direct $\AC^{0}$ reduction from $\SMP$ to $\CCV$, using the three-valued variant of $\CCV$ developed in~\secref{ccvpi}.
%For this goal, it turns out to be conceptually simpler to go through a new variant of $\CCV$, where the comparator gates are three-valued instead of Boolean.

We will describe a sequence of algorithms, starting with Gale and Shapley's algorithm, which is historically the first algorithm solving the stable marriage problem, and ending with Subramanian's algorithm. 

\subsubsection{Notation}
Let $M$ denote the set of men, and $W$ denote the set of women; both are of size $n$. The preference list for a person $p$ is given by
\[ \pref_1(p) \succ_p \pref_2(p) \succ_p \cdots \succ_p \pref_n(p) \succ_p \nessun. \]
The last place on the list is taken by the placeholder $\nessun$ which represents $p$ being unmatched, a situation less preferable than being matched. If $p$ is a man then $\pref_1(p),\ldots,\pref_n(p)$ are women, and vice versa. 

The preference relation $\succ_p$ is defined by $\pref_i(p) \succ_p \pref_j(p)$ whenever $i < j$; we say that $p$ prefers $\pref_i(p)$ over $\pref_j(p)$. For a set of women $W_0$ and a man $m$, the woman $m$ prefers the most is $\max_m W_0$; if $W_0$ is empty, then $\max_m W_0 = \nessun$. Let $S$ be a set, then we write $q \succ_{p} S$ to denote that $p$ prefers $q$ to any person in $S$; similarly, $q \prec_{p} S$ denotes that $p$ prefers any person in $S$ to $q$.

A marriage $P$ is a set of pairs $(m,w)$ which forms a perfect matching between the set of men and the set of women.  In a marriage $P$, we let $P(p)$ denote the person $p$ is married. A marriage is stable if there is no unstable pair $(m,w)$, which is a pair satisfying $w \succ_m P(m)$ and $m \succ_w P(w)$, i.e.$m$ and $w$ prefer each other more than their current partner.

\subsubsection{Gale-Shapley} \label{s:gale-shapley}

Gale and Shapley's algorithm~\cite{GS62} proceeds in rounds. In the first round, each man proposes to his top woman among the ones he hasn't proposed, and each woman selects her most preferred suitor. In each subsequent round, each rejected man proposes to his next choice, and each woman selects her most preferred suitor (including her choice from the previous round). The situation eventually stabilizes, resulting in the man-optimal stable marriage.

There are many ways to implement the algorithm. One of them is illustrated below in Algorithm~\ref{alg:gale-shapley}. The crucial object is the graph $G$, which is a set of possible matches. Each round, each man $m$ selects the top woman $\topp(m)$ currently available to him. Among all men who chose her (if any), each woman $w$ selects the best suitor $\best(w)$. Whenever any man $m$ is rejected by his top woman $w$, we remove the possible match $(m,w)$ from $G$.

\begin{algorithm}
\caption{Gale-Shapley}
\label{alg:gale-shapley}
\begin{algorithmic}
\STATE $G \gets \{(m,w) : m \in M, w \in W\}$
\REPEAT
\STATE $\topp(m) \gets \max_m \{ w : (m,w) \in G \}$ for all $m \in M$
\STATE $\best(w) \gets \max_w \{ m : \topp(m) = w \}$ for all $w \in W$
\STATE Remove $(m,\topp(m))$ from $G$ whenever $\best(\topp(m)) \neq m$
\UNTIL {$G$ stops changing}
\RETURN {$\bset{ (m, \topp(m)) : m \in M }$}
\end{algorithmic}
\end{algorithm}

\begin{lemma} \label{lem:gale-shapley}
 Algorithm \ref{alg:gale-shapley} returns the man-optimal stable matching, and terminates after at most $n^2$ rounds.
\end{lemma}
\begin{proof}
\textbf{Admissibility:} If a pair $(m,w)$ is removed from $G$, then no stable marriage contains the pair $(m,w)$. This is proved by induction on the number of pairs removed. A pair $(m,w)$ is removed when $w = \topp(m) = \topp(m')$ for some other man $m'$, and $m' \succ_w m$. Suppose for a contradiction that $P$ is a stable marriage and  if $P(m) = w$. By the induction hypothesis, we know that  $m'$ can never be married to any woman $w'$ such that $w'\succ_{m'} w$ since that edge $(m',w')$ was removed previously. Thus $w \succeq_{m'} P(m')$. But then $(m',w)$ would be an unstable pair, a contradiction. 

\textbf{Definiteness:} For all men $m$ and at all times, $\topp(m) \neq \nessun$. For any man $m$, $(m,\pi_n(m))$ is never removed from $G$, and so $\topp(m)$ is always well-defined. Indeed, for each $w$, after each iteration $\best(w)$ is non-decreasing in the preference order of $w$. So if $(m,w) \notin G$, $\best(w) \succ_w m$. On the other hand, for any two women $w$ and $w'$, if $\best(w),\best(w') \neq \nessun$ then $\best(w) \neq \best(w')$. Thus, if $(m,w) \notin G$ for all $w \in W$, then $\best$ is a injective mapping from $W$ into $M \setminus \{m\}$, contradicting the pigeonhole principle.

\textbf{Completeness:} The output of the algorithm is a marriage. The algorithm ends when $\best(\topp(m)) = m$ for every $m$, which implies that $\best$ and $\topp$ are mutually inverse bijections.

\textbf{Stability:} The output of the algorithm is a stable marriage. Suppose $(m,w)$ were an unstable pair, so at the end of the algorithm, $m \succ_w \best(w)$ and $w \succ_m \topp(m)$ (we're using the fact that $\topp$ and $\best$ are inverses at the end of the algorithm). However, $m \succ_w \best(w)$ implies $\topp(m) \neq w$, which implies $\topp(m) \succ_m w$.

\textbf{Optimality:} The output of the algorithm is the man-optimal stable marriage. This is obvious, since each man gets his best choice among all possible stable marriages.

%Since each man gets his top choice, the output of the algorithm is the man-optimal stable matching. %The matching is also woman-pessimal: if $w$ is matched to $m \prec_w \best(w)$ then $(\best(w),w)$ is an unstable pair, since $\best(w)$ is matched to a woman inferior to $w$ by man-optimality.

\textbf{Runtime:} The algorithm terminates in $n^2$ iterations since at most $n^2$ edges can be deleted from $G$.
\end{proof}

Gale-Shapley has one disadvantage: it only computes the man-optimal stable matching. This is easy to rectify by symmetrizing the algorithm, resulting in Algorithm~\ref{alg:gale-shapley-symm}. While in the original algorithm, only the men propose (select their top choices), and only the women accept or reject (choose the most promising suitor), in the symmetric algorithm, both sexes participate in both tasks in parallel. The algorithm returns both the man-optimal and the woman-optimal stable marriages.

\begin{algorithm}
\caption{Symmetric Gale-Shapley}
\label{alg:gale-shapley-symm}
\begin{algorithmic}
\STATE $G \gets \{(m,w) : m \in M, w \in W\}$
\REPEAT
\STATE $\topp(m) \gets \max_m \{ w : (m,w) \in G \}$ for all $m \in M$
\STATE $\topp(w) \gets \max_w \{ m : (m,w) \in G \}$ for all $w \in W$
\STATE $\best(w) \gets \max_w \{ m : \topp(m) = w \}$ for all $w \in W$
\STATE $\best(m) \gets \max_m \{ w : \topp(w) = m \}$ for all $m \in M$
\STATE Remove $(m,\topp(m))$ from $G$ whenever $\best(\topp(m)) \neq m$
\STATE Remove $(\topp(w),w)$ from $G$ whenever $\best(\topp(w)) \neq w$
\UNTIL {$G$ stops changing}
\RETURN {$\bset{ (m, \topp(m)) : m \in M }$ and $\bset{ (\topp(w), w) : w \in W }$ as the man-optimal and the woman-optimal stable marriages respectively}
\end{algorithmic}
\end{algorithm}

\begin{lemma} \label{lem:gale-shapley-symm}
 Algorithm \ref{alg:gale-shapley-symm} returns the man-optimal and woman-optimal stable matchings, and terminates after at most $n^2$ rounds. 
\end{lemma}
\begin{proof}
 The analysis is largely analogous to the analysis of the original algorithm. Every pair $(m,w)$ removed from $G$ belongs to no stable marriage. Furthermore, since a stable marriage exists, $\topp(m)$ and $\topp(w)$ are always defined after the algorithm finishes. At the end of the algorithm, $\topp$ and $\best$ are  mutually inverse bijections on $M \cup W$, hence the outputs are marriages. The same arguments as before show that the marriages returned by the algorithm are  man-optimal and woman-optimal stable marriages respectively. Finally, the algorithm terminates in $n^2$ iterations since we can only remove at most $n^2$ edges $G$.
\end{proof}

\subsubsection{Interval algorithms} \label{s:interval}

At the end of Algorithm~\ref{alg:gale-shapley-symm}, for each man $m$, his partner in the man-optimal stable marriage is $\topp(m)$, while his partner in the woman-optimal stable marriage is $\best(m)$. The same holds for women (with the roles of the sexes reversed). This prompts our next algorithm, Algorithm~\ref{alg:interval}, which explicitly keeps track of an interval $\interval(p)$ of possible matches for each person $p$ (these are intervals in the person's preference order).

At each round, each person $p$ first picks their top choice $\topp(p)$. Then each person $q$ picks their top suitor $\best(q)$, if any. People over whom $\best(q)$ is preferred are removed from $\interval(q)$. If $p$ was rejected by his top choice $\topp(p)$, then $\topp(p)$ is removed from $\interval(p)$. These update rules maintain the contiguous nature of the intervals. The situation eventually stabilizes, and the algorithm returns the man-optimal and the woman-optimal stable marriages.

\begin{algorithm}
\caption{Interval algorithm}
\label{alg:interval}
\begin{algorithmic}
\STATE $\interval_0(m) \gets W$ for all $m \in M$
\STATE $\interval_0(w) \gets M$ for all $w \in W$
\STATE $t \gets 0$
\REPEAT
\STATE $\topp_t(p) \gets \max_p \interval_t(p)$ for all $p \in M \cup W$
\STATE $\best_t(q) \gets \max_q \{ p : q = \topp_t(p) \}$ for all $q \in M \cup W$
\STATE Remove $p$ from $\interval_t(q)$ whenever $p \prec_q \best_t(q)$, for all $p,q$ of opposite sex
\STATE Remove $\topp_t(p)$ from $\interval_t(p)$ if $p \neq \best_t(\topp_t(p))$
\STATE $t \gets t + 1$
\UNTIL {$\interval_{t+1}(p) = \interval_t(p)$ for all $p \in M \cup W$}
\RETURN {$\bset{(m,\max_m \interval_t(m)) : m \in M}$ and $\bset{(\max_w \interval_t(w),w) : w \in W}$ as the man-optimal and the woman-optimal stable marriages respectively}
\end{algorithmic}
\end{algorithm}

\begin{lemma} \label{lem:interval}
 Algorithm \ref{alg:interval} returns the man-optimal and woman-optimal stable matchings, and terminates after at most $2n^2$ rounds.
Furthermore, the man-optimal and woman-optimal matchings are given  by 
 \[\bset{(m,\max_m \interval_t(m)) : m \in M}\text{ and }\bset{(\max_w \interval_t(w),w) : w \in W} \text{ respectively}.\]
\end{lemma}
\begin{proof}

\textbf{Admissibility:} In every stable marriage, every person $p$ is matched to someone from $\interval(p)$. This is proved by induction on the number of rounds. A person $q$ can be removed from $\interval(p)$ for one of two reasons: either $q \prec_p \best(p)$, or $q = \topp(p)$ and $p \neq \best(q)$. In the former case, if $p$ were matched to $q$, then $(p,\best(p))$ would be an unstable pair, since $p = \topp(\best(p))$ implies that $\best(p)$ prefers $p$ to any other partner in $\interval(\best(p))$. In the latter case, if $p$ were matched to $q = \topp(p)$, then $(q,\best(q))$ would be an unstable pair, since $q$ prefers $\best(q)$ over $p$ by definition, and $\best(q)$ prefers $q$ as in the former case.

The remaining analysis of this algorithm is similar to the analysis of Gale-Shapley. The outputs of the algorithm are marriages, since the algorithm ends when $\best(\topp(p)) = p$ for all $p$, hence $\topp$ and $\best$ are inverse bijections. The marriages are stable for the same reason given for Gale-Shapley. They are man-optimal and woman-optimal for the same reason. The number of iterations is at most $2n^2$ since there are $2n$ intervals, each of initial length $n$.

Finally, at the termination of the algorithm, $\best(q) = \max_q \interval(q)$. Since $\topp$ and $\best$ are inverses, this explains the dual formulas for the man-optimal and woman-optimal matchings.
\end{proof}

Our next algorithm introduces a new twist. Instead of removing $\topp(p)$ from $\interval(p)$ whenever $p \neq \best(\topp(p))$, we remove $\topp(p)$ from $\interval(p)$ whenever $p \notin \interval(\topp(p))$ as shown in Algorithm~\ref{alg:interval-delayed}. The idea is that if at some point $p \neq \best(\topp(p))$, then $\best(\topp(p)) \succ_{\topp(p)} p$, so $p$ is removed from $\interval(\topp(p))$. At the following iteration, $\topp(p)$ will be removed from $\interval(p)$ in reciprocity.  Thus, Algorithm~\ref{alg:interval-delayed} emulates Algorithm~\ref{alg:interval} with a delay of one round. We will later show that the advantage of this strange rule is the nice representation of the same algorithm in three-valued logic which can then be transformed to Subramanian's algorithm, implementable by comparator circuits.

\begin{algorithm}
\caption{Delayed interval algorithm}
\label{alg:interval-delayed}
\begin{algorithmic}
\STATE $\interval_0(m) \gets W$ for all $m \in M$
\STATE $\interval_0(w) \gets M$ for all $w \in W$
\STATE $t \gets 0$
\REPEAT
\STATE $\topp_{t}(p) \gets \max_p \interval(p)$ for all $p \in M \cup W$
\STATE $\best_{t}(q) \gets \max_q \{ p : q = \topp_{t}(p) \}$ for all $q \in M \cup W$
\STATE Remove $p$ from $\interval_t(q)$ whenever $p \prec_q \best_{t}(q)$, for all $p,q$ of opposite sex
\STATE Remove $\topp_{t}(p)$ from $\interval_t(p)$ if $p \notin \interval_t(\topp_{t}(p))$
\STATE $t \gets t + 1$
\UNTIL {$\interval_{t+1}(p) = \interval_t(p)$ for all $p \in M \cup W$}
\RETURN {$\bset{(m,\max_m \interval_t(m)) : m \in M}$ and $\bset{(\max_w \interval_t(w),w) : w \in W}$ as the man-optimal and the woman-optimal stable marriages respectively}
\end{algorithmic}
\end{algorithm}

\begin{lemma} \label{lem:interval-delayed}
 Algorithm \ref{alg:interval} returns the man-optimal and woman-optimal stable matchings, and terminates after at most $2n^2$ rounds. Furthermore, the man-optimal and woman-optimal matchings are given  by
 \[\bset{(m,\max_m \interval_t(m)) : m \in M}\text{ and }\bset{(\max_w \interval_t(w),w) : w \in W} \text{ respectively}.\]
 \end{lemma}
\begin{proof}
 Clearly Algorithm \ref{alg:interval-delayed} is admissible, that is $p$ is matched to someone from $\interval(p)$ in any stable matching. Furthermore, at the end of the algorithm, $p = \best(\topp(p))$. Otherwise, there are two cases. If $p \in \interval(\topp(p))$, then $p$ would be removed from $\interval(\topp(p))$, and the algorithm would continue. If $p \notin \interval(\topp(p))$, then $\topp(p)$ would be removed from $\interval(p)$, and the algorithm would continue; note that by definition, at the beginning of the round, $\topp(p) \in \interval(p)$.

 The rest of the proof follows the one for Algorithm~\ref{alg:interval}.
\end{proof}

\begin{corollary} \label{cor:interval-delayed}
 The intervals at the end of Algorithm~\ref{alg:interval} coincide with the intervals at the end of Algorithm~\ref{alg:interval-delayed}.
\end{corollary}
\begin{proof}
 That follows immediately from the two formulas for the output. 
\end{proof}

The delayed interval algorithm can be implemented using three-valued logic. The key is the following encoding of the intervals using matrices, which we call the \emph{matrix representation}.
\begin{align*}
\MM(m,w) &= \begin{cases}
1 & \text{if  } w \succeq_m \max_m \interval(m) \\
\ast & \text{if  } \max_m \interval(m) \succ_m w \succeq_m \min_m \interval(m) \\
0 & \text{if  } \min_m \interval(m) \succ_m w
\end{cases} \\
\WW(w,m) &= \begin{cases}
0 & \text{if  } m \succeq_w \max_w \interval(w) \\
\ast & \text{if  } \max_w \interval(w) \succ_w m \succeq_w \min_w \interval(w) \\
1 & \text{if  } \min_w \interval(w) \succ_w m
\end{cases}
\end{align*}

In other words, for every man $m$, the array $\MM(m,\pi_1(m)),\ldots,\MM(m,\pi_n(m))$ has the form
\[
 \begin{array}{cc|cccc|ccc} \cline{3-6}
  1 & \cdots & 1 & \ast & \cdots & \ast & 0 & \cdots & 0 \\\cline{3-6}
 \end{array}
\]
where the men whose corresponding values are contained in the box are precisely the men in $\interval(m)$.

For every woman $w$, the array $\WW(w,\pi_1(w)),\ldots,\WW(w,\pi_n(w))$ has the form
\[
 \begin{array}{cc|cccc|ccc} \cline{3-6}
  0 & \cdots & 0 & \ast & \cdots & \ast & 1 & \cdots & 1 \\\cline{3-6}
 \end{array}
\]
where the women whose corresponding values are contained in the box are precisely the women in $\interval(w)$.

Algorithm~\ref{alg:interval-logic} is an implementation of Algorithm~\ref{alg:interval-delayed} using three-valued logic. We will show, in a sequence of steps, that at each point in time, the matrices representing the intervals in Algorithm~\ref{alg:interval-delayed} equal the matrices in Algorithm~\ref{alg:interval-logic}.

\begin{algorithm}
\caption{Delayed interval algorithm, three-valued logic formulation}
\label{alg:interval-logic}
\begin{algorithmic}
\STATE $\MM_0(m,w) = \begin{cases} 1 & \text{if  } w = \pi_1(m) \\ \ast & \text{otherwise} \end{cases}$
\STATE $\WW_0(w,m) = \begin{cases} 0 & \text{if  } m = \pi_1(w) \\ \ast & \text{otherwise} \end{cases}$
\STATE $t \gets 0$
\REPEAT
\STATE $\MM_{t+1}(m, \pi_i(m)) = \begin{cases} 1 & \text{if  } i = 1 \\ \MM_t(m,\pi_{i-1}(m)) \land \bigwedge_{j \leq i-1} \WW_t(\pi_j(m),m) & \text{otherwise} \end{cases}$
\STATE $\WW_{t+1}(w, \pi_i(w)) = \begin{cases} 0 & \text{if  } i = 1 \\ \WW_t(w,\pi_{i-1}(w)) \lor \bigvee_{j \leq i-1} \MM_t(\pi_j(w),w) & \text{otherwise} \end{cases}$
\STATE $t \gets t + 1$
\UNTIL {$\MM_t = \MM_{t-1}$ and $\WW_t = \WW_{t-1}$}
\STATE $S_M \gets \bset{ (m,w) : \MM_t(m,w) = 1 \text{ and } \WW_t(w,m) \in \{0, \ast\} }$\hfill \% man-optimal stable marriage
\STATE $S_W \gets \bset{ (m,w) : \WW_t(w,m) = 0 \text{ and } \MM_t(m,w) \in \{1, \ast\} }$ \hfill \% woman-optimal stable marriage
\RETURN {$S_M,S_W$}
\end{algorithmic}
\end{algorithm}

First, we show that the matrices properly encode intervals.

\begin{lemma} \label{lem:interval:monotone}
 At each time $t$ in the execution of Algorithm~\ref{alg:interval-logic}, and for each man $m$, the sequence 
 \[\MM_t(m,\pi_1(m)),\ldots,\MM_t(m,\pi_n(m))\] is non-increasing (with respect to the order $1 > \ast > 0$). Similarly, for each woman $w$, the sequence \[\WW_t(w,\pi_1(w)),\ldots,\WW_t(w,\pi_n(w))\] is non-decreasing
\end{lemma}

\begin{proof}
The proof is by induction.  The claim is clearly true at time $t = 0$.  For the inductive case, it suffices to analyze $\MM$ since $\WW$ can be handled dually. Furthermore, at each iteration, we ``shift'' each sequence $\MM(m,\cdot)$ one step to the right and  add  a $1$ to the left end to get the following non-increasing sequence
\[
1,\MM_{t-1}(m,\pi_{1}(m)),\MM_{t-1}(m,\pi_{2}(m)),\ldots,\MM_{t-1}(m,\pi_{n-1}(m))
\]  
Then we take a component-wise AND of the above sequence with the non-increasing sequence 
 \[
  1,\WW_{t-1}(\pi_1(m),m), \bigl(\WW_{t-1}(\pi_1(m),m) \land \WW_{t-1}(\pi_2(m),m)\bigr), \ldots,
  \bigl(\WW_{t-1}(\pi_1(m),m) \land \cdots \land \WW_{t-1}(\pi_{n-1}(m),m)\bigr). 
 \]
 It's not hard to check that the result is also a non-increasing sequence by the properties of three-valued logic.
\end{proof}

Second, we show that the intervals encoded by the matrices can only shrink. This is the same as saying that whenever an entry gets determined (to a value different from $\ast$), it remains constant.

\begin{lemma} \label{lem:interval:fixed}
 If for some time $t$, for some man $m$ and some woman $w$, $\MM_t(m,w) \in \{0,1\}$, then $\MM_s(m,w) = \MM_t(m,w)$ for $s \geq t$. A similar claim holds for $\WW$.
\end{lemma}
\begin{proof}
 We prove the claim by induction on $t$. Let $w = \pi_i(m)$.  If $i = 1$, then the claim is trivial. Now suppose $i > 1$. If $\MM_t(m,\pi_i(m)) = 1$, then 
 \[\MM_{t-1}(m,\pi_{i-1}(m)) = \WW_{t-1}(\pi_1(m),m) = \cdots = \WW_{t-1}(\pi_{i-1}(m),m) = 1.\] 
 The induction hypothesis shows that all these elements retain their values in the next iteration and hence $\MM_{t+1}(m,\pi_{i+1}(m)) = 1$. If $\MM_t(m,\pi_i(m)) = 0$, then at least one of these elements is equal to zero; this element retains its value in the next iteration by the induction hypothesis; hence $\MM_{t+1}(m,\pi_{i+1}(m)) = 0$.
\end{proof}

It remains to show that the way that the underlying intervals are updated matches the update rules of Algorithm~\ref{alg:interval-delayed}.

\begin{lemma} \label{lem:interval-logic}
 At each time $t$, the matrix representation of the intervals in Algorithm~\ref{alg:interval-delayed} is the same as the matrices $\MM_t,\WW_t$ in the execution of Algorithm~\ref{alg:interval-logic}. Furthermore, both algorithms return the same marriages.
\end{lemma}
\begin{proof}
 The proof is by induction on $t$. The base case $t = 0$ is clear by inspection.

 We now compare the update rules in some round $t$ of both algorithms. There are two ways an interval $\interval(m)$ can be updated: either a woman is removed from the bottom of the interval, or a woman is removed from the top of the interval.

 In the former case, a woman $\pi_i(m)$ is removed from $\interval_{t+1}(m)$ since $\pi_i(m) \prec_m \best_t(m)$. Suppose $\best_t(m) = \pi_j(m)$, where $j < i$. Since $m = \topp_t(\pi_j(m))$, we know that $\WW_t(\pi_j(m),m) = 0$, and so
\begin{align}
 \MM_{t+1}(m,\pi_i(m)) = \MM_t(m,\pi_{i-1}(m)) \land \bigwedge_{j \leq i-1} \WW_t(\pi_j(m),m) = 0. \label{eq:interval-logic}
\end{align}
 Conversely, suppose $\MM_{t+1}(m,\pi_i(m)) = 0$ while $\MM_t(m,\pi_i(m)) = \ast$. Since $\MM_t(m,\cdot)$ is non-increasing, we have $\MM_t(m,\pi_{i-1}(m)) \neq 0$. Thus for $\MM_{t+1}(m,\pi_i(m)) = 0$, we must have $\WW_t(\pi_j(m),m) = 0$ for some $j < i$. Now suppose $m \neq \topp_t(\pi_j(m))$, then that equation \eref{interval-logic} were true at an earlier time $s<t$, at which $\MM_{s}(m,\pi_i(m))$ would have become~$0$. Hence $m = \topp_t(\pi_j(m))$, and $\pi_i(m)$ is removed from $\interval_{t+1}(m)$.

 In the latter case, a woman $\pi_i(m)$ is removed from $\interval_{t+1}(m)$ since $\pi_i(m) = \topp_t(m)$ and $m \notin \interval_t(\pi_i(m))$. We claim that $m \prec_{\pi_i(m)} \interval(\pi_i(m))$, since otherwise $m \succ_{\pi_i(m)} \interval(\pi_i(m))$. Thus $(m,\pi_i(m))$ would be an unstable pair in any marriage produced by the algorithm ($m$ will be matched to a woman inferior to $\topp_t(m) = \pi_i(m)$, and $\pi_i(m)$ will be matched to a man from $\interval(\pi_i(m))$ whom  $\pi_i(m)$ doesn't like as much as $m$), and this contradicts that the algorithm produces some stable marriage. Hence we have $m \prec_{\pi_i(m)} \interval(\pi_i(m))$ and so $\WW_t(\pi_i(m),m) = 1$.  For $j \leq i-1$, the woman $\pi_j(m)$ must have been removed from $\interval(m)$ in the past (since women can only be removed from the top of $\interval(m)$ one at a time), and at that time $\WW(\pi_j(m),m) = 1$ (just as at the current time, $\WW_t(\pi_i(m),m) = 1$); by Lemma~\ref{lem:interval:fixed}, this value stays $1$. Hence
 \[ \MM_{t+1}(m,\pi_{i+1}(m)) = \MM_t(m,\pi_i(m)) \land \bigwedge_{j \leq i} \WW_t(\pi_j(m),m) = 1. \]
 Conversely, suppose $\MM_{t+1}(m,\pi_{i+1}(m)) = 1$ while $\MM_t(m,\pi_{i+1}(m)) = \ast$. Since $\MM_t(m,\pi_i(m)) = 1$, we must have $\pi_i(m) = \topp_t(m)$. Further, $\WW_t(\pi_i(m),m) = 1$ shows that $m \notin \interval_t(\pi_i(m))$. Therefore $\pi_i(m)$ is removed from $\interval_{t+1}(m)$.

 Finally, we conclude that the two algorithms return the same man-optimal and woman-optimal stable marriages by  \lemref{interval:monotone}.
\end{proof}

Let us illustrate the workings of Algorithm~\ref{alg:interval-logic}. The following diagram illustrates a situation in which a man's interval shrinks from the bottom. The diagram illustrates $\WW_t(w,\cdot), \MM_t(m,\cdot), \MM_{t+1}(m,\cdot)$, respectively, where $m = \topp_t(w)$. The two elements in red form a pair $\MM(m,w),\WW(w,m)$.
\begin{align*}
 \begin{array}{lc|cccccc|cc} 
  \multicolumn{2}{c}{} & \multicolumn{1}{c}{m} & \multicolumn{1}{c}{} & \hphantom{w} & \multicolumn{5}{c}{} \\ \cline{3-8}
  \WW_t(w,\cdot) & 0 & {\color{red}0} & \ast & \ast & \ast & \ast & \ast & 1 & 1 \\ \cline{3-8}
 \end{array} \\
 \begin{array}{lcc|cccc|ccc} 
  \multicolumn{2}{c}{} & \multicolumn{1}{c}{\hphantom{m}} & \multicolumn{1}{c}{} & \multicolumn{1}{c}{w} & \multicolumn{5}{c}{} \\ \cline{4-7}
  \MM_t(m,\cdot) & 1 & 1 & 1 & {\color{red}\ast} & \ast & \ast & 0 & 0 & 0 \\ \cline{4-7}
 \end{array} \\
 \begin{array}{lcc|cc|ccccc}
   \multicolumn{2}{c}{} & \multicolumn{1}{c}{\hphantom{m}} & \multicolumn{1}{c}{} & \multicolumn{1}{c}{\hphantom{w}} & \Downarrow & \multicolumn{4}{c}{} \\ \cline{4-5}
  \MM_{t+1}(m,\cdot) & 1 & 1 & 1 & {\color{red}\ast} & 0 & 0 & 0 & 0 & 0 \\ \cline{4-5}
 \end{array}
\end{align*}

The next diagram illustrates a situation in which a man's interval shrinks from the top. This time $w = \topp_t(m)$.
\begin{align*}
 \begin{array}{lc|cccccc|cc} 
  \multicolumn{3}{c}{} & \hphantom{w} & \multicolumn{5}{c}{} & m \\ \cline{3-8}
  \WW_t(w,\cdot) & 0 & 0 & \ast & \ast & \ast & \ast & \ast & 1 & {\color{red}1} \\ \cline{3-8}
 \end{array} \\
 \begin{array}{lcc|cccc|ccc}
  \multicolumn{3}{c}{} & w & \multicolumn{5}{c}{} & \hphantom{m} \\ \cline{4-7}
  \MM_t(m,\cdot) & 1 & 1 & {\color{red}1} & \ast & \ast & \ast & 0 & 0 & 0 \\ \cline{4-7}
 \end{array} \\
 \begin{array}{lccc|ccc|ccc}
  \multicolumn{3}{c}{} & \multicolumn{1}{c}{\hphantom{w}} & \Downarrow & \multicolumn{4}{c}{} & \hphantom{m} \\ \cline{5-7}
  \MM_{t+1}(m,\cdot) & 1 & 1 & {\color{red}1} & 1 & \ast & \ast & 0 & 0 & 0 \\ \cline{5-7}
 \end{array}
\end{align*}

\subsubsection{Subramanian's algorithm} \label{s:subramanian}

Subramanian's algorithm is very similar to Algorithm~\ref{alg:interval-logic}. The latter algorithm is not implementable using comparator circuits, since (for example) the value $\WW_t(\pi_1(m),m)$ is used in the update rule of \[\MM_{t+1}(m,\pi_2(m)),\ldots,\MM_{t+1}(m,\pi_n(m)).\] Subramanian's algorithm, displayed as Algorithm~\ref{alg:subramanian}, corrects this issue by retaining only the most important term in each conjunction and disjunction.

\begin{algorithm}
\caption{Subramanian's algorithm}
\label{alg:subramanian}
\begin{algorithmic}
\STATE $\MM_0(m,w) = \begin{cases} 1 & \text{if  } w = \pi_1(m) \\ \ast & \text{otherwise} \end{cases}$
\STATE $\WW_0(w,m) = \begin{cases} 0 & \text{if  }m = \pi_1(w) \\ \ast & \text{otherwise} \end{cases}$
\STATE $t \gets 0$
\REPEAT
\STATE $\MM_{t+1}(m, \pi_i(m)) = \begin{cases} 1 & \text{if  } i = 1 \\ \MM_t(m,\pi_{i-1}(m)) \land \WW_t(\pi_{i-1}(m),m) & \text{otherwise} \end{cases}$
\STATE $\WW_{t+1}(w, \pi_i(w)) = \begin{cases} 0 & \text{if  } i = 1 \\ \WW_t(w,\pi_{i-1}(w)) \lor \MM_t(\pi_{i-1}(w),w) & \text{otherwise} \end{cases}$
\STATE $t \gets t + 1$
\UNTIL {$\MM_t = \MM_{t-1}$ and $\WW_t = \WW_{t-1}$}
\STATE $S_M \gets \bset{ (m,w) : \MM_t(m,w) = 1 \text{ and } \WW_t(w,m) \in \{0, \ast\} }$\hfill \% man-optimal stable marriage
\STATE $S_W \gets \bset{ (m,w) : \WW_t(w,m) = 0 \text{ and } \MM_t(m,w) \in \{1, \ast\} }$\hfill \% woman-optimal stable marriage
\RETURN {$S_M,S_W$}
\end{algorithmic}
\end{algorithm}

In the analysis of Algorithm~\ref{alg:interval-logic}, we already saw that Subramanian's update rule works when an entry of $\MM$ is set to~$1$. When Algorithm~\ref{alg:interval-logic} sets an entry to~$0$, say $\MM_{t+1}(m,\pi_i(m)) = 0$, it is due to $\WW_t(\pi_j(m),m) = 0$ for some $j \leq i-1$. If $j = i-1$, then Subramanian's algorithm will also set $\MM_{t+1}(m,\pi_i(m)) = 0$. Otherwise, Subramanian's algorithm will set $\MM_{t+1}(m,\pi_{j+1}(m)) = 0$, and this zero will propagate, so that $\MM_{t+i-j}(m,\pi_i(m)) = 0$. This shows that in some sense, Subramanian's algorithm mimics Algorithm~\ref{alg:interval-logic}. It remains to show that Subramanian's algorithm computes the same final $\MM$ and $\WW$ as Algorithm~\ref{alg:interval-logic}.

First, we notice that the termination conditions of both algorithms are really the same. For Subramanian's algorithm, the conditions are that for $i > 1$,
\begin{equation}  \label{eq:subramanian:termination}
\begin{aligned}
 \MM(m, \pi_i(m)) &= \MM(m,\pi_{i-1}(m)) \land \WW(\pi_{i-1}(m),m), \\
 \WW(w, \pi_i(w)) &= \WW(w,\pi_{i-1}(w)) \lor \MM(\pi_{i-1}(w),w).
\end{aligned}
\end{equation}
For Algorithm~\ref{alg:interval-logic}, the termination conditions are
\begin{equation} \label{eq:interval:termination}
\begin{aligned}
 \MM(m, \pi_i(m)) &= \MM(m,\pi_{i-1}(m)) \land \bigwedge_{j \leq i-1} \WW(\pi_j(m),m), \\
 \WW(w, \pi_i(w)) &= \WW(w,\pi_{i-1}(w)) \lor \bigvee_{j \leq i-1} \MM(\pi_j(w),w).
\end{aligned}
\end{equation}
\begin{lemma} \label{lem:subramanian:termination}
 The matrices $\MM,\WW$ at the end of Subramanian's algorithm satisfy the termination conditions of Algorithm~\ref{alg:interval-logic}, and vice versa. Moreover, these are always matrix representations of intervals.
\end{lemma}
\begin{proof}
We observe in both algorithms the update rules guarantee that $\MM(m,\cdot)$ is monotone non-increasing and that $\WW(w,\cdot)$ is monotone non-decreasing, which implies
\begin{align*}
 \MM(\pi_{i-1}(w),w) = \bigvee_{j \leq i-1} \MM(\pi_j(w),w),&&
 \WW(\pi_{i-1}(m),m) = \bigwedge_{j \leq i-1} \WW(\pi_j(m),m). 
\end{align*} 
Thus, these two termination conditions are equivalent. 
\end{proof}

We call a pair of matrices $(\MM,\WW)$ a \emph{feasible pair} if they satisfy the equations in \eref{subramanian:termination} or equivalently in~\eref{interval:termination}o, and furthermore $\MM(m,\pi_1(m)) = 1$ and $\WW(w,\pi_1(w)) = 0$ for all man $m$ and woman $w$. The following lemma shows that, in some sense, Subramanian's algorithm is admissible.
\begin{lemma} \label{lem:subramanian:admissible}
 Let $\MM,\WW$ be the matrices at the end of Subramanian's algorithm. If $\MM(m,w) = c \neq \ast$ for some $m,w$, then $\MM'(m,w) = c$ for any feasible pair $(\MM,\WW)$. Same for $\WW$.
\end{lemma}
\begin{proof}
 The proof is by induction on the time $t$ in which $\MM_t(m,w)$ is set to $c$. If $t = 0$, then the claim follows from the definition of feasible pair. Otherwise, for some $i$ we have $\MM_t(m, \pi_i(m)) = \MM_{t-1}(m, \pi_{i-1}(m)) \land \WW_{t-1}(\pi_{i-1}(m), m)$. If $c = 1$ then $\MM_{t-1}(m, \pi_{i-1}(m)) = \WW_{t-1}(\pi_{i-1}(m), m) = 1$, and by induction these entries get the same value in all feasible pairs. The definition of feasible pair then implies that $\MM'(m,\pi_i(m)) = 1$ in any feasible pair $\MM',\WW'$. The case when $c = 0$ can be shown similarly.
\end{proof}

The following lemma shows that we can uniquely extract a stable marriage from each $0/1$-valued feasible pair, i.e. both of the matrices in the pair have $0/1$ values, and vice versa.

\begin{lemma} \label{lem:feasible-pair}
 Suppose $(\MM,\WW)$ is a $0/1$-valued feasible pair. If we marry each man $m$ to $\min_m \{w : \MM(m,w) = 1\}$, and each woman $w$ to $\min_w \{m : \WW(m,w) = 0\}$, then the result is a stable marriage.

 Conversely, every stable marriage $P$ can be encoded as a $0/1$-valued feasible pair, as follows: For each man $m$, we put $\MM(m,w) = 1$ if $w \succeq_m P(m)$, and $\MM(m,w) = 0$ otherwise. For each woman $w$, we put $\WW(w,m) = 0$ if $m \succeq_w P(w)$, and $\WW(w,m) = 1$ otherwise.
 
 The two mappings are inverses of each other.
\end{lemma}
\begin{proof}
 \textbf{Feasible pair implies stable marriage.} Suppose $(\MM,\WW)$ is a $0/1$-valued feasible pair. We start by showing that the mapping $P$ in the statement of the lemma is indeed a marriage.

We call a person \emph{desperate} if he or she is married to the last choice in his or her preference list. 

If a man $m$ is not desperate then for some $1\le i<n$, $\MM(m,\pi_i(m)) = 1$ and $\MM(m,\pi_{i+1}(m)) = 0$. This can only happen if $\WW(\pi_i(m),m) = 0$, and furthermore if $m \succ_{\pi_i(m)} m'$, then $\WW(\pi_i(m),m') = 1$ due to $\MM(m,\pi_i(m)) = 1$. This shows that whenever a man $m$ is not desperate, $m$ is married to someone in the marriage $P$.

 If $m$ is desperate and $\WW(\pi_i(m),m) = 0$, then $m$ is married as before. Otherwise, no woman is desperate, and so similarly to the previous argument, it can be shown that every woman $w$ is married in $P$. However, the fact that $P(w) \neq m$ for all women $w$ contradicts the pigeonhole principle. Thus, we conclude that $P$ is a marriage.

It remains to show that $P$ is stable. Suppose that $(m,w)$ were an unstable pair in $P$. Then $\MM(m,w) = 1$ and $\WW(w,m) = 0$, and moreover $\MM(m,w') = 1$ for some $w' \prec_m w$. Yet~\eqref{eq:interval:termination} shows that $\MM(m,w') \leq \WW(w,m)$, and we reach a contradiction.

 \textbf{Stable marriage implies feasible pair.} Suppose $m$ is matched to $\pi_k(m)$. Consider first the case $i \leq k$. Then $\MM(m,\pi_i(m)) = \MM(m,\pi_{i-1}(m)) = 1$, and we have to show that $\WW(\pi_{i-1}(m),m) = 1$. If the latter weren't true then $(m,\pi_{i-1}(m))$ would be an unstable pair, since $\pi_{i-1}(m) \succ_m \pi_i(m)$ while $\WW(\pi_{i-1}(m),m) = 0$ implies that $\pi_{i-1}(m)$ prefers $m$ to every other man which is matched to her. If $i = k + 1$ then $\MM(m,\pi_i(m)) = 0$ and also $\WW(\pi_{i-1}(m),m) = 0$. If $i > k$ then $\MM(m,\pi_i(m)) = \MM(m,\pi_{i-1}(m)) = 0$.

 \textbf{The mappings are inverses of each other.} It is easy to check directly from the definition that if we start with a stable marriage $P$, convert it to a feasible pair $(\MM,\WW)$, and convert it back into a stable marriage $P'$, then $P = P'$. 
 
For the other direction, Lemma~\ref{lem:subramanian:termination} shows that if $(\MM,\WW)$ is a $0/1$-valued feasible pair then for each man $m$, $\MM(m,\cdot)$ consists of a positive number of $1$s followed by $0$s, and dually for each woman $w$, $\WW(w,\cdot)$ consists of a positive number of $0$s followed by $1$s. Thus, given the fact that our rule of converting $(\MM,\WW)$ to a stable marriage $P$ indeed results in a marriage, it is clear that converting $P$ back to a feasible pair results in $(\MM,\WW)$.
\end{proof}

\begin{lemma} \label{lem:subramanian}
 Subramanian's algorithm returns the man-optimal and woman-optimal stable marriages. Furthermore, the matrices $\MM,\WW$ at the end of Subramanian's algorithm coincide with the matrices $\MM,\WW$ at the end of Algorithm~\ref{alg:interval-logic}.
\end{lemma}
\begin{proof}
 The monotonicity of $\land$ and $\lor$ shows that if we replace every $\ast$ in $\MM,\WW$ with $0$, then the resulting $(\MM,\WW)$ is still a feasible pair; the same holds if we replace every $\ast$  with $1$.

 Lemma~\ref{lem:subramanian:admissible} and Lemma~\ref{lem:feasible-pair} together imply that the first output is the man-optimal stable matching, and the second output is the woman-optimal stable matching. Lemma~\ref{lem:subramanian:termination} shows that at termination, the matrices $\MM,\WW$ are matrix representations of intervals, hence they must coincide with the matrices at the end of Algorithm~\ref{alg:interval-logic}.
\end{proof}

\begin{lemma} \label{lem:subramanian:runtime}
 Subramanian's algorithm terminates after at most $2n^2$ iterations.
\end{lemma}
\begin{proof}
Since there are $2n^2$ entries in both matrices, and at each iteration at least one entry changes from $\ast$ to $0$ or $1$, the algorithm terminates after at most $2n^2$ iterations.
\end{proof}

A formal correctness proof of Subramanian's algorithm can be found in~\cite{LCY11}.

\subsubsection{$\MOSM$ and $\WOSM$  are $\AC^{0}$ many-one reducible to $\CCV$}

In the remaining section, we will show that Subramanian's algorithm can be implemented as a three-valued comparator circuit. 

First, since for each man $m$, the pair of values $\MM_t(m,\pi_{i-1}(m))$ and  $\WW_t(\pi_{i-1}(m),m)$ is only used once to compute the two outputs $\MM_t(m,\pi_{i-1}(m)) \land \WW_t(\pi_{i-1}(m),m)$ and $\MM_t(m,\pi_{i-1}(m)) \lor \WW_t(\pi_{i-1}(m),m)$, and then each output is used at most once  when updating $\MM_{t+1}(m,\pi_{i}(m))$ and $\WW_{t+1}(m,\pi_{i}(m))$. Thus the whole update rule can be easily implemented using comparator gates.   

Second, we know that the algorithm converges within $2n^2$ iterations to a fixed point. Therefore, if we run the loop for exactly $2n^2$ iterations, the result would be the same. Hence, we can build a comparator circuit to simulate exactly $2n^2$ iterations of Subramanian's algorithm. 

Finally, we can extract the man-optimal stable matching using a simple comparator circuit with negation gates. Recall that the logical values $0,\ast,1$ are represented in reality by pairs of wires with values $(0,0),(0,1),(1,1)$. In the man-optimal stable matching, a man $m$ is matched to $\pi_i(m)$ if $\MM(m,\pi_i(m)) = 1$ and either $i = n$ or $\MM(m,\pi_{i+1}(m)) \in \{0,\ast\}$. In the latter case, if the corresponding wires are $(\alpha,\beta)$ and $(\gamma,\delta)$, then the required information can be extracted as $\alpha \land \beta \land \lnot \gamma$.

\begin{theorem}\label{theo:mosm}
   $\MOSM$ and $\WOSM$ are $\AC^{0}$ many-one reducible to $\CCVN$.
\end{theorem}
\begin{proof} We will show only the reduction from $\MOSM$ to $\CCVN$ since the reduction from $\WOSM$ to $\CCVN$ works similarly.

Following the above construction, we can define an $\AC^{0}$ function that takes as input an instance of $\MOSM$ with preference lists for all the men and women, and produces a three-valued comparator circuit that implements Subramanian's algorithm, and then extracts the man-optimal stable matching.
\end{proof}

\corref{lfm} and Theorems \ref{theo:tcv}, \ref{theo:me} and \ref{theo:mosm} give us the following corollary.

\begin{corollary} 
The ten problems $\MOSM$, $\WOSM$, $\SMP$, $\CCV$, $\CCVN$, $\TCV$, $\TLFMM$,  $\LFMM$,  $\TVLFMM$ and  $\VLFMM$ are all equivalent under $\AC^{0}$ many-one reductions.
%, where the equivalence of $\SMP$ is with respect to the search problem version of the reduction defined in \secref{reductions}.
\end{corollary}
\begin{proof}  \corref{lfm} and \theoref{tcv} show that $\CCV$, $\CCVN$, $\TCV$, $\TLFMM$,  $\LFMM$,  $\TVLFMM$ and  $\VLFMM$  are all equivalent under $\AC^{0}$ many-one reductions.

\theoref{mosm} shows that  $\MOSM$ and $\WOSM$ are
$\AC^{0}$ many-one reducible to $\TCV$. \theoref{me} also shows that $\TLFMM$ is $\AC^0$ many-one reducible to $\MOSM$, $\WOSM$, and $\SMP$. Hence,   $\MOSM$, $\WOSM$, and $\SMP$  are equivalent to the above problems under $\AC^0$ many-one reductions. 
\end{proof}

\section{Conclusion \label{conclusion}}
Although we have shown that there are problems in relativized $\NC$ but
not in relativized $\CC$ (uniform or nonuniform), it is quite possible
that some of the standard problems in $\NC^2$ that are not known to be in
$\NL$ might be in (nonuniform) $\CC$ .
Examples are integer matrix powering and context free languages (or
more generally problems in $\LogCFL$).
Of particular interest is matrix powering over the field
$GF(2)$.  We cannot show this is in $\CC$, even though we know that
Boolean matrix powering is in $\CC$ because $\NL \subseteq \CC$.
Another example is the matching problem for bipartite graphs or
general undirected graphs, which is in $\RNC^2$
 \cite{KUW86,MVV87} and hence in nonuniform $\NC^2$.
It would be interesting to show that some (relativized) version of any of
these problems is, or is not, in (relativized) (nonuniform) $\CC$.

Let $\SucCC$ be the class of problems $p$-reducible
to succinct $\CC$ (where a description of an
exponential size comparator circuit is given using linear size
Boolean circuits).  It is easy to show that $\SucCC$ lies between
$\PSPACE$ and $\EXPTIME$, but we
we are unable to show that it is equal to either
class\footnote{Thanks to Scott Aaronson for pointing this out.}.
We are not aware of other complexity classes which appear to lie properly
between $\PSPACE$ and $\EXPTIME$.

We have defined $\CC$ in terms of uniform families comparator circuits, analogously to the way that complexity classes such as uniform $\AC^i$ and $\NC^i$ are defined. The latter also have machine characterizations: $\NC^i$ is the class of relations computable by alternating Turing machines in space $O(\log n)$ and $O(\log^i n)$ alternations, and also the class of relations computable by polynomially many processors in time $O(\log^i n)$. Similarly, $\P$ can be defined either as those relations computable in polynomial time, or as those relations computable by $\AC^0$-uniform polynomial size circuits. An important open question, appearing already in Subramanian's thesis~\cite{Sub90}, is to come up with a similar machine model for $\CC$. While we do outline some machine characterization in Theorem~\ref{t:cctm}, the machine model we use is not as natural as the ones for $\P$ and $\NC$, and so it is still an open question to come up with a natural machine model for $\CC$.

We believe that $\CC$ deserves more attention, since on the one hand it
contains interesting complete problems, on the other hand  the limitation of 
fanout restriction of comparator gates has not yet been studied outside this paper.   
Furthermore, $\CC$ provides us another research direction for separating $\NL$ from $\P$ 
by analyzing the limitation of the fanout restriction.

\appendix
\section{Proof of Proposition \ref{prop:p1}}\label{a:pointDown}

\begin{proposition}
$\CCV$ is $\AC^0$ many-one reducible to the special case in which
all comparator gates point down (or all point up).
\end{proposition}

 \begin{proof}Suppose we have a gate on the left of \figref{samedirection} with the arrow pointing upward. We can construct a circuit that outputs the same values as those of $x$ and $y$, but all the gates will now point downward as shown on the right of \figref{samedirection}.
 \ifslow
 \begin{figure}[!h]%\vspace{-5mm}
 \begin{minipage}{0.35\textwidth}
 \centering
 \small 
 $\Qcircuit @C=2em @R=1em {
 \push{x\,} & \ctrlu{1}        & \qw\\
 \push{y\,} & \ctrl{-1} & \qw  
 }$
 \end{minipage}
 \begin{minipage}{0.2\textwidth}
 \centering
 \begin{tikzpicture}
 \draw[-to,line width=1.5pt,snake=snake,segment amplitude=.5mm,
          segment length=2mm,line after snake=1mm]  (0,0) -- (2,0);
 \end{tikzpicture}
 \end{minipage}
 \begin{minipage}{0.4\textwidth}
 \centering
 \small
 $\Qcircuit @C=2em @R=0.7em {
               &\push{x_{0}} & \ctrl{2}                & \qw   & \qw           & \qw\\
               &\push{y_{0}} &\qw              & \ctrl{1}      & \ctrl{2}              & \qw\\
 \push{0}      &\push{x_{1}} & \ctrld{-1}      &\ctrld{-1}     &\qw                    & \qw\\
 \push{0}      &\push{y_{1}} & \qw             & \qw   &\ctrld{-1}             & \qw \\
 }$
 \end{minipage}
 \caption{}
 \label{fig:samedirection}
 \end{figure}
\vspace{-3mm}
 \fi

It is not hard to see that the wires $x_{1}$ and $y_{1}$ in this new comparator circuit will output the same values as the wires $x$ and $y$ respectively in the original circuit. For the general case, we can simply make copies of all wires for each layer of the comparator circuit, where each copy of a wire will be used to carry the value of a wire at a single layer of the circuit. Then apply the above construction to simulate the effect of each gate. Note that additional comparator gates are also needed to forward the values of the wires from one layer to another,
 in the same way that
 we use the gate $\seq{y_{0},y_{1}}$ to forward the value carried in wire $y_{0}$ to wire $y_{1}$ in the above construction. 
 
 To carry this out in $\AC^{0}$, one way would be to add a complete copy of all wires for every
 comparator gate in the original circuit.  Each new wire has input $0$.  For
 each original gate $g$, first put in gates copying the values to the
 new wires.  If $g$ points down, put in a copy of $g$ connecting the new
 wires, and if $g$ points up, put in the construction in
 \figref{samedirection}.
 \end{proof}

\section{Simplified proof of \texorpdfstring{$\NL\subseteq\CC$}{CC contains NL}} \label{NLinCC}

Each instance of the \textsc{Reachability} problem consists of a directed acyclic  graph $G=(V,E)$, where $V=\set{u_{0},\ldots,u_{n-1}}$, and  we want to decide if there is a path from $u_{0}$ to $u_{n-1}$. It is well-known that \textsc{Reachability} is $\NL$-complete. 
Since a directed graph can be converted by an $\AC^0$ function into a
layered graph with an equivalent reachability problem, it
suffices to give a comparator circuit construction  that solves instances
of   \textsc{Reachability} satisfying the following assumption: 
\begin{align}
\text{The graph $G$ only has directed edges of the form $(u_{i},u_{j})$, where $i<j$.} \label{eq:a1}
\end{align}

The following construction from \cite{LCY11} for showing that
$\NL\subseteq \CC$ is simpler than the one in \cite{Sub90,MS92}. Moreover
it reduces \textsc{Reachability} to $\CCV$ directly without going through some intermediate complete problem, and this was stated as an open problem in \cite[Chapter 7.8.1]{Sub90}.
The idea is to perform a depth-first search of the nodes reachable from
the source node by successively introducing $n$ pebbles into the source,
and sending each pebble along the lexicographically last pebbled path
until it reaches an unpebbled node, where it remains.
After $n$ iterations, all nodes reachable from the source are pebbled, and we can check whether the target is one of them.

\ifslow
\begin{wrapfigure}{r}{0.3\textwidth}
\vspace{-3mm}
\begin{center}
\tikzstyle{vertex}=[circle,fill=black!30,minimum size=15pt,inner sep=0pt]
\tikzstyle{edge} = [draw,line width=1pt,-to]
\begin{tikzpicture}[scale=1.2, auto,swap]
    % First we draw the vertices
    \foreach \pos/\name/\label in {{(0,0)/v0/u_{0}}, {(1,0.5)/v1/u_{2}}, {(1,-0.5)/v2/u_{1}},{(2,0)/v3/u_{3}},						{(2,1)/v4/u_{4}}}
        \node[vertex] (\name) at \pos {$\label$};
        
    % Connect vertices with edges 
    \foreach \source/ \dest  in {v0/v1,v0/v2,v1/v3,v1/v4}
        \path[edge] (\source) -- (\dest);
\end{tikzpicture}\end{center}
\vspace{-4mm}
\caption{}
\label{fig:g}
\end{wrapfigure}
\fi

We will demonstrate the construction through a simple example, where we have the directed graph in \figref{g}  satisfying the assumption \eref{a1}. We will build a comparator circuit as in \figref{NL2CC}, where the wires $\nu_0,\ldots,\nu_4$ represent the vertices $u_{0},\ldots,u_{4}$ of the preceding graph and the wires $\iota_0,\ldots,\iota_4$ are used to feed $1$-bits into the wire $\nu_0$, and from there to the other wires $\nu_{i}$ reachable from $\nu_{0}$. We let every wire $\iota_{i}$ take input 1 and every wire $\nu_{i}$ take input 0.

We next show how to construct the gadget contained in each box. For a graph with
$n$ vertices ($n=5$ in our example), the $k^{\rm th}$ gadget is
constructed as follows:
\begin{algorithmic}[1]
\STATE Introduce a comparator gate from wire $\iota_{k}$ to wire $\nu_{0}$
\FOR{$i=0,\ldots,n-1$}
\FOR{$j=i+1,\ldots,n-1$}
\STATE Introduce a comparator gate from $\nu_i$ to $\nu_j$ if  $(u_i,u_j)\in E$, or a dummy gate on $\nu_i$ otherwise.
\ENDFOR
\ENDFOR
\end{algorithmic}
Note that the gadgets are identical except for the first comparator gate.

We only use the loop structure to clarify the order the gates are added. The construction can easily be done in $\AC^{0}$ since the position of each gate can be calculated exactly, and thus all gates can be added independently from one another. Note that for a graph with $n$ vertices, we have at most $n$ vertices reachable from a single vertex, and thus we need $n$ gadgets as described above. In our example, there are at most 5 wires reachable from wire $\nu_0$, and thus we utilize the gadget 5 times. 
\ifslow
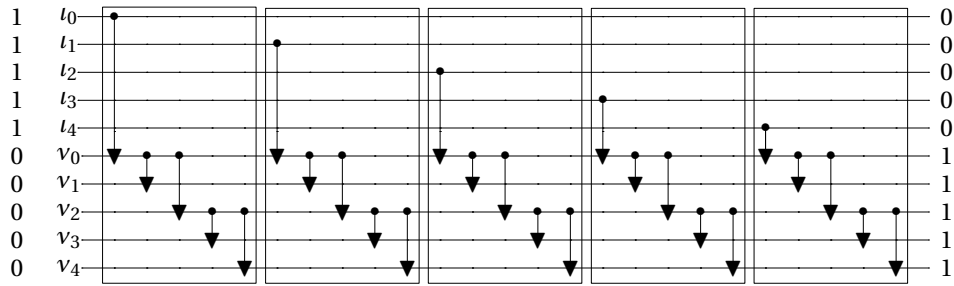
\begin{figure}[!h]
\centering
{
$\Qcircuit @C=1.4em @R=.5em {
\push{\bit{1}}&\push{\iota_0} 
	&\ctrl{4}	& \qw 	& \qw	& \qw	&\qw
	&\qw	& \qw 	& \qw	& \qw	&\qw	
	&\qw	& \qw 	& \qw	& \qw	&\qw	
	&\qw	& \qw 	& \qw	& \qw	&\qw	
	&\qw	& \qw 	& \qw	& \qw	&\qw		
	&\rstick{\bit{0}} \qw \\
\push{\bit{1}}&\push{\iota_1} 
	&	\qw	&\qw  	&\qw  	& \qw	&\qw	
	&\ctrl{3}& \qw 	& \qw	& \qw	&\qw
	&\qw	& \qw 	& \qw	& \qw	&\qw	
	&\qw	& \qw 	& \qw	& \qw	&\qw	
	&\qw	& \qw 	& \qw	& \qw	&\qw		
	&\rstick{\bit{0}} \qw\\
\push{\bit{1}}&\push{\iota_2} 
	&	\qw	&\qw  	&\qw  	& \qw	&\qw	
	&\qw	& \qw 	& \qw	& \qw	&\qw
	&\ctrl{2}& \qw 	& \qw	& \qw	&\qw
	&\qw	& \qw 	& \qw	& \qw	&\qw
	&\qw	& \qw 	& \qw	& \qw	&\qw		
	&\rstick{\bit{0}} \qw\\
\push{\bit{1}}&\push{\iota_3} 	
	& \qw	& \qw 	& \qw	& \qw	&\qw	
	&\qw	& \qw 	& \qw	& \qw	&\qw
	&\qw	& \qw 	& \qw	& \qw	&\qw
	&\ctrl{1}& \qw 	& \qw	& \qw	&\qw
	&\qw	& \qw 	& \qw	& \qw	&\qw	
	&\rstick{\bit{0}} \qw \\
\push{\bit{1}}&\push{\iota_4} 
	& \qw	& \qw 	& \qw	& \qw	&\qw	
	&\qw	& \qw 	& \qw	& \qw	&\qw
	&\qw	& \qw 	& \qw	& \qw	&\qw
	&\qw	& \qw 	& \qw	& \qw	&\qw	
	&\ctrl{1}& \qw 	& \qw	& \qw	&\qw
	&\rstick{\bit{0}} \qw \\	
\push{\bit{0}}&\push{\nu_0} 
	&\ctrld{-1}&\ctrl{0}&\ctrl{0}& \qw	
	&\qw	&\ctrld{-1}&\ctrl{0}&\ctrl{0}& \qw
	&\qw	&\ctrld{-1}&\ctrl{0}&\ctrl{0}& \qw
	&\qw	&\ctrld{-1}&\ctrl{0}&\ctrl{0}& \qw
	&\qw	&\ctrld{-1}&\ctrl{0}&\ctrl{0}& \qw
	&\qw	&\rstick{\bit{1}} \qw \\
\push{\bit{0}}&\push{\nu_1} 
	& \qw 	&\ctrld{-1}& \qw& \qw	&\qw
	& \qw 	&\ctrld{-1}& \qw& \qw	&\qw	
	& \qw 	&\ctrld{-1}& \qw& \qw	&\qw	
	& \qw 	&\ctrld{-1}& \qw& \qw	&\qw	
	& \qw 	&\ctrld{-1}& \qw& \qw	&\qw		
	&\rstick{\bit{1}} \qw \\
\push{\bit{0}}&\push{\nu_2} 
	& \qw 	&\qw 	&\ctrld{-2}&\ctrl{0}&\ctrl{0}
	& \qw 	&\qw 	&\ctrld{-2}&\ctrl{0}&\ctrl{0}
	& \qw 	&\qw 	&\ctrld{-2}&\ctrl{0}&\ctrl{0}
	& \qw 	&\qw 	&\ctrld{-2}&\ctrl{0}&\ctrl{0}
	& \qw 	&\qw 	&\ctrld{-2}&\ctrl{0}&\ctrl{0}
	&\rstick{\bit{1}} \qw \\
\push{\bit{0}}&\push{\nu_3} 
	&	\qw	&\qw  	&\qw  	&\ctrld{-1}&\qw
	&	\qw	&\qw  	&\qw  	&\ctrld{-1}&\qw  
	&	\qw	&\qw  	&\qw  	&\ctrld{-1}&\qw  
	&	\qw	&\qw  	&\qw  	&\ctrld{-1}&\qw  
	&	\qw	&\qw  	&\qw  	&\ctrld{-1}&\qw  	
	&\rstick{\bit{1}} \qw\\
\push{\bit{0}}&\push{\nu_4} 
	& \qw 	&\qw 	&\qw 	&\qw	&\ctrld{-2}
	& \qw 	&\qw 	&\qw 	&\qw	&\ctrld{-2}
	& \qw 	&\qw 	&\qw 	&\qw	&\ctrld{-2}
	& \qw 	&\qw 	&\qw 	&\qw	&\ctrld{-2}	
	& \qw 	&\qw 	&\qw 	&\qw	&\ctrld{-2}
	&\rstick{\bit{1}} \qw 
\gategroup{1}{3}{10}{7}{1em}{-}
\gategroup{1}{8}{10}{12}{1em}{-}
\gategroup{1}{13}{10}{17}{1em}{-}
\gategroup{1}{18}{10}{22}{1em}{-}
\gategroup{1}{23}{10}{27}{1em}{-}
}$}
\caption{A comparator circuit that solves \textsc{Reachability}. (The dummy gates are omitted.)}
\label{fig:NL2CC}
\end{figure}
\fi

The successive gadgets in the boxes each introduce a `pebble' (i.e. a
1 bit) which ends up at the next node in the depth-first search
(i.e. its wire will now carry 1) and is thus excluded from the search of
the gadgets that follow. For example, the gadget from the left-most dashed box in  \figref{NL2CC} will move a value 1 from wire $\iota_0$ to wire $\nu_0$ and from wire $\nu_0$ to wire $\nu_1$. This essentially ``marks'' the wire $\nu_1$ since we cannot move this value 1 away from $\nu_1$, and thus $\nu_{1}$ can no longer receive any new incoming 1. Hence, the gadget from the second box in \figref{NL2CC} will repeat the process of finding the lex-first maximal path  from $v_0$ to the remaining (unmarked) vertices. These searches end when all vertices reachable from $v_0$ are marked.
%Note that this has the same effect as applying the \
%emph{depth-first search} algorithm to find all the vertices reachable from $v_{0}$. Thus, the following theorem follows from the above construction.
\begin{theorem}[Feder \cite{Sub90}] \label{theo:NL2CCV}
$\NL\subseteq\CC$. 
\end{theorem}

% \nocite{*} 
\bibliographystyle{plain}
\bibliography{wphp}

\paragraph*{Acknowledgment} We would like to thank Yuli Ye for his contribution in the early parts of this research. A portion of this work was done when the second and third authors received funding from the [European Community's] Seventh Framework Programme [FP7/2007-2013] under grant agreement n$^{\rm o}$ 238381. This work was also supported by the Natural Sciences and Engineering Research Council of Canada.

\end{document}